\documentclass[opre]{informs4rad}

\TheoremsNumberedThrough 
\EquationsNumberedThrough
 \usepackage[T1]{fontenc}
 \usepackage[utf8]{inputenc}
 \usepackage{lmodern}

\usepackage{ifthen}
\usepackage{enumerate}
\usepackage{enumitem}
\usepackage{color}
\usepackage{pgfplots}
\usepgfplotslibrary{fillbetween}

\usepackage[algo2e,ruled,vlined,linesnumbered]{algorithm2e}
\SetNoFillComment

\usepackage{array,float}
\usepackage{url}

\usepackage{amsfonts, amsmath, amssymb, mathtools, dsfont, bbm, url, graphicx}
\usepackage{bm}
\usepackage{pgf,tikz,pgfplots}

\usepackage[breakable, skins]{tcolorbox}
\definecolor{myLightGray}{gray}{0.97} 

\DeclareRobustCommand{\mybox}[2][myLightGray]{%
\begin{tcolorbox}[
        left=0.5pt,
        right=0.5pt,
        top=0.5pt,
        bottom=0.5pt,
        colback=#1,
        colframe=#1,
        width=\dimexpr\textwidth\relax, 
        boxsep=3pt,
        arc=2pt,outer arc=2pt
        ]
        #2    
\end{tcolorbox}
}
\usepackage{csquotes}
\makeatletter
\renewenvironment*{displayquote}
  {\begingroup
   \setlength{\leftmargini}{0.25cm}%
   \linespread{1.5}\selectfont 
   \csq@getcargs{\csq@bdquote{}{}}}
  {\csq@edquote\endgroup}
\makeatother

\usetikzlibrary{shapes,arrows,patterns}

\definecolor{cornellred}{rgb}{0.7, 0.11, 0.11}
\definecolor{maroon}{rgb}{0.52, 0, 0}
\definecolor{dgreen}{rgb}{0.0, 0.5, 0.0}
\definecolor{ballblue}{rgb}{0.13, 0.67, 0.8}
\definecolor{royalblue(web)}{rgb}{0.25, 0.41, 0.88}
\definecolor{bleudefrance}{rgb}{0.19, 0.55, 0.91}
\definecolor{royalazure}{rgb}{0.0, 0.22, 0.66}
\usepackage[authoryear,round]{natbib}
\usepackage{accents}

\usepackage[hypertexnames=false]{hyperref}
\hypersetup{
	colorlinks = true,
	linkcolor=cornellred,
	citecolor=royalazure,
	linkbordercolor = {white}
}
\usepackage[capitalize]{cleveref}

\allowdisplaybreaks

\newtheorem{defn}{Definition}

\newcommand{\RomanNumeralCaps}[1]{\MakeUppercase{\romannumeral #1}}

\RequirePackage{tgtermes}
\RequirePackage{newtxtext}
\RequirePackage{newtxmath}
\RequirePackage{bm}
\RequirePackage{endnotes}

\usepackage{thmtools} 
\usepackage{thm-restate}
\OneAndAHalfSpacedXII
\usepackage{algorithm}
\usepackage{algpseudocode}
\usepackage{tikz}
\usetikzlibrary{decorations.markings,calc}
\tikzset{
  > = stealth,
  midarrow/.style={
    line width=1.2pt,        
    decoration={
      markings,
      mark=at position 0.5 with {\arrow[scale=1.5]{stealth}} 
    },
    postaction={decorate}
  }
}

\usepackage{natbib}
 \bibpunct[, ]{(}{)}{,}{a}{}{,}%
 \def\bibfont{\small}%
 \def\bibsep{\smallskipamount}%
\ECRepeatTheorems

\newenvironment{proofof}[1]{%
  \Trivlist
  \item[\hskip\labelsep {\it #1.}]\ignorespaces
}{\hfill \qed
\endTrivlist
\addvspace{0pt}
}

\newcommand{\probP}{\text{I\kern-0.15em P}}
\newcommand{\EX}{\mathbb{E}}

\newcommand{\sample}{z}

\newcommand{\indic}{\mathds{1}}
\newcommand{\crf}{\theta_f}
\newcommand{\yf}{y_{f}}

\newcommand{\propPar}{x}
\newcommand{\cf}{c_f}

\newcommand{\cD}{\mathcal{D}}
\newcommand{\cV}{\mathcal{V}}
\newcommand{\cP}{\mathcal{P}}
\newcommand\kk{{\cf}}
\newcommand\ZZ{{\mathbb Z}}
\newcommand{\commentcolor}{blue}

\renewcommand{\qed}{$\hfill\square$}
\newcommand{\compratio}{\alpha(f)}

\usepackage[showdeletions]{color-edits}
\newcommand{\revcolor}[1]{{\color{black}#1}}

\addauthor{RN}{blue} 

\OneAndAHalfSpacedXI

\begin{document}

\ARTICLEAUTHORS{%
\AUTHOR{Farbod Ekbatani}
\AFF{Booth School of Business,
University of Chicago, \EMAIL{fekbatan@chicagobooth.edu}}

\AUTHOR{Rad Niazadeh}
\AFF{Booth School of Business,
University of Chicago, \EMAIL{rad.niazadeh@chicagobooth.edu}}

\AUTHOR{Pranav Nuti}
\AFF{Booth School of Business,
University of Chicago, \EMAIL{pranav.nuti@chicagobooth.edu}}

\AUTHOR{Jan Vondr\'ak}
\AFF{Department of Mathematics,
Stanford University, \EMAIL{jvondrak@stanford.edu}}
}

\RUNAUTHOR{}
\RUNTITLE{Prophet Inequalities with Cancellation Costs}
\TITLE{Prophet Inequalities with Cancellation Costs}

\ABSTRACT{
Most of the literature on online algorithms and sequential decision-making in revenue management focuses on settings with irrevocable decisions, where once a decision is made upon the arrival of a new input, it cannot be canceled later. Motivated by modern revenue management applications---such as cloud spot markets, selling banner ads, or online hotel booking---we introduce and study ``prophet inequalities with cancellations'' under linear cancellation costs (known as the \emph{buyback model} in the literature). In the classic prophet inequality problem, a sequence of independent random variables $X_1, X_2, \ldots$ with known distributions is revealed one by one, and a decision-maker must decide when to stop and accept the current variable in order to maximize the expected value of their choice. In our model, after accepting $X_j$, one may later discard $X_j$ and accept another $X_i$ at a cost of $f \times X_j$, where $f\geq 0$ is a given parameter in this model. The goal is to maximize the expected net reward: the value of the final accepted variable minus the total cancellation cost. We aim to design online policies that are competitive against an omniscient optimal offline ``prophet'' benchmark.  

Our first main result is an optimal prophet inequality for \emph{all} parameters $f \ge 0$. We show this result by fully characterizing the worst-case competitive ratio of the optimal online policy against the optimal offline benchmark via the solution to a certain differential equation (for which we provide a constructive solution). Our second main result is to design and analyze a simple and polynomial-time randomized adaptive policy that achieves this optimal competitive ratio. Importantly, our policy is \emph{order-agnostic} (à la \citet{samuel1984comparison}), meaning that it only needs to know the set of distributions and not the order in which the random variables arrive. These results are obtained by a sequence of reductions to a (continuous) ``generalized flow problem,'' which can be viewed as the dual to a factor-revealing LP for characterizing the worst-case competitive ratio in the original problem. We develop several techniques to solve this reduced problem, including a geometric interpretation of its solution via a differential equation and the embedding of our random variables into specific Poisson point processes.  We also leverage these building blocks in novel ways, both for constructing our lower-bound instances and in designing (and analyzing) our order-agnostic policy.

}

\maketitle
\newpage
\section{Introduction}
\label{sec:intro}
Prophet inequalities have emerged as a powerful framework for studying a wide range of problems in revenue management that require real-time decision-making under uncertainty. Introduced in the context of optimal stopping~\citep{krengel1978semiamarts,samuel1984comparison}, this framework provides provable performance guarantees by comparing online policies (which have distributional knowledge of future arrivals) to an omniscient optimal offline benchmark known as the ``prophet.'' This perspective also unifies diverse applications, including sequential posted pricing~\citep{hill1982comparisons,gallego1994optimal,correa2017posted}, Bayesian online resource allocation~\citep{alaei2012online, ma2021dynamic,feng2022near}, and combinatorial auctions~\citep{chawla2010multi,feldman2013simultaneous,dutting2020prophet}, among others. By distilling complex stochastic environments into tractable models, prophet inequalities serve as a foundational tool for theoretical analysis and for designing practical algorithmic solutions in revenue management---solution that are often simple, interpretable, and approximately optimal.

A common feature of nearly all work on prophet inequalities, motivated by classical revenue management settings such as pricing a product for sale or assortment planning in retail, is that online policies make only \emph{irrevocable} decisions; once a decision is made (for example, accepting an arriving user’s service request), it is set in stone and cannot be changed later. Relaxing this assumption grants the online policy more flexibility: it can revisit past decisions and cancel some of them, thereby avoiding early commitments to low-valued requests. This paradigm is particularly relevant in certain modern applications such as cloud spot markets, where an assigned compute resource may be reclaimed and reassigned to improve supply efficiency, or in revenue management settings that permit overbooking (e.g., online hotel reservations or selling banner ads), where multiple assignments can initially be made and then selectively revoked later. Similarly, in volunteer matching platforms, the assignments of volunteers to tasks often change over time to improve match quality when better-suited volunteers arrive.

Although cancellations can help improve the performance of online policies, there is a trade-off due to the negative externalities---implicit or explicit---that they impose on users whose initial assignments are canceled. Therefore, regardless of the application, it is paramount to control the extent of cancellations an online policy makes. One approach to capture controlled cancellations, known in the literature as the ``buyback model''~\citep{babaioff2008selling,ashwinkumar2009randomized,ekbatani2023online}, is to maximize a mixed-sign objective function that incorporates \emph{costly cancellations}: the total reward of an online policy from serving (uncanceled) requests is penalized by the total reward of canceled requests, rescaled by a parameter $f \geq 0$. This parameter, often referred to as the \emph{buyback parameter}, can be viewed as the cost of canceling each unit of reward, reflecting how significantly cancellations penalize the online policy.  In particular, this model provides a smooth interpolation between the offline scenario, where $f=0$ (i.e., no penalty for cancellations),  and the fully online scenario with irrevocable decisions, where $f=+\infty$ (i.e., no cancellation is allowed).

\smallskip
\noindent\textbf{Problem statement \& model:} Motivated by the above, we initiate the study of the single-item \emph{prophet inequality problem with costly cancellations}. More formally, we consider a setting in which a decision-maker is presented with a sequence of random variables $X_1, X_2, \ldots$ one by one (think of them as rewards for fulfilling a request or willingness-to-pay of users for a service), where each $X_i$ is drawn independently from a known distribution $F_i$. Once $X_i$ is revealed in round $i$, the decision-maker can accept it or proceed to the next round. In the standard (vanilla) prophet inequality setting, the process ends as soon as the first variable is accepted. We diverge from this setting by allowing the decision-maker to accept $X_i$ even after previously accepting $X_j$ for some $j < i$. In this case, the previously accepted variable $X_j$ must be revoked (and its reward forfeited) by paying a cancellation cost of $f \cdot X_j$. The goal is to maximize the expected net reward, defined as the value of the final accepted variable minus all cancellation costs.\footnote{While it may appear that cancellations happen in an online manner, one could equivalently view this as an online allocation problem with overbooking, where only the request yielding the highest reward is ultimately allocated, and the decision-maker pays to cancel all other overbooked requests at the end of the decision-making horizon.}

From a technical perspective, the above is a natural model due to its invariance under scaling of rewards. From an applications perspective, the model is also quite natural: not only does it fit certain revenue management settings described above, but the cancellation cost can also be alternatively interpreted as a \emph{buyback cost}, i.e., a ``compensation fee'' in the form of a fixed percentage of the original transfer paid to the user to reclaim the resource---which is prevalent in practice, e.g., in airline booking
~\citep{USDOT2011}.

To establish prophet inequalities in this setting, we compare the expected net reward of an online policy with the prophet benchmark---the offline omniscient policy that observes all realizations in advance and selects the maximum, thereby achieving a reward of \(\EX[\max_i X_i]\). We measure performance via the \emph{competitive ratio}, defined as the worst-case ratio between these two expected rewards. Let \(\alpha(f)\) denote the optimal competitive ratio achievable by any online policy. Because each distribution \(F_i\) is known, one can compute an optimal online policy using polynomial-time dynamic programming.\footnote{The state of this dynamic program is essentially the pair of the current time and the value of the previously accepted variable, so it is polynomial-time solvable by using backward induction, once values are discretized properly.} Hence, \(\alpha(f)\) is the worst-case ratio of the expected net rewards of the optimal online and offline policies. 
In the special case of \(f = +\infty\), our model reduces to the classic prophet inequality problem, which admits a simple threshold-based \(0.5\)-competitive policy~\citep{samuel1984comparison}. 
Importantly, this policy is \emph{order-agnostic}~\citep{kleinberg2012matroid}, which means that it requires only knowledge of the set $\{F_i\}$, not the order in which random variables arrive. As \(f \to 0\), cancellations become free, and a greedy policy can exactly obtain the offline solution, implying \(\lim_{f \to 0}\alpha(f) = 1\).
Consequently, we expect \(\alpha(f)\in [0.5,1]\) for all $f\geq 0$.
We now ask the following research questions:

\smallskip
\begin{displayquote}
\begin{enumerate}[label=(\emph{\roman*}),leftmargin=0.22in]
\item  \emph{Can we characterize the optimal competitive ratio $\alpha(f)$ as a function of the buyback parameter $f$, for all possible choices of $f\geq 0$?}
\item \emph{Can we design and analyze a simple, polynomial-time, and order-agnostic policy, à la \citet{samuel1984comparison}, that achieves the optimal competitive ratio $\alpha(f)$ for all $f\geq 0$?}
\end{enumerate} 
\end{displayquote}

\smallskip
\noindent\textbf{Our main contributions:} In summary, we answer both of the above questions in the affirmative. We not only characterize the extra gain from cancellations in the buyback model, but also show that there is a simple policy that achieves this gain, with all the desired properties described earlier.

\smallskip
\textbf{(i)~Optimal competitive ratio in all parameter regimes.} For the first question, we identify a particular nonlinear differential equation, parameterized by $f$, along with certain boundary conditions (see \Cref{defn:differentialequation}). We then show that it always admits a proper solution $y_f: \bigl[\tfrac{f}{f+1}, 1\bigr] \to [0, 1]$ by constructing one for any $f\geq 0$. Using this solution, we establish our first main result, which is a \emph{full characterization} of the optimal competitive ratio $\compratio$ for \emph{all} buyback parameters $f\geq 0$:

\mybox{
\begin{displayquote}
{\textbf{{$\large{\boldsymbol{[}}$Main Result \RomanNumeralCaps{1}$\large{\boldsymbol{]}}$}}} For any buyback parameter $f \geq 0$, the optimal online policy achieves a competitive ratio $\alpha(f) = \tfrac{1}{2 - y_f(1)}$ against the optimal offline benchmark. Moreover, for any $f \geq 0$, there exists an instance with finitely many two-point random variables for which no online algorithm can attain a competitive ratio greater than $\tfrac{1}{2 - y_f(1)}$.
\end{displayquote}
}
See \Cref{fig:competitive-ratios} for a comparison of this bound with the literature. For small values of $f$, say $f=0.2$ (with $\alpha(0.2)\approx 0.82$) , there is a substantial improvement in the competitive ratio versus $0.5$. We note that even for moderately large values of $f$, say $f=1$ (with $\alpha(1)=\frac{2}{3}$) or $f=2$ (with $\alpha(2)=\frac{3}{5}$), there is still a considerable improvement, showing the value of introducing costly cancellations in the prophet inequality setting. We can also compare our bound with the best bounds in the adversarial setting, where prior stochastic information is not available. In this setting, \citet{babaioff2008selling} established an optimal competitive ratio of $\tfrac{1}{1 + 2f + 2\sqrt{f(1+f)}}$ for deterministic algorithms, and the follow-up work of \citet{ashwinkumar2009randomized} showed an optimal competitive ratio of $\tfrac{1}{W_{-1}\bigl(-\tfrac{1}{e(1+f)}\bigr)}$, where $W_{-1}$ is the non-principal branch of the Lambert $W$ function.\footnote{The Lambert $W$ function is the inverse of $y(x) = x e^x$. The non-principal branch is the solution  $x = W_{-1}(y)$ when $x \le -1$.} We can clearly observe a considerable gap between their competitive ratios and ours, even 
as $f \to 0$,\footnote{For $f=+\infty$, our bound converges to $0.5$ while adversarial bounds converge to $0$. It can also be shown that our bound has an asymptotic behavior of $1 - \Theta(f \log \tfrac{1}{f})$ for $f \to 0$ (\Cref{sec:smallbuybackregime}), while both adversarial bounds are asymptotically $1 - \Theta(\sqrt{f})$, thus converging to 1 slower than our bound.} which can be interpreted as the ``value of information'' about the underlying distributions in our setting.
\input{figures/yf}

We also highlight that for the special case when $f \geq 1$, we derive a closed-form solution for the differential equation, resulting in a competitive ratio of $\compratio = \frac{1+f}{1+2f}$ in this regime. The instance that demonstrates the tightness of this bound for $f \ge 1$ is quite intuitive and serves as a natural generalization of the classic ``bad example'' in prophet inequalities: consider two random variables $X_1 = 1$ and $X_2 = (1+f)\textrm{Ber}(\tfrac{1}{1+f})$. It is straightforward to verify that any online policy achieves a net reward of at most 1 in expectation. In fact, if the policy selects $X_1 = 1$, it is indifferent between ignoring $X_2$ and performing a buyback, obtaining $1$ in expectation either way. If it ignores $X_1$, it still obtains at most $1$, since $\EX[X_2] = 1$. Meanwhile, $\EX[\max\{X_1, X_2\}] = \tfrac{1}{f+1}(f+1) + \tfrac{f}{f+1} = \tfrac{2f + 1}{f + 1}$.

Interestingly, this example is \emph{not} the worst-case instance for $f < 1$. For example, when $\tfrac{1}{3} \le f < 1$, one can construct an instance with $n = 3$ variables for which no online algorithm can attain a competitive ratio better than $\tfrac{(1+f)\bigl(\sqrt{f(2-f)} + 1\bigr)}{(1+f)\sqrt{f(2-f)} + 3f + 1} < \tfrac{1+f}{1+2f}$ (see \Cref{sec:three-variable-example} for the example). This suggests that the problem poses new algebraic challenges for smaller values of $f$, as the worst-case instance appears to require more variables when $f$ decreases. By devising another (more intricate) closed-form solution for our differential equation when $\tfrac{1}{3} \le f < 1$, we prove that this bound is indeed optimal for that regime. However, this investigation serves as evidence that it may be infeasible to obtain closed-form solutions for the entire interval $0 \leq f < 1$. Although we do not provide an explicit solution to the differential equation (and therefore for $\compratio$) for $0 \leq f < \tfrac{1}{3}$, we do provide the means to construct a proper solution and identify its properties. We then apply this solution not only to analyze the competitive ratio of the optimal online policy (\Cref{sec:upper-bound}), but also to construct a tight worst-case instance (\Cref{sec:lower-bound}). Our results imply that, indeed, the worst-case instance involves more random variables, and the optimal online policy performs more buyback operations, as $f$ becomes smaller.


\smallskip
\textbf{(ii) Order-agnostic \& simple optimal-competitive policy.} To address the second question, we begin by examining the structural properties of the (order-aware) optimal online policy under the worst-case instance identified in our first main result. Building on the insights and technical building blocks that helped us prove our first main result, we arrive at our second main result.

 \mybox{
\begin{displayquote}
\textbf{{$\large{\boldsymbol{[}}$Main Result \RomanNumeralCaps{2}$\large\boldsymbol{]}$}} 
For any buyback parameter $f \ge 0$, there is a polynomial-time randomized adaptive policy (\Cref{alg:order-oblivious}) that (i) achieves the optimal competitive ratio $\compratio = \frac{1}{2 - y_f(1)}$, and (ii) is order-agnostic, relying only on the set of distributions $\{F_i\}$ rather than their arrival order---in contrast to the \revcolor{instance-}optimal online policy, which requires the full sequence of distributions. 
\end{displayquote}
}
Our order-agnostic~\footnote{\revcolor{We make an implicit assumption that the index of an arriving random variable $X_i$ is observable by our online policy. Therefore, the algorithm knowing the set of distributions is able to identify the set of remaining distributions $\{F_i\}_{i>t}$ at each time $t$; as will be clear later in the paper, our order-agnostic policy uses this information.}} online policy is a simple algorithm with randomized adaptive thresholds. At a high level, the algorithm is greedy in nature, but has a safeguard: it starts with a randomized threshold to make the first acceptance. After each acceptance, it essentially searches for the ``right next quantile'' of the distribution of $\max_i X_i$, which if accepted and swapped with the currently accepted random variable, generates a sufficiently large positive marginal gain. This marginal gain would offset the forfeited reward, the current cancellation cost, and any potential future cancellation costs that might arise as a consequence of this decision. Importantly, this decision to identify the next threshold (or quantile) heavily uses the solution $y_f$ to the differential equation and the last accepted quantile, and implicitly ``mimics'' the behavior of the optimal online policy in the worst-case instance (using randomization in a certain way). We explain these connections in detail later (\Cref{sec:algorithm}). 

Finally, in \Cref{sec:numerical}, we evaluate the numerical performance of our algorithm on instances beyond
the worst case, comparing it against existing benchmarks for the classic prophet inequality problem and for
the adversarial version of the problem, as well as an alternative threshold-based greedy algorithm that we introduce and analyze in \Cref{sec:thresholdalgorithm}. We observe that in almost all of our simulations, our algorithm either outperforms the other benchmarks or is among the best in terms of empirical performance, while it is the only one that enjoys the above theoretical performance guarantee.

\medskip
\noindent\textbf{Overview of our techniques:}
To establish our main results, we propose a rich technical framework that involves multiple steps. We sketch the key ideas here and defer the full details to later sections.

\smallskip
 \emph{Flow-based LP formulation.} At a high level, our results are derived from a linear programming (LP) approach. Although in principle one could formulate a nonlinear program to characterize the worst-case performance ratio of the optimal online policy against the prophet benchmark, analyzing such a program can be intractable due to its exponential size. Inspired by LP-based analyses for the classic prophet inequality problem and its extensions~\citep{alaei2014bayesian,niazadeh2018prophet,feng2022near,jiang2022tight}, we circumvent this issue by proposing an alternative \emph{polynomial-size parametric LP} (and its dual) that captures the worst-case competitive ratio of the optimal online policy in our generalization of the prophet inequality problem for any given $f>0$.

\revcolor{To set up this parametric LP, we begin by proving that for every instance of our problem, there exists a harder instance consisting of two-point random variables
$X_i=v_i\mathrm{Ber}(q_i)$, where $\mathrm{Ber}(q_i)$ is a $0/1$ Bernoulli random variable equal to $1$ with probability $q_i$ and to $0$ otherwise, $v_1\leq v_2\leq\cdots\leq v_n$, and $q_i\in[0,1]$.} This reduction to a discrete, monotonic instance is intuitive because it ensures that information about the maximum value is revealed as slowly as possible, thereby making the problem more challenging for an online decision-maker who observes these variables sequentially. Formally, we establish this reduction by ``splitting'' the original random variables into two-point distributions without altering the prophet benchmark---in the same spirit as similar variable-splitting techniques, e.g., \citet{liu2021variable}, in other prophet-inequality-type problems. We then show that any policy for the split instance can be simulated in the original instance without loss.

Once we have a monotonic instance, we set up a factor-revealing program to analyze the worst-case competitive ratio. However, it turns out that the constraints are nonlinear in the probabilities $q_i$ of the two-point variables, so these probabilities must remain external parameters to obtain our (parametric) LP. Applying LP duality yields another LP, which we informally call ``online contention resolution with recourse.'' This is somewhat analogous to vanilla contention resolution LPs~\citep{alaei2014bayesian,jiang2022tight},  where each variable represents the probability of allocating in a certain state, and we aim to achieve a certain level of \emph{allocation coverage} at each time.

However, our program includes new features: each state $(i,j)$ with $i < j$ corresponds to a directed edge from node $i$ to $j$, indicating that at time $j$, the last accepted variable was at time $i$. The LP then takes the form of a \emph{generalized flow problem} on this digraph, with ``leakage'' along each edge governed by the buyback parameter $f$ (see \Cref{fig:flow}). There are also capacities on the edges to ensure that the final policy is implementable online, and certain demands should be maintained at each node $i$ to account for the allocation coverage at each time in the final policy. 

\smallskip
\emph{Continuous LP, greedy-type flow, \& differential equation.} The remaining task is to find an explicit solution to this generalized flow problem for every choice of the probabilities $q_i$. It turns out that in the worst case, it is safe to assume that all the probabilities $q_i$ are small. This allows us to reduce the problem of finding a solution to the LP for potentially many different values of $q_i$ to solving a single continuous generalized flow problem. We then aim to exploit the special combinatorial structure of this problem to derive an explicit feasible solution that characterizes our competitive ratio.

Informally speaking, any feasible solution to this continuous LP can be equivalently interpreted as an implementable online algorithm in a surrogate ``quantile acceptance problem,'' where a decision-maker sequentially observes a continuum of Bernoulli random variables in the interval $(0,1]$, each active with a probability of $\frac{dt}{t}$. At each step $t \in (0,1]$, if the current quantile $t$ is active, the decision-maker should decide whether to select $t$ for the first time (if no quantile has been selected before), or  whether to swap a previously accepted quantile $s < t$ with $t$. The goal is to satisfy the coverage condition of the LP for every $t \in (0,1]$, where coverage is defined as the total probability of selecting $t$ minus the ``leakage'' due to the buyback feature of our problem, which is $(1+f)$ times the total probability of swapping $t$ with a future active quantile (see \Cref{sec:reduction-to-diff-eq}). We may interpret this as requiring that each quantile $t$ be accepted with a sufficiently large \emph{effective probability}.

Viewing these quantiles as small random variables in a monotonic instance, intuitively speaking, an effective algorithm should strive to pick the last active quantile (the analog of the prophet benchmark) with as few swaps as possible, or somewhat similarly, maximize the effective probability that each quantile is selected.\footnote{Quite surprisingly, under this dual continuous LP perspective, we only need to consider effective probability of selecting each quantile, rather than their values in the monotonic instance---mirroring how online contention resolution schemes remain agnostic to values when employed in designing prophet inequalities.}  Consequently, one can guess that a candidate optimal strategy would follow a greedy-type rule but with some safeguards: the algorithm selects an active quantile as its first choice and then swaps a previously accepted quantile $s$ with the current active quantile $t$ only if $t > \tau(s)$, for some function $\tau: [0,1] \rightarrow [0,1]$ (likely satisfying $\tau(s) > s$).

It turns out that restricting to this class of solutions for the quantile acceptance problem---or equivalently our continuous LP---is without loss. Formally speaking, we characterize a particular differential equation for the function $\tau$ (or its inverse function $y\triangleq \tau^{-1}$) in a somewhat unconventional form (\Cref{defn:differentialequation}), so that the solution $y_f$ to this differential equation produces a feasible greedy-type solution for our continuous LP with coverage $\Theta = \tfrac{1}{2 - y_f(1)}$. Our differential equation is obtained by first identifying structural properties of the greedy-type algorithm described above, then encoding those properties into a single differential equation for the $y$-function used by the algorithm.

We conclude our analysis by showing that a solution to this differential equation (satisfying certain regularity conditions) indeed exists and describing how to build such a solution algorithmically. Our construction is involved and employs various techniques---such as backward-in-time recursion and the identification of ``correct'' recursive invariants---that may be of independent interest (see \Cref{sec:diff-equation-sketch} and \Cref{sec:diff-eq-soln} for details; see also \Cref{fig:diff-equ} for a visualization of our solution).

We note that our algorithm for solving the ``quantile acceptance problem,''  previously described as providing a feasible solution to our continuous LP (see \Cref{alg:quantile-selection}), is also more directly linked to our original prophet inequality problem with cancellation costs under the buyback model. Indeed, every time we observe a random variable $X_i$ in the original problem, we implicitly observe a quantile of $X_{\max}$ (by translating the observed value to its corresponding quantile in $X_{\max}$'s distribution), and we must decide whether to accept it. In the prophet inequality with buyback problem, we may swap $X_i$ with $X_j$, analogous to how we may swap quantiles in the quantile acceptance problem. Linearity of expectation shows that by accepting each quantile of $X_{\max}$ with a certain guarantee on the effective probability, we obtain the same competitive ratio for the original prophet inequality problem under the buyback model.

However, there remain two major differences between the quantile acceptance problem and the prophet inequality problem under buyback model. Firstly, in the quantile acceptance problem, the activation of each quantile of $X_{\max}$ is observed in increasing order of quantiles; however, our observations of $X_i$ do not need to be in increasing order. Secondly, the probability with which a quantile $t$ of $X_{\max}$ is observed in the sequence $\{X_i\}$ need not be $\frac{dt}{t}$.\footnote{Note that if we have a monotonic instance where each random variable has a sufficiently small chance of being active, then these differences do not exist; therefore, the solution to the quantile acceptance problem can be used directly to create an algorithm for the original problem in such an instance.} How can we overcome these differences and use our solution to the quantile acceptance problem to design a simple and order-agnostic algorithm for a general instance of the prophet inequality with buyback problem? The answer to this question is precisely the next technical result of our paper, which gives rise to an order-agnostic and optimal competitive online algorithm. 


\smallskip
\emph{Simple and order-agnostic optimal-competitive algorithm.} To answer the above question, we formalize a continuum of Bernoulli random variables on $(0, 1]$ in which each $t$ is active with a probability of $\frac{dt}{t}$ as a \emph{Poisson point process} (PPP) with rate function $\frac{1}{t}$. Our goal is to embed each observed $X_i$ (implicitly a quantile of $X_{\max}$) into this PPP via a simple coupling, and then apply the quantile acceptance algorithm on the resulting coupled process to obtain an algorithm for the original problem.

Our algorithm generates this PPP in an online fashion. As it observes $X_i$, it tracks the current maximum observed value in the sequence so far, denoted by $\hat{x}$. Whenever a new $X_i$ exceeds $\hat{x}$, to couple $X_i$ with the PPP, the algorithm simply inserts the quantile of $X_{\max}$ corresponding to $X_i$ into the PPP. It then ``fills the gap'' between $\hat{x}$ and $X_i$ by artificially inserting extra points between these two quantities drawn from a PPP with rate exactly $\frac{1}{t}$. Notice that the algorithm ignores observations that are not peaks, thereby resolving the issue that the observed quantiles may not arrive in increasing order. After generating the PPP, our algorithm flags the same quantiles that are accepted in the quantile acceptance problem. At this point, if we could treat the flagged values as if they were the values the algorithm accepted, we would be done.

 However, our algorithm can only accept the values $X_i$, not all the values in the PPP (since the rest are artificially generated). Fortunately, by the linearity of expectation, accepting $X_i$ with a carefully chosen probability $p$ can simulate accepting a (flagged) value $z < X_i$, in the sense that the expected net gain of this simulated allocation equals that of selecting the factitious (i.e., non-existing) flagged value. This allows us to bound the competitive ratio of the algorithm by the performance of our quantile acceptance algorithm, thereby establishing that it is optimal-competitive. Ultimately, the resulting algorithm is simple: observe the sequence of peaks in the realizations of $X_i$ and  accept each peak with a carefully chosen probability. These probabilities, defined using the PPP and our solution to the quantile acceptance problem, are easy to calculate and independent of the order of future arrivals. See \Cref{alg:order-oblivious} for more details. 


\smallskip
\emph{Indifference condition \& bad examples.} Somewhat surprisingly, the same solution to the differential equation that was foundational in our algorithm design also allows us to construct a lower bound instance with two-point distributions for our original problem, which \emph{exactly} matches our upper bound of $\theta_f = \frac{1}{2-y_f(1)}$. The key idea is to identify (and formalize) a common property that we observe in the bad examples for $n=2$ and $n=3$: for these worst-case instances, the optimal competitive online algorithm, when holding a random variable, is indifferent between accepting or rejecting the next one. Restricting ourselves to classes of two-point ``instances with indifference'' (see \Cref{sec:indifference} for a formal definition), we use our differential equation solution $y_f$ in a novel way to construct activation probabilities in our instance. In particular, defining the sequence $\{k_i\}$, the ``orbits of 0 under $y_f$,'' that is, $k_0=0$ and $y_f(k_{i+1})=k_i$, we let the $i^\textrm{th}$ random variable be the maximum with probability $k_i - k_{i-1}$. Together with the corresponding values calculated from the indifference conditions, this instance established our desired bound.

\subsection{Further Related Work}
\label{sec:further-related}
Our work is connected to various lines of research in operations research and computer science. See below for a summary of the related literature and discussions. 

\paragraph{Buyback and recourse in online allocations.} The study of the single-item buyback setting under adversarial arrivals originates from the work of \citet{babaioff2008selling}, which established the optimal competitive ratio for deterministic integral algorithms using a modified greedy approach. The follow-up work of \citet{ashwinkumar2009randomized} showed that the optimal competitive ratio can be achieved via an an elegant correlated randomization step combined with the greedy algorithm. More recently, \citet{ekbatani2023online} developed an optimal primal-dual fractional algorithm that, together with a lossless randomized rounding, also yields an optimal competitive randomized algorithm. Beyond the single-item case, the buyback setting has been studied for general matroids~\citep{babaioff2008selling,ashwinkumar2009randomized}, online matchings~\citep{ekbatani2023online}, and matroid intersections~\citep{badanidiyuru2011buyback}. It can also be viewed as a special case of unconstrained online allocation with a submodular (though non-monotone) combinatorial valuation \citep{rubinstein2017combinatorial}, where only an $\mathcal{O}(1)$-competitive algorithm is currently known. The special case of the buyback model with $f=0$ is also studied in the literature on online resource allocation with free-disposal~\citep{feldman2009online,devanur2016whole,feng2024batching}. Another related problem is prophet inequalities with overbooking~\citep{ezra2018prophets,assaf2000simple}, which differ by allowing the decision-maker to accept more variables than their capacity, whereas here we consider a single-item capacity and a linear cancellation cost.  Although ``recourse'' has been studied in various other contexts---e.g., \cite{azar2003combining,buchbinder2014online,canetti1995bounding,han2014online,chaudhuri2009online,constantin2009online}---our focus is specifically on the classic prophet inequality with buyback.


\paragraph{Other related models in dynamic revenue management.}
Our work is also conceptually related to dynamic revenue management, where offline supplies have soft inventory constraints. Examples include dynamic revenue management with cancellation (on the demand side) and overbooking~\citep[e.g.,][]{rot-71,ET-10,ABFN-13,DKX-19,FZ-21}, as well as dynamic inventory control problems~\citep[e.g.,][]{LRS-07,mos-10,BM-13,CS-19,AJ-22,QSW-22}. Another related line of research is on revenue management with \emph{callable products}, first studied by \citet{GKP-08} in the context of airlines. They proposed a model allowing low-fare tickets to be recalled (bought back) at a certain price once high-fare customers arrive. \citet{GL-18} extended this model to multi-fare consumers, while \citet{GL-20} examined online assortment with callable products and characterized optimal policies for different choice models.

\paragraph{Prophet inequalities.} Besides the works described, other related noteworthy results in the prophet inequality literature include refined bounds for special cases such as i.i.d.\ values or independent arrivals in random order~\citep{correa2017posted,abolhassani2017beating,esfandiari2017prophet}, as well as extensions to more general feasibility environments and valuations. These include multi-unit settings~\citep{hajiaghayi2007automated,yan2011mechanism,alaei2014bayesian,chawla2024static}, matroids~\citep{kleinberg2012matroid,anari2019nearly}, matchings~\citep{alaei2012online,papadimitriou2021online}, general downward-closed environments~\citep{rubinstein2016beyond}, and combinatorial valuations~\citep{rubinstein2017combinatorial}. For a more comprehensive overview, see \citet{lucier2017economic} and \citet{correa2019recent}.

\paragraph{Online contention resolution and LP duality.} As noted, our generalized flow perspective closely relates to a generalization of the online contention resolution scheme (OCRS) problem under costly recourse. The OCRS problem traces back to the seminal works of \citet{feige2006maximizing} and \citet{vondrak2011submodular} in the offline setting, and \citet{feldman2016online} in the online setting. Since then, OCRS has been extensively studied within both computer science and operations research. For instance, \citet{alaei2014bayesian} developed an (almost) optimal OCRS for uniform matroids, later refined and extended by \citet{jiang2022tight,dinev2023simple,feng2022near}. More recently, \cite{aminian2026stationary} has developed a new framework relying on the notion of ``stationarity'' of allocation distributions, which provides new optimal and near-optimal OCRSs for a broad range of environments. Moreover, OCRS-based techniques have found diverse applications in prophet inequalities and Bayesian mechanism design. Classic examples include sequential posted pricing~\citep{yan2011mechanism,hartline2009simple,chawla2010multi}, combinatorial auctions~\citep{dutting2020prophet,feige2015unifying}, correlations in prophet inequalities~\citep{immorlica2020prophet}, matchings~\citep{ezra2022prophet,pollner2022improved}, and prophet secretary~\citep{esfandiari2017prophet,adamczyk2018random}.

\revcolor{The idea of using ``factor-revealing LPs''  has a long history in designing online algorithms (or even approximation algorithms) against worst-case instances; in particular, it has found several interesting applications in the line of work on prophet inequalities and designing online contention resolution schemes~\citep{alaei2012online,alaei2014bayesian,jiang2022tight,jiang2022tightness,feng2022near,perez2026iid}, and even in the secretary (or prophet secretary) models~\citep{buchbinder2014secretary,chan2014revealing,correa2024sample} Also, this technique is highly related to the so-called ``equivalence of LP and DP'' studied in classical work on approximate dynamic programming, e.g., see \cite{adelman2004price}. Intimately related to this is also the work of \citet{lee2018optimal}, which has established a duality between prophet inequalities (with respect to ex-ante relaxation) and OCRS. This LP duality is both conceptually and mathematically linked to our LP-based approach and other LP approaches in the literature on prophet inequalities, e.g., LP-based techniques for designing competitive online algorithms against fluid approximations in various Bayesian online allocation problems~\citep{feng2022near,ma2021dynamic,aminian2026optimal}.}

\paragraph{Generalized flow with leakage} Generalized flows have been studied from the perspective of network flow~\citep{bhaumik1974optimum,glover1973equivalence,glover1978generalized} and a generalized max-min relation has been found~\citep{pulat1989relation}. We do not appeal to this result directly, as it would merely lead us back to our primal LP.

\paragraph{Optimum online policies.} A significant portion of the prophet inequality literature focuses on simple, threshold-based algorithms and their competitive ratios against the prophet benchmark. In contrast, recent work has examined the optimal online policy from a computational perspective: \citet{anari2019nearly} established a PTAS for laminar matroids, while \citet{segev2021efficient} developed a PTAS for free-order prophet inequalities and stochastic probing, among others. Nevertheless, finding the optimal online policy can be PSPACE-hard in certain combinatorial settings (e.g., matchings~\citep{papadimitriou2021online}). Additional research explores how knowledge of the order impacts performance relative to the optimal online policy~\citep{niazadeh2018prophet,ezra2022significance,feng2021two,braverman2025new,braun2024approximating}. In our setting, the state space for the dynamic program is polynomial in size, allowing for an exact solution via backward induction.

\section{Characterizing the Optimal Competitive Ratio}
\label{sec:upper-bound}
In general, one could formulate a nonlinear, exponential-size mathematical program to characterize the worst-case performance of the optimal online policy against the prophet benchmark---mirroring the approach of finding the worst-case competitive ratio through the equilibrium of a zero-sum game between an algorithm (choosing an online policy) and an adversary (choosing the sequence of distributions). In this section, rather than pursuing this na\"ive approach, we develop a parametric linear program whose parameters encode a lower-dimensional ``sufficient statistic'' of the worst-case instance. We then show how this LP characterizes the worst-case ratio between the optimal online policy and the offline optimum, under the adversary's worst-case choices of these parameters.

To build this parametric factor-revealing LP, we begin in \Cref{sec:monotone reduction} by introducing a series of reductions that convert any instance to a restricted form, without changing the worst-case competitive ratio. Next, in \Cref{sec:LP}, we present a discrete version of our LP and its dual, along with a continuous variant of the dual (which corresponds to taking a particular limit of the discrete dual). Finally, inspired by treating our continuous LP as a \emph{generalized flow problem}~\citep{pulat1989relation}, we show how to effectively solve the continuous LP (or identify a ``good'' feasible solution) via a specific differential equation in \Cref{sec:reduction-to-diff-eq}. We briefly sketch how to construct a solution for the differential equation in \Cref{sec:diff-equation-sketch} and defer the full details to \Cref{sec:diff-eq-soln}.

\subsection{Reduction to a monotonic sequence of two-point distributions}
\label{sec:monotone reduction}
We claim that, without loss of generality, we can restrict our attention to problem instances comprising a sequence of random variables $X_i$ such that (i)~they have \emph{scaled Bernoulli} distributions, that is, $X_i \in \{0, v_i\}$, and (ii)~they are \emph{monotonic}, that is, $v_1 \le v_2 \le \ldots \le v_n$ for some $n\in\mathbb{N}$. More precisely, we show the following result, which paves the way for the rest of our analysis.


\begin{proposition}
\label{thm:reduction}
If a competitive ratio $\alpha(f)$ can be achieved for any instance with random variables $X_i = v_i \cdot \mathrm{Ber}(q_i)$, where $\mathrm{Ber}(q_i)$ is a $0/1$-Bernoulli random variable with expectation $q_i$ and $0 < v_1 \leq v_2\leq \ldots \leq v_n$, then the same competitive ratio $\alpha(f)$ can be achieved for any instance with nonnegative random variables with finite expectation.
\end{proposition}

To prove \Cref{thm:reduction}, we proceed in a sequence of simple reductions. The first step is to discretize the random variables. See \Cref{apx:lem:discrete} for the proof of the following lemma. 

\begin{lemma}
\label{lem:discretize}
If a competitive ratio $\alpha(f)$ can be achieved for any instance with discrete nonnegative random variables (with finitely many points in the support), then the same competitive ratio $\alpha(f)$ can be achieved for any instance with nonnegative random variables with finite expectations.
\end{lemma}

The second step is to reduce the random variables by ``splitting'' into scaled Bernoulli distributions.

\begin{lemma}
\label{lem:splitting}
If a competitive ratio $\alpha(f)$ can be achieved for any instance with scaled Bernoulli random variables (with values in $\{0,v_i\}$ for some $v_i > 0$), then the same competitive ratio $\alpha(f)$ can be achieved for any instance with discrete nonnegative random variables.
\end{lemma}

\begin{proofof}{Proof}
Assume that $\cD_i$ is supported on a finite set of values $\cV_i = \{0, v^1_i, \ldots, v^k_i\}$ with $0<v^1_i < \cdots < v^k_i$. We define $k$ new random variables $X^1_i,\ldots,X^k_i$,
where $X^j_i$ has support $\{0, v^j_i\}$ and the probabilities are chosen so that $\max_{1 \leq j \leq k} X^j_i$ has the same distribution as $X_i$ (in particular, we set $\Pr[X^j_i = v^j_i] \;=\; \frac{\Pr[X_i = v^j_i]}{\Pr[X_i \le v^j_i]}$). We denote the distribution of $X^j_i$ by $\cD^j_i$.


Given the original instance $I_1\;=\;(\cD_1,\ldots,\cD_n)$ with $n$ random variables, consider a new instance $I_2$ with $n+k-1$ random variables, formed by ``splitting'' $X_i$ into $k$ scaled Bernoulli random variables $\{X_i^j\}$ as described above. Specifically, let $I_2 \;=\; \bigl(\cD_1,\cD_2,\ldots,\cD_{i-1}, \cD^1_i,\cD^2_i,\ldots,\cD^k_i,\cD_{i+1},\ldots,\cD_n\bigr)$.
We couple the two instances so the instantiations of all the random variables are the same, except for $X_i$ in $I_1$ which equals $\max_{1 \leq j \leq k} X^j_i$ in $I_2$ (and the instantiations of $X^j_i$ are otherwise independent).


Let us define a policy $\cP_1$ for the instance $I_1$ based on a policy $\cP_2$ for the instance $I_2$. First, $\cP_1$ mimics $\cP_2$ up to the arrival of $X_{i}$. Suppose both policies hold value $x$ after processing $X_{i-1}$. Upon observing $X_{i} = v_i^r$, $\cP_1$ \emph{simulates} $\cP_2$ via the above coupling. Specifically, $\cP_1$ generates samples $X^j_i$ from $\cD^j_i$ for $j < r$, and observes the last acceptance of $\cP_2$ on $X^1_i, X^2_i,\ldots, X^{r-1}_i, X_i = v_i^r, X^{r+1}_i = 0, \ldots, X^k_i = 0$. If this last accepted value is $v$ (one of the $v_i^j$ with $j \le r$, or $x$ if $\cP_2$ does not accept any $v_i^j$), then $\cP_1$ accepts $X_{i}$ with probability $\frac{v - x}{v^r_i - x}$. Afterward, $\cP_1$ continues to follow $\cP_2$, \emph{pretending} to hold $v$.

We claim that $\cP_1$ performs at least as well as $\cP_2$. To see why, recall we couple the instantiation of random variables in $I_1$ and $I_2$ as described earlier. Then, upon observing $\{X_i^j\}_{j=1}^k$ and $X_i$ respectively, $\cP_1$ and $\cP_2$ hold the same value in expectation after processing these random variables, because:
\[
x \left(1 - \frac{v - x}{v^r_i - x}\right) + v^r_i \left(\frac{v - x}{v^r_i - x}\right)=
\frac{x(v^r_i - v) + v^r_i(v - x)}{v^r_i - x}=
v.
\]
Since the expected net reward (including the buyback cost) is linear in the accepted values, by linearity of expectation, the future expected net reward of $\cP_1$ after observing $X_i$ equals that of $\cP_2$. $\cP_1$ can only outperform $\cP_2$ because $\cP_2$ might accept multiple $X_i^j$ and therefore pay extra buyback costs.
\end{proofof}

The final step of our reduction is to reorder the scaled Bernoulli distributions so that the values $v_i$ appear in increasing order. We show that this reordering can only make the instance more difficult.

\begin{lemma}
\label{lem:monotonize}
If a competitive factor $\alpha(f)$ can be achieved for any instance with monotonic scaled Bernoulli random variables, with values in $\{0,v_i\}$ for some $v_i > 0$ 
and $v_1 \leq \ldots \leq v_n$, then the same competitive factor $\alpha(f)$ can be achieved for any instance with scaled Bernoulli random variables.
\end{lemma}

\begin{proofof}{Proof sketch} 
Consider an instance $I_2$ with variables $X_i = v_i \,\mathrm{Ber}(q_i)$ and $X_{i+1} = v_{i+1}\,\mathrm{Ber}(q_{i+1})$, where $v_i < v_{i+1}$. Let $I_1$ be the same instance, except that $X_i$ and $X_{i+1}$ are swapped in order. In $I_1$, we might as well assume that at time $i$ we observe only $\max(X_{i+1}, X_i)$ (and nothing else), which gives us strictly less information than observing $X_{i+1}$ followed by $X_i$. Yet, as shown in the proof of \cref{lem:splitting}, there exists a policy for $I_1$ (observing $\max(X_{i+1},X_i)$ at time $i$) that outperforms a given policy on $I_2$ (observing $X_i$ and then $X_{i+1}$, where $v_i < v_{i+1}$). Since any permutation can be written as a composition of repeated swaps, this completes the proof. We also include a complete ``from scratch'' proof of \cref{lem:monotonize}, not dependent on \cref{lem:splitting} in \cref{sec:lem:monotonize}.
\end{proofof}

Combining Lemmas~\ref{lem:discretize}, \ref{lem:splitting}, and \ref{lem:monotonize} implies Theorem~\ref{thm:reduction}. By virtue of this result, in the remainder of this section, we focus only on instances with a sequence of monotonic scaled Bernoulli distributions.


\subsection{Factor-revealing LPs: finite and infinite-dimensional formulations}
\label{sec:LP}
Let us now formulate a parametric LP that captures the worst-case ratio of the optimal online policy to the prophet benchmark. To this end, first observe that the optimal online policy can be expressed as a simple dynamic program (DP) given the knowledge of the sequence of distributions. We use the DP formulation to set up our ``factor-revealing'' linear program.

\smallskip
\noindent\textbf{DP formulation:} To derive this DP, we identify a simple recurrence, also known as the \emph{Bellman optimality equation}~\citep{bellman1954theory}, which characterizes the optimal expected reward-to-go at each stage of the process. Suppose that the instance consists of the sequence of random variables $X_1, X_2, \ldots, X_n$ with known distributions. For $t = 0, \ldots, n-1$, define
\[
\Phi_t(x) \triangleq  \EX\left[\textrm{reward of an optimal online algorithm  holding $x$ and observing $(X_{t+1},\ldots,X_n)$}\right].
 \]
Here, any buyback cost the algorithm might have paid before observing $X_{t+1}$ is ignored. For consistency with the above, $\Phi_n(x) \triangleq x$, since there is no subsequent variable $X_{n+1}$ to observe and the algorithm simply keeps the value $x$. Given the above definition, $\Phi_t(x)$ satisfies the following recurrence equation:
\begin{equation}
\label{eq:recurrance}
\Phi_{t-1}(x) = \EX\left[\max \left\{ \Phi_t(x), \Phi_t(X_t) - fx\right\} \right]~,~~~~~~~\forall\, t\in[1:n],\forall\, x\in \mathbb{R}~.
\end{equation}
The above recurrence holds, simply because upon observing $X_t$, any online algorithm has two options: either keeping the value $x$ and skipping $X_t$ (in which case it obtains an expected reward of $\Phi_t(x)$), or paying the buyback cost of $f\cdot x$ and accepting the value $X_t$ (in which case it obtains an expected reward of $\Phi_t(X_t)-fx$). Given the value of $X_t$, the optimal online algorithm chooses the better of the two options, and in expectation obtains $\EX\left[\max \{ \Phi_t(x), \Phi_t(X_t) - fx\} \right]$.

Now we focus on instances in the form provided by Theorem~\ref{thm:reduction}: monotonic sequences of scaled Bernoulli distributions, with $X_i = v_i \mathrm{Ber}(q_i)$, where $0\triangleq v_0 <v_1 \leq v_2 \leq \ldots \leq v_n$. With a slight abuse of notation, let us define $\Phi_{i,t}\triangleq \Phi_t(v_i)$. Then the recurrence in \eqref{eq:recurrance} takes the following form:
\begin{align}
\label{eq:recurr}
\Phi_{i,t-1} = \max \{ \Phi_{i,t}, (1-q_t)  \Phi_{i,t} + q_t ( \Phi_{t,t} - f v_i) \}~,~~~~~~~\forall\,0\leq i<t\leq n~,
\end{align}
and we define $\Phi_{i,n}=v_i$ for all $i\in[0:n]$.
Here, we use the fact that $\Phi_t(x) \geq \Phi_t(0) - fx$ (or, in other words, it never makes sense to accept a random variable that is equal to $0$), which can be proved for the optimal online policy using an easy induction.

\smallskip
\noindent\textbf{Discrete LP formulation:} Fixing $n\in\mathbb{N}$ and the probabilities $\{q_j\}_{j\in[n]}$ as parameters, where $q_j\in[0,1]$, we formulate a minimization LP with variables $\{\Phi_{i,t}\}_{i,t\in[n]\cup\{0\}}$ and $\{v_i\}_{i\in[n]}$. The LP's objective is $\Phi_{0,0}$, representing the expected reward of an optimal online policy from the beginning (holding nothing). The LP constraints capture the optimality equation in \eqref{eq:recurr} and are valid for the optimal online policy. We also add a constraint that normalizes the optimal offline (the expected maximum value) in the monotonic sequence of scaled Bernoulli random variables $X_i=v_i\mathrm{Ber}(q_i)$, namely, $\sum_{i=1}^{n} v_i \, q_i \prod_{j=i+1}^{n} (1-q_j) = 1$. Note the non-linearity in the parameters $q_j$: this is why we cannot include them as variables in the LP, but leave them as external parameters. For any fixed choice of these parameters $\{q_j\}_{j\in[n]}$, this LP precisely captures the worst-case ratio between the optimal online policy and the optimal offline.\footnote{Note that the worst-case instance for the competitive ratio of the optimal online policy is an instance of the form $X_i = v_i \,\mathrm{Ber}(q_i)$, $0 \leq v_1 \leq \ldots \leq v_n$ due to \Cref{lem:monotonize}. Therefore, intuitively speaking, this LP characterizes the best-response of an adversary who picks the problem instance to minimize the competitive ratio, when we tie the hands of the adversary by fixing $n$ and $\{q_j\}_{j\in[n]}$.}   Formally, we have the following lemma, introducing our first LP.

\begin{lemma}
\label{lem:LP-1}
Given $q_1,\ldots,q_n \in [0,1]$, the following LP characterizes the (worst-case) competitive ratio of an optimal online policy for any instance of the form $\forall\,i\in[n]:~X_i = v_i \,\mathrm{Ber}(q_i)$:
\begin{align*}
 \qquad\underset{\Phi_{i,t}\geq 0\,,\,v_i\geq 0}{\min} \;\;\;\;\;\; & \Phi_{0,0} && \textrm{s.t} \\
& \Phi_{i,t-1} \geq  \Phi_{i,t} && \forall\, 0 \leq i < t \leq n\\
& \Phi_{i,t-1} \geq (1-q_t) \Phi_{i,t} + q_t (\Phi_{t,t} - f v_i) && \forall\,0 \leq i < t \leq n\\
& \Phi_{0,n} = 0~~,~~\Phi_{i,n} = v_i && \forall\, 1 \leq i \leq n\\
& \sum_{i=1}^{n} v_i q_i \prod_{j=i+1}^{n} (1-q_j) \geq 1
\end{align*}
\end{lemma}


\begin{proofof}{Proof}
Given any candidate worst-case instance $X_i = v_i\,\mathrm{Ber}(q_i)$, with $0 \leq v_1 \leq \dots \leq v_n$ (\Cref{lem:monotonize}), we first normalize $\EX[X_{\textrm{max}}]$ to $1$ by uniformly scaling all $v_i$ so that $\sum_{i=1}^{n} v_i\,q_i \prod_{j=i+1}^{n} (1 - q_j) \;=\; 1.$
For an optimal online policy in this normalized instance, $\Phi_{i,t} = \Phi_t(v_i)$ is feasible in our LP by the optimality equation in \eqref{eq:recurr}. Thus, the LP's optimal objective value is at most the expected reward of this online policy in the normalized instance, which in turn equals its competitive ratio in the unnormalized instance.

Conversely, let $\{\Phi^*_{i,t}, v^*_i\}$ be an optimal solution to the LP, and consider the instance $X_i = v^*_i\,\mathrm{Ber}(q_i)$. By \eqref{eq:recurr} and backward induction---and using that our LP is a minimization program---the optimal expected reward of holding $v^*_i$ at time $t$ is $\Phi^*_{i,t}$. 
Moreover, we claim we can add the constraint $v^*_{i+1} \geq v^*_i$ without changing the optimal objective value. Hence we may assume $v^*_1 \leq v^*_2 \leq \dots \leq v^*_n$. To see this claim, note that \Cref{lem:monotonize} proves that the worst-case performance of the optimal online policy is attained when the values $v_i$ are ordered in ascending order, and, the expression $\sum_{i=1}^{n} v_i\,q_i \prod_{j=i+1}^{n} (1-q_j)$ 
(which is equal to the expected last nonzero $X_i$ in the given order) can only increase when the values are rearranged in ascending order.

Putting these observations together, the optimal online policy has an overall expected reward $\Phi_{0,0}$, and $\EX[X_\text{max}] \;=\;\sum_{i=1}^{n} v^*_i\,q_i \prod_{j=i+1}^{n} (1-q_j) \;\geq\; 1$ when $v^*_1 \le v^*_2 \le \cdots \le v^*_n$. Thus, the worst-case competitive ratio of the optimal online policy is at most the LP's optimal objective value.
\end{proofof}

We now slightly reformulate the LP to obtain a simpler equivalent form by a change of variables $\Delta_{i,t} \triangleq \Phi_{i,t-1} - \Phi_{i,t}$ for $0 \leq i < t \leq n$, and applying the constraints $\Phi_{i,n} = v_i, \Phi_{0,n} = 0$:
\begin{align}
\label{eq:disc-primal-LP}\tag{\textsc{Primal-Disc-LP}}
\qquad\underset{\Delta_{i,t}\geq 0\,,\,v_i\geq 0}{\min} \;\;\;\;\;\;& \sum_{t=1}^{n} \Delta_{0,t}
     && \text{s.t.} 
    \\
     & \Delta_{i,t} \geq q_t \left( \sum_{j=t+1}^{n} (\Delta_{t,j} - \Delta_{i,j}) + v_t - (1+f) v_i \right) 
     && \forall \, 0 \leq i < t \leq n \nonumber\\
     & \sum_{i=1}^{n} v_i q_i \prod_{j=i+1}^{n} (1-q_j)\geq 1
     && \nonumber
\end{align}

Now we appeal to LP duality to derive another LP that will be much more useful for our purposes. We postpone the proof of the following lemma, primarily based on strong duality, to \Cref{apx:lem:duality}.

\begin{lemma}
\label{lem:LP-3}
Given $q_1,\ldots,q_n \in [0,1]$, the following (dual) LP gives the (worst-case) competitive ratio of an optimal online policy for any instance of the form $\forall\,i\in[n]:~X_i = v_i \,\mathrm{Ber}(q_i)$:
\begin{align*}
\label{eq:discrete-LP}\tag{\textsc{Dual-Disc-LP}}
\underset{x_{s,t}\geq 0\,,\,\Theta\geq 0}{\max} \;\;\;\;\;\; & \Theta && \textrm{s.t} \\
& x_{0,t} \leq q_t \left( 1 - \sum_{i=1}^{t-1} x_{0,i} \right) && \forall\, 1 \leq t \leq n \\
& x_{s,t}  \leq  q_t \left( \sum_{i=0}^{s-1} x_{i,s} - \sum_{j=s+1}^{t-1} x_{s,j} \right) && \forall\, 1 \leq s < t \leq n \\
& \Theta\, \hat{q}_t \leq \sum_{i=0}^{t-1} x_{i,t} - (1+f) \sum_{j=t+1}^{n} x_{t,j} && \forall\, 1 < t \leq n 
\end{align*}
where we define $\hat{q}_t \triangleq q_t \prod_{j=t+1}^{n} (1-q_j)$.
(Equivalently, $\hat{q}_t\in[0,1]$ is the probability that $X_t = X_{\textrm{max}}$ in an instance with monotonic scaled Bernoulli distributions).
\end{lemma}



The LP in \ref{eq:discrete-LP} has a natural interpretation. Specifically, $x_{0,t}$ represents the probability that $X_t$ is the first variable we accept, and for $1 \le s < t$, $x_{s,t}$ represents the probability that we swap $X_s$ for $X_t$ \revcolor{upon its arrival.} The first constraint encodes that by the time we reach $X_t$, there is probability $1 - \sum_{i=1}^{t-1} x_{0,i}$ of not having accepted any variable yet. Since $X_t$ appears with probability $q_t$, we get $q_t\bigl(1 - \sum_{i=1}^{t-1} x_{0,i}\bigr)$ as an upper bound on the probability of accepting $X_t$. Similarly, $\sum_{i=0}^{s-1} x_{i,s} - \sum_{j=s+1}^{t-1} x_{s,j}$ is the probability we have accepted $X_s$ and not yet replaced it by $X_j$ for some $j < t$. Thus, $q_t\Bigl(\sum_{i=0}^{s-1} x_{i,s} - \sum_{j=s+1}^{t-1} x_{s,j}\Bigr)$ upper-bounds the probability of swapping $X_s$ for $X_t$. We refer to these constraints as the \emph{online implementability} constraints.

The third constraint, which we refer to as the \emph{coverage constraint}, enforces that the probability of accepting $X_t$ minus $(1+f)$ times the probability that we discard it later (that is, the \emph{effective probability} of accepting $X_t$ given the buyback cost), should be at least $\Theta$ times the probability $X_t$ is the maximum. Satisfying this constraint ensures that we recover at least a $\Theta$-fraction of $X_t$'s contribution to the optimal offline. Although it is not immediately obvious that this condition is also necessary, \cref{lem:LP-3} shows that it must hold if a competitive ratio of $\Theta$ is achievable.

Finally, from a combinatorial perspective, we can view this LP as a \emph{generalized flow problem}, that is, a network flow with a linear leakage across edges~\citep{pulat1989relation}. In this combinatorial interpretation,  we have nodes $\{0,1,\ldots,n\}$ and directed edges $\{(s,t):0\leq s<t\leq n\}$. Moreover, $(1+f)x_{s,t}$ represents the flow going out from node $s$ to node $t$, each node extracts a demand proportional to $\Theta$, and an amount $f x_{s,t}$ of the flow $x_{s,t}$ along each edge $(s,t)$ is ``lost'' due to the buyback cost (hence, the residual flow $x_{s,t}$ goes into $t$ from $s$) (see \Cref{fig:flow}). Inspired by this interpretation, we demonstrate how to solve this LP in the following sections.

\begin{figure}[ht]
    \centering
\begin{tikzpicture}[>=stealth, scale=0.7]
\node[draw, rectangle] (N0) at (0, 0) {$0$};
\node[draw, circle]    (Ns) at (3.5, 0) {$s$};
\node[draw, circle]    (Nt) at (7, 0) {$t$};
\node[draw, circle]    (Nn) at (14, 0) {$n$};


\coordinate (Break0sLeft) at ($(N0)!0.35!(Ns)$);
\coordinate (Break0sRight) at ($(N0)!0.65!(Ns)$);
\coordinate (Mid0s) at ($(N0)!0.5!(Ns)$);

\draw (N0) -- (Break0sLeft);
\draw (Break0sRight) -- (Ns);

\node[above] at (Mid0s) [yshift=-4.5pt]{\dots};

\coordinate (BreaksTLeft) at ($(Ns)!0.35!(Nt)$);
\coordinate (BreaksTRight) at ($(Ns)!0.65!(Nt)$);
\coordinate (MidsT) at ($(Ns)!0.5!(Nt)$);

\draw (Ns) -- (BreaksTLeft);
\draw (BreaksTRight) -- (Nt);

\node[above] at (MidsT) [yshift=-4.5pt]{\dots};

\coordinate (BreaktNLeft) at ($(Nt)!0.42!(Nn)$);
\coordinate (BreaktNRight) at ($(Nt)!0.58!(Nn)$);
\coordinate (MidtN) at ($(Nt)!0.5!(Nn)$);

\draw (Nt) -- (BreaktNLeft);
\draw (BreaktNRight) -- (Nn);

\node[above] at (MidtN) [yshift=-4.5pt]{\dots};

\coordinate (Mid) at (5.25, -1.25);

\draw[midarrow, bend right=30] (Ns) to node[below, xshift=-10pt, font=\small]{$\boldsymbol{(1+f)x_{s,\,t}}$} (Mid);
\draw[midarrow, bend right=30] (Mid) to node[below, xshift=5pt, yshift=-4pt, font=\small]{$\boldsymbol{x_{s,\,t}}$} (Nt);

\coordinate (Below) at (5.25, -2.25);
\draw[midarrow,red] (Mid) -- (Below) node[below, font=\small]{$\boldsymbol{f x_{s,\,t}}$};

\draw[midarrow, bend left] (N0) to node[above] {$\boldsymbol{x_{0,\,s}}$} (Ns);
\draw[midarrow, bend left=60] (N0) to node[above] {$\boldsymbol{x_{0,\,t}}$} (Nt);
\end{tikzpicture}
\caption{Interpreting \ref{eq:discrete-LP} as a generalized flow; the red arrow indicates ``lost'' flow or the leakage (due to the buyback).}
\label{fig:flow}
\vspace{-3mm}
\end{figure}

\smallskip

\noindent\textbf{Continuous LP formulation:} Recall that the number of random variables $n\in\mathbb{N}$ served as a parameter in our discrete LP formulations. However, there is a worst-case instance for the competitive ratio involving infinitely many ``vanishing'' scaled Bernoulli random variables, each active with only a tiny probability. 
Intuitively, by splitting a variable $X_i = v_i \,\mathrm{Ber}(q_i)$ into smaller pieces (for example $X_{i_1},X_{i_2} = v_i \,\mathrm{Ber}(1-\sqrt{1-q_i})$) one can move towards a problem with vanishing probabilities without actually changing the problem.
This observation suggests that the problem can be reformulated in terms of a continuum of (scaled) Bernoulli random variables with infinitesimal probabilities of being non-zero. The first step in solving our discrete LP is to rigorously develop this idea by reformulating the problem in a continuous framework, as established in the next proposition. We refer to the parameter $\Theta$ in this proposition as the \emph{coverage parameter} in the remainder of the paper.
\begin{proposition}
\label{lem:continuous-reduction}
For a given $f>0$, if there exist non-negative, integrable functions $h$ and $g$ with domains $(0, 1]$ and $(0,1]\times (0,1]$, and parameter $\Theta>0$, such that
\begin{align}\label{integrals}\tag{\textsc{Dual-Cont-LP}} \\ \label{OIfromzero}\tag{A}
& th(t) \leq 1 - \int_0^t h(r)dr && \forall 0 < t \leq 1\\
\tag{B}
&\label{OI} tg(s, t) \leq h(s) + \int_0^s g(r, s)dr - \int_s^t g(s, r)dr && \forall 0 < s \leq t \leq 1 \nonumber\\
\label{coverage}\tag{C}
& \Theta \leq h(s) + \int_0^s g(r, s)dr  - (1+f) \int_{s}^{1} g(s, r)dr && \forall 0 < s \leq 1 \nonumber 
\end{align}
then $\alpha(f) \geq \Theta.$ 
\end{proposition}

\begin{proofof}{Proof} 
We construct a solution for \ref{eq:discrete-LP} using a pair of functions $h(.)$ and $g(.)$ that satisfy the above inequalities. To this end, without loss of generality, we first add a dummy variable $0$ at the beginning of our instance, distributed as $0 \cdot \mathrm{Ber}(1)$ (hence, we set $v_0 = 0$, and $q_0 = 1$), to ensure that the first random variable is deterministic. Since $q_0 = 1$ and we have defined $\hat{q}_t = q_t \prod_{j=t+1}^{n} (1-q_j)$, it follows that $\sum \hat{q_i} = 1$. Let us introduce the notation $r_j = \frac{\hat{q_j}}{q_j} = \sum_{i = 0}^{j}\hat{q_i}$. The $r_j$'s vary from $0$ to $1$, and represent the probability that the last non-zero random variable occurs before or at $X_j$. Note that $\hat{q_j} = r_j - r_{j-1},$ and thus $q_j = \frac{\hat{q_j}}{r_j} = \frac{r_j-r_{j-1}}{r_j}$. To obtain a feasible solution for \ref{eq:discrete-LP} , we just need to set the decision variables as follows:
$$x_{0,t} = \int_{r_{t-1}}^{r_t}h(r)dr~~~~~~\text{ and }~~~~~~x_{s,t}  = \int_{r_{s-1}}^{r_s}\int_{r_{t-1}}^{r_t}g(r, u)du dr.$$ 
Feasibility of the above solution in \ref{eq:discrete-LP} can be shown immediately by integrating the inequalities in \ref{integrals} from $r_{t-1}$ to $r_t$ and from $r_{s-1}$ to $r_s$ as appropriate. We postpone the details and include the calculations in \cref{apx:proof-cont-red} for the sake of completeness.
\end{proofof}

\subsection{Reduction to the differential equation}
\label{sec:reduction-to-diff-eq}
\newcommand{\buybackfactor}{f}
Our next goal is to construct a particular feasible solution to the continuous LP feasibility problem \ref{integrals} for any given buyback factor $f>0$. This solution is a candidate for being optimal, that is, the feasible solution to \ref{integrals} with the maximum possible $\Theta$. The existence of such a feasible solution with $\Theta = \theta_f$ (defined later) certifies that $\alpha(f) \geq \theta_f$. \revcolor{To guide this construction, we provide a natural interpretation of \ref{integrals}, which we call the \emph{quantile acceptance problem}, and show that constructing a feasible solution of \ref{integrals} is equivalent to finding an algorithm for this interpretation. We postulate that the optimal algorithm must  satisfy certain natural structural properties. As we will show, solving for these required properties directly gives rise to a differential equation that determines our candidate optimal solution to \ref{integrals}, as well as the algorithm for the quantile acceptance problem with an effective (conditional) selection probability no lower than our desired competitive ratio $\theta_f$, hence showing $\alpha(f) \geq \theta_f$.} Later, in \Cref{sec:lower-bound}, we will establish that $\alpha(f) \leq \theta_f$, confirming that our candidate solution is indeed optimal. 

\smallskip
\noindent\textbf{A differential equation:} To explain our construction and the main result, we first define the following differential equation (in an unconventional form), which plays a key role in our construction.

\begin{definition}[Y-Function Differential Equation]
\label{defn:differentialequation}
\revcolor{For any $f \geq 0$, let $\cf\triangleq \frac{f}{1+f}$. Let $k_1 \geq \cf$ be a parameter establishing the lower bound of the feasibility region for $y(t)$. Consider the following differential equation for $t \in [k_1,1]$ and the function $y \colon [k_1,1] \to [0, 1]$ satisfying the initial condition $y(k_1)=0$:}
\begin{align*}\label{difEq}\tag{\textsc{Diff-Eq}}
    & y'(y(t)) 
      = \frac{y(t)-\cf}{t-\cf}
        \left(2 - \frac{y(t)}{y'(t)\,\bigl(t-\cf\bigr)}\right), 
      &&\forall t \in [k_1,1] : y(t)\geq k_1 \nonumber \\[6pt]
    & y(t) 
      = \frac{(t-\cf)^2}{1-\cf} - 1 + \cf + y(1),
      && \forall t \geq y(1)
\end{align*}
\end{definition}



If we conjecture that a solution to this differential equation satisfies $y(t)< t$, we could then solve the equation running ``backward in time,'' since the equation relates the derivative of $y$ at an earlier point $y(t)$ to the value (and derivative) of $y$ at a later point $t$. Building on this observation, in the following theorem, we show that \ref{difEq} admits such a ``proper'' solution. 

\begin{theorem}[Existence of Proper Y-Function]
\label{thm:diff-eq-existence}
For all $f > 0$, there exists a solution $\yf$ to the differential equation when $k_1=\cf$. Moreover, this solution satisfies the following properties:
\begin{enumerate}[label=(\roman*)]
    \item $\yf(t)$ is continuous, strictly increasing, and $\yf(t)< t$ for all $t \in [k_1, 1]$.
    \item $\frac{\yf(t)}{t - \cf}$ is increasing for all $t \in (k_1, 1]$ $\big($or equivalently $\frac{t}{y^{-1}_f(t)-\cf}$ is increasing for all $t\in (0,y(1)]\big)$
    \item There exists a finite integer $n\leq 1+\frac{1}{f}$ such that $y^{(n)}_f(1)\leq k_1$, where $y^{(n)}_f\triangleq \underbrace{y_f \circ y_f \circ \cdots \circ y_f}_{n\text{ times}}$.
\end{enumerate}
\end{theorem}
We sketch the proof of the existence of the solution $\yf$ to \ref{difEq} in \Cref{sec:diff-equation-sketch}, and defer the formal proof to \Cref{sec:diff-eq-soln}.\footnote{Notice that while our differential equation is defined for $k_1 \geq c_f$, we only prove the existence of a solution when $k_1 = c_f$. This is because our hard instance (in particular, \cref{lem:indifference}) implies that this is the best choice for $k_1$.} We are now ready to state the main result of this section.

 \begin{theorem}[Optimal Competitive Ratio -- Lower bound]
 \label{thm:alg-existence}
    The optimal online algorithm obtains a competitive ratio of at least $\theta_f \triangleq \frac{1}{2-\yf(1)}$ (see \Cref{fig:competitive-ratios}).
 \end{theorem}

\smallskip
\noindent\textbf{Overview of the construction:} In the rest of this section, we sketch how to construct a feasible solution to \ref{integrals} using the solution of the differential equation $\ref{difEq}$, to prove \Cref{thm:alg-existence} (formal proofs are postponed to the appendix). This construction will proceed in multiple steps:
\smallskip
\begin{itemize}
    \item \textbf{\Cref{sec:overview-interpret}:} We start by providing an interpretation of \ref{integrals} . Then, we argue informally that constructing a feasible solution of  \ref{integrals} is equivalent to finding an algorithm that acts in an online fashion on a continuum of Bernoulli random variables, accepting each one with an effective probability $\Theta$. Such an algorithm is in some sense an analog to the \emph{online contention resolution schemes}~\citep{feldman2016online} adapted to the buyback model.
    \smallskip
    \item \textbf{\Cref{sec:overview-greedy}:} Next, we identify a natural class of algorithms that are likely to accept each variable with a large effective probability. We parameterize this class of algorithms by a number $k_1\in (0,1]$, and a pair of functions, $\varphi:(0,1]\rightarrow [0,1]$ and \revcolor{$y \colon [k_1,1] \to [0, 1]$}.
    \smallskip
    \item \textbf{\Cref{sec:overview-property}:} Then, we identify the properties that $\varphi$ and $y$ should satisfy to obtain an online algorithm with an effective probability $\Theta$ of accepting each random variable.
    \smallskip
    \item \textbf{\Cref{sec:overview-diffeq}:} Finally, we show how these properties for $\varphi$ and $y$ can be expressed in terms of just the function $y$. This immediately leads to the differential equation \ref{difEq} for $y$.
\end{itemize}

In summary, for a given $f>0$, the (proper) solution $y_f$ of \ref{difEq} allows us to construct $\varphi$ and $y$ with desirable properties, leading to an algorithm for the problem of accepting Bernoulli random variables on a continuum that obtains the maximum effective probability $\Theta=\theta_f$. This algorithm can be interpreted as a solution to \ref{integrals} with the maximum possible coverage $\Theta=\theta_f$ .

\subsubsection{Interpretation of \ref{integrals} and \revcolor{the ``quantile acceptance problem''}}
\label{sec:overview-interpret}
Recall the natural interpretation of \ref{eq:discrete-LP};
we provide a similar (informal) interpretation for \ref{integrals} by considering an instance with a continuum of (monotone) scaled Bernoulli random variables, each having infinitesimal non-zero probability, arriving over the time interval $[0,1]$. Inspired by the proof of \cref{lem:continuous-reduction}, we interpret the infinitesimal $ds$ as the probability that the random variable arriving at time $s$ is the ``maximum'' or the last non-zero random variable (i.e., $\hat{q}_s \approx ds$), so the location of the maximum is uniformly distributed in $[0,1]$. Furthermore, we interpret $q_s \approx \tfrac{ds}{s}$ as the probability that the random variable arriving at time $s$ is non-zero (since $\hat{q}_s = \Pr[\text{location of maximum is in }[0,s]] \cdot q_s = s \, q_s$).\footnote{The interval $[0,1]$ can be interpreted as all quantiles of $X_\textrm{max}$, arriving sequentially. Once a quantile $s$ arrives, it will be \emph{active} with probability $\tfrac{ds}{s}$. The last active quantile corresponds to the realized quantile of $X_{\textrm{max}}$, which is uniformly distributed over $[0,1]$.} 


We now consider running an online algorithm on this continuum. The algorithm starts at $0$,\footnote{In fact, to be more mathematically precise, the algorithm starts at $0^+$ by skipping the $0$ itself.} makes a first pick, performs swaps when necessary, and attempts to ensure that the random variable arriving at time $t$ is picked with an effective probability of $\Theta dt$. That is, the probability of picking $t$ minus $(1+f)$ times the probability of swapping away from $t$ is at least $\Theta dt$. \revcolor{We informally refer to the problem of designing an online algorithm for maximizing the effective probability as the ``quantile acceptance problem.'' Now given an online algorithm for this problem,} we interpret $h(s)\,ds$ as the probability that the random variable arriving at $s$ is our algorithm's first pick, and $g(s,t)\,ds\,dt$ as the probability that the algorithm will swap the random variable arriving at $s$ with the one arriving at $t > s$. Similar to our discussion of \ref{eq:discrete-LP}, the existence of such an algorithm is equivalent to the existence of a feasible solution to \ref{integrals}. 

\subsubsection{Greedy-type algorithms with a monotonicity property}
\label{sec:overview-greedy}
To search for an online algorithm \revcolor{for the quantile acceptance problem} on the continuum with the largest possible $\Theta$, we posit a particular \emph{monotonicity} property: if the algorithm decides to swap from $s$ to $t$ for $s<t$, then for any $s'<s$, it would also swap from $s'$ to $t$. Define $y(t)$ as the largest $s$ for which the algorithm is willing to swap to $t$. We further assume that the function $y(\cdot)$ satisfies the following properties:
\begin{align}
    \label{eq:y-property-1}\tag{\texttt{Property I}}&\textrm{$y(t)$ is strictly increasing in $t$, continuous, and $y(t)< t$}~,\\
    \label{eq:y-property-2}\tag{\texttt{Property II}}&\textrm{$y(k_1)=0$ for a parameter $k_1$},
\end{align}
Note that $y(t)$ is only well defined for $t\in[k_1,1]$. We also define the inverse function $\tau(\cdot)\triangleq y^{-1}(\cdot)$ with support $[0,y(1)]$. By convention, we extend the support by setting $\tau(s)=1$ for $s\in[y(1),1]$.


To fully specify the algorithm, we must also define the probability with which it selects each $s \in (0,k_1]$ as its first pick. We denote this probability by $\varphi(s)\,ds$, for a function $\varphi : (0,k_1] \to \mathbb{R}$. With these ingredients in place, we now informally describe the execution of the resulting simple candidate algorithm on the continuum $[0,1]$. See \Cref{alg:quantile-selection} for the details and formal description of this candidate algorithm.

\begin{algorithm}[ht]
\raggedright
\caption{(informal) candidate algorithm for the continuum $[0,1]$}
\label{alg:quantile-selection}
Initialize $s\leftarrow 0$ {\small\color{\commentcolor}\tcc{$s\in[0,1]$ tracks what the algorithm holds at any time}}

\For{active $t\in(0,1]$}{
\If{$t<k_1$ and $s=0$}{
Pick $t$ with probability $\frac{t\varphi(t)}{1-\int_{0}^{t}\varphi(s)\,ds}$. {\small\color{\commentcolor} \tcc{$t$ is active w.p. $\frac{dt}{t}$, and nothing was accepted before $t$ w.p. $1-\int_{0}^{t}\varphi(s)\,ds$, so the algorithm ends up selecting $t$ w.p. $\displaystyle\varphi(t)dt$. }}

If $t$ is picked, update $s\leftarrow t$.
}\Else{
\If{$t>\tau(s)\left(\equiv s<y(t)\right)$}{Swap from $s$ to $t$. {\small\color{\commentcolor} \tcc{$t$ is active w.p. $\frac{dt}{t}$, so the algorithm ends up selecting $t$ w.p. $\frac{dt}{t}$ conditional on holding $s$.}}

Update $s\leftarrow t$.
}
}
}
\end{algorithm}



While we only defined $\varphi(t)\,dt$ for $t \in (0,k_1]$, we can naturally define the \emph{extension} of $\varphi(t)$ to $(0,1]$ by letting it be the probability with which $t$ is ever picked by our algorithm (this includes both the probability that it is the algorithm's final pick, and the probability that the algorithm switches away from $t$). We now ask the following natural question regarding $\{y(\cdot), \varphi(\cdot)\}$:
\begin{displayquote}
\emph{What conditions must $\varphi$ and $y$ satisfy so that (i) \Cref{alg:quantile-selection} is an implementable online algorithm, and (ii) it achieves an effective selection probability of $\Theta\,dt$ for every $t \in (0,1]$?}
\end{displayquote}


\subsubsection{Properties of the functions $\varphi$ and $y$}
\label{sec:overview-property}

We now present the following simple and natural observations, which help pinpoint all the properties of $\varphi$ and $y$ needed to satisfy (i) and (ii).
\begin{itemize}[leftmargin=0.25in]

\item \Cref{alg:quantile-selection} is well-defined if the probability of accepting $t<k_1$ does not exceed $1$. Therefore:
\begin{equation}
\label{eq:OIfromzero-phi-property3}\tag{\texttt{Property III}}
    t\varphi(t)\le1-\int_{0}^{t}\varphi(s)\,ds
    \quad\forall\,t\in(0,k_1],
\end{equation}

\item Similarly, for any $t\in[k_1,1]$, the probability that \Cref{alg:quantile-selection} does not pick anything between $y(t)$ and $t$ is exactly $1-\int_{y(t)}^{t}\varphi(s)\,ds$. At the same time, the algorithm should pick the random variable arriving at time $t$ as long as it is non-zero (i.e., is \emph{active}) \textit{if and only if} it has not picked anything between $y(t)$ and $t$. Since $\varphi(t)\,dt$ is the probability of ever picking $t$, the quantity $t\varphi(t)=\frac{\varphi(t)\,dt}{\frac{dt}{t}}$ is the probability of picking $t$ conditioned on $t$ being non-zero. Therefore, we get:
\begin{equation}
\label{eq:phi-after-k1-property4}\tag{\texttt{Property IV}}
t\varphi(t)=1-\int_{y(t)}^{t}\varphi(s)\,ds,
\quad\forall\,t\in[k_1,1].
\end{equation}

\item We further posit that \Cref{alg:quantile-selection} always attempts to select the random variable arriving at time $t$ with the lowest possible effective probability, while still maintaining an effective selection probability of $\Theta dt$ for all $t$. That is, for all $t\in(0,1)$ we have:
\begin{align}
\label{eq:coverage-tight}
\Theta\,dt
&=\varphi(t)\,dt-(1+f)\Pr[\text{\Cref{alg:quantile-selection} swaps from }t].
\end{align}
Assuming this, we make the following observations:
\begin{enumerate}
    \item For $t\in(0,y(1)]$, \Cref{alg:quantile-selection} ends up swapping from $t$ if it picks $t$ at some point \emph{and} the last non-zero random variable arrives after $\tau(t)$. These two events are independent. Moreover, the probability of picking $t$ is $\varphi(t)\,dt$, and the probability that the last non-zero arrival is after $\tau(t)$ is $1-\tau(t)$. Combining these with \cref{eq:coverage-tight}, for any $t\in (0,y(1)]$ we obtain:
   \begin{equation}
   \label{eq:coverage-property5-part1}
   \Theta
   =\varphi(t)
   -(1+f)\varphi(t)\bigl(1-\tau(t)\bigr)
   =\varphi(t)\Bigl((1+f)\tau(t)-f\Bigr).
   \end{equation}

    \item For $t\ge y(1)$, we have $\varphi(t)=\Theta$, since there is no further swapping if we pick something after $y(1)$. Moreover, setting $t=1$ in \ref{eq:phi-after-k1-property4}, we see:
   \begin{align*}
   \Theta=\varphi(1)=1-\int_{y(1)}^{1}\varphi(s)\,ds
   =1-\int_{y(1)}^{1}\Theta\,ds
   =\frac{1}{2-y(1)}.
   \end{align*}
\end{enumerate}
Putting the above observations together yields:
\begin{equation}
\label{eq:coverage-property5}\tag{\texttt{Property V}}
\Theta =\frac{1}{2-y(1)}
\quad\text{and}\quad
\varphi(t)=\frac{\Theta}{(1+f)\,y^{-1}(t)-f}
\quad\forall\,t\in(0,1],
\end{equation}
where we adopt the convention $y^{-1}(t)\triangleq 1$ for $t\in(y(1),1]$. Notice that this is where we need $k_1 = y^{-1}(0) \geq \frac{f}{1+f} = k_1$.

\item  Taking the derivative on both sides in \ref{eq:phi-after-k1-property4}, we get that:
\begin{equation}
\label{eq:differentiated-property-iv}
2\varphi(t) + t \varphi'(t) = \varphi(y(t))y'(t),
\quad\forall t \in [k_1,1].
\end{equation}

Multiplying both sides by $t$ and integrating the resulting formula from $k_1$, we get:
\begin{equation}
\label{eq:property6}\tag{\texttt{Property VI}}
    t^2\varphi(t)=k_1^2 \varphi(k_1)+\int_{0}^{y(t)} \varphi(s)y^{-1}(s)\,ds,\quad \forall t\in[k_1,1].
\end{equation}
\end{itemize}

\subsubsection{Identifying the differential equation}
\label{sec:overview-diffeq}
Observe that as a consequence of \ref{eq:coverage-property5},  $\varphi$ can be defined in terms of $y$. Therefore, statements about the existence of functions $\varphi$ and $y$ with the desired properties ought to be reducible to a statement purely about $y$. By substituting all instances of $\varphi$ in \cref{eq:differentiated-property-iv} using \ref{eq:coverage-property5}, we derive a version of \ref{eq:phi-after-k1-property4} entirely in terms of $y$:
\begin{equation}
\label{eq:rearranged-differential}
\frac{2}{y^{-1}(t)-\cf} - \frac{t\cdot(y^{-1})'(t)}{\left(y^{-1}(t)-\cf\right)^2} = \frac{y'(t)}{t-\cf},
\quad \forall t \in [k_1,1].
\end{equation}
Notice that for $t \geq y(1)$, $y^{-1}(t) = 1$, and hence the second term on the left-hand side becomes $0$, implying that:
\begin{equation*}
 \frac{2}{1-\cf} = \frac{y'(t)}{t-\cf}, \quad \forall t \in [y(1),1]
\end{equation*}
which yields that  $y(t) 
      = \frac{(t-\cf)^2}{1-\cf} - 1 + \cf + y(1),$ for all $t \geq y(1)$ (an ``initial condition'' in \ref{difEq}). For $t \leq y(1)$, by replacing $t$ with $y(t)$ in \cref{eq:rearranged-differential} we obtain the following simplified formula:
\begin{align*}
     \frac{2}{t-\cf} - \frac{\textcolor{blue}{y(t)}}{y'(t)\left(t-\cf\right)^2} = \frac{y'(y(t))}{y(t)-\cf}  \qquad  \forall t \in [k_1,1] : y(t) \geq k_1,
\end{align*}
which is exactly our differential equation in \ref{difEq} after rearrangement.

\smallskip
\subsubsection{Formal analysis} Equipped with all the properties identified in our earlier informal discussion, we are now ready to formally show that the \cref{thm:diff-eq-existence} guarantees the existence of functions $\{y(\cdot),\varphi(\cdot)\}$ satisfying these properties, and furthermore that the existence of these functions certifies the existence of the desired feasible solution for \ref{integrals}. We summarize these results in the following propositions. The proofs are essentially a formalization of \revcolor{the} informal arguments. We postpone the proofs to \cref{apx:proof:phiy-existence} and \Cref{apx:proof:properties}.

\begin{proposition}
\label{prop:from-diffeq-to-phiy}
If there exists a solution to \ref{difEq} as defined by \cref{thm:diff-eq-existence}, then there \revcolor{exists} nonnegative real-valued functions $y:[k_1,1]\rightarrow [0,1]$ and $\varphi:(0,1]\rightarrow\mathbb{R}$ that satisfy \ref{eq:y-property-1}, \ref{eq:y-property-2}, \ref{eq:OIfromzero-phi-property3}, \ref{eq:phi-after-k1-property4}, \ref{eq:coverage-property5}, and \ref{eq:property6}.
\end{proposition}

\begin{proposition}
\label{prop:from-phiy-to-integrals}
If there exist nonnegative real-valued functions $y:[k_1,1]\rightarrow [0,1]$ and $\varphi:(0,1]\rightarrow\mathbb{R}$ that satisfy \ref{eq:y-property-1}, \ref{eq:y-property-2}, \ref{eq:OIfromzero-phi-property3}, \ref{eq:phi-after-k1-property4}, \ref{eq:coverage-property5}, and \ref{eq:property6}, then there exist a feasible solution for \ref{integrals} with coverage $\displaystyle\Theta=\frac{1}{2-y(1)}$.
\end{proposition}

We have succeeded in reducing \Cref{thm:alg-existence} to \Cref{thm:diff-eq-existence}.

\subsection{Solving the differential equation: proof sketch}
\label{sec:diff-equation-sketch}
As discussed above, we have reduced the problem of analyzing the optimal competitive ratio to solving \ref{difEq} with $k_1 = \cf$. Since $y(1)$ is unknown, let us introduce the parameter $\theta = \frac{1}{2-y(1)}.$ Recalling that $\cf = \frac{f}{1+f}$, we must then solve this differential equation:
\begin{itemize}
\item For $t \in [2-\frac{1}{\theta}, 1]$, $y(t)$ is given explicitly as
\begin{equation}
\label{eq0:init-segment}
 y(t) = 1 - \frac{1}{\theta} + \kk + \frac{(t-\kk)^2}{1-\kk} 
\end{equation}
As stated earlier, this segment can be viewed as the initial condition of the differential equation. 
\item For every $t<1$, if $y(t) \geq k_1$, then
\begin{equation}
\label{eq0:diff-eq}
y'(y(t)) = \frac{y(t) - \kk}{t - \kk} \left( 2 - \frac{y(t)}{y'(t) (t - \kk)} \right)
\end{equation}
\end{itemize}

Given a parameter $\theta \in [1/2,1]$, a solution (or at least a partial solution) can be found as follows:
The function $y(t)$ is determined explicitly on the interval $[y(1), 1]$; recall that $y(1) = 2 - \frac{1}{\theta}$. Given this segment of the solution, we can now compute an extension to the interval $[y(y(1)), y(1)]$, using equation (\ref{eq0:diff-eq})---note that all the quantities on the right-hand side are known for $t \in [y(1), 1]$, due to the initial segment. By integration, we obtain a solution $y(t)$ for $t \in [y(y(1)), y(1)]$. Then we iterate this process, obtaining a solution on the interval $[y(y(y(1))), y(y(1))]$, etc. The process terminates if we extend the solution to the full interval $[\cf, 1]$, or if the solution fails for some reason. 


Our goal is to find a solution $y(t)$ on $[\cf, 1]$ such that $y(\cf) = 0$ and $y(t)$ is increasing on $[\cf, 1]$. As we prove in \cref{sec:diff-eq-soln}, such a solution can always be found for a suitable choice of $\theta \in [1/2,1]$, and we will show in \cref{sec:lower-bound} that this is exactly the optimal bound for a given value of $f$. Here we discuss at a high-level why such a solution exists and how the proof works. Before proceeding to the general case, let us discuss two special cases in which we can explicitly solve the differential equation and obtain a closed-form solution $y(t)$, as well as the optimal competitive ratio.

\paragraph{The case of $f \geq 1$:}
In this case, the choice of 
 \[\crf = \frac{1}{2-y_f(1)} = \frac{1+f}{1+2f} \] 
and $k_1 = \cf$ leads to a complete solution given by the initial segment: 
\[y_f(t) = (1+f)(t-\cf)^2 \text{ for } \cf \leq t \leq 1. \]
We have $y_f(1) = 2 - \frac{1}{\crf} = \frac{1}{1+f} \leq \frac{f}{1+f} = \cf$. This shows that the initial segment already covers the interval $[\cf, 1]$ and the solution is complete. Given the simple two-variable example discussed in the introduction, we also obtain that $\frac{1+f}{1+2f}$ is the optimal competitive ratio for $f \geq 1$.

\begin{figure}[ht]
\centering
\begin{minipage}{0.49\textwidth}
\centering
\begin{tikzpicture}[scale=0.85]
  \begin{axis}[
    axis lines=middle,  
    enlargelimits,      
    clip=false,
    xmin=0.0, xmax=1,
    ymin=0.0, ymax=0.5,
    label style={font=\footnotesize},   
    tick label style={font=\footnotesize},
  ]
    \addplot [thick, domain=0.5:1.0] {2*(x-0.5)*(x-0.5)};
    \filldraw (axis cs:0.5,0) circle (2pt);
    \filldraw (axis cs:1,0.5) circle (2pt);
  \end{axis}

   \node at (7,0.12) {\footnotesize{${t}$}};
   \node at (1.1,5.5) {\footnotesize{$y_f(t)$}};
   \node at (6.5,3) {\footnotesize{$f=1$}};
   \node at (3.3,1) {\footnotesize{$\cf=0.5$}};
   \node at (6.5,2) {\footnotesize{$y_f(1)=0.5$}};
   \node at (6.5,1) {\footnotesize{$\crf=\tfrac{2}{3}$}};
\end{tikzpicture}
\end{minipage}
\hfill
\begin{minipage}{0.49\textwidth}
\centering
\begin{tikzpicture}[scale=0.85]
  \begin{axis}[
    axis lines=middle,  
    enlargelimits,      
    clip=false,
    xmin=0.0, xmax=1,
    ymin=0.0, ymax=0.713,
    label style={font=\footnotesize},   
    tick label style={font=\footnotesize},
  ]
    \addplot [thick, domain=0.25:0.71] {1.2*(x-0.25)*(x-0.25)};
    \addplot [thick, domain=0.71:1.00] {1.2*(x-0.25)*(x-0.25) + 0.5*(x-0.71)*(x-0.71)};
    \filldraw (axis cs:0.25,0) circle (2pt);
    \filldraw (axis cs:0.713,0.257) circle (2pt);
    \filldraw (axis cs:1,0.713) circle (2pt);
  \end{axis}

   \node at (7,0.12) {\footnotesize{${t}$}};
   \node at (1.1,5.5) {\footnotesize{$y_f(t)$}};
   \node at (6.5,3) {\footnotesize{$f=\tfrac{1}{3}$}};
   \node at (2.3,1) {\footnotesize{$\cf=0.25$}};
   \node at (6.5,2) {\footnotesize{$y_f(1)=\tfrac{4\sqrt{5}+6}{4\sqrt{5}+12} \simeq 0.713$}};
   \node at (6.5,1) {\footnotesize{$\crf=\tfrac{4\sqrt{5}+12}{4\sqrt{5}+18} \simeq 0.777$}};
  % ...
\end{tikzpicture}
\end{minipage}
\end{figure}

\paragraph{The case of $\frac{1}{3} \leq f < 1$:}
Here, after setting $k_1= \cf$, we need to extend the solution beyond the initial segment. By integration, we obtain a solution on $[y(y(1)), y(1)]$, and then we determine that in order to obtain $y(\cf) = 0$, the correct choice of the parameter $\theta$ is
\begin{equation}
\label{eq:cr-3}
    \crf = \frac{1}{2-y(1)} = \frac{(1+f)(\sqrt{f(2-f)}+1)}{(1+f)\sqrt{f(2-f)}+3f+1}.
\end{equation}

If we define $a_f = \frac{f\left(1-\sqrt{f(2-f)}\right)^2}{(1+f)(1-f)^2}$, the solution has the form
\begin{align*}
    y_f(t) &= \begin{cases} 
     \frac{(t-\cf)^2}{1-\cf} - a_f & \text{ for } y(1) \leq t \leq 1 \\
      \sqrt{\frac{1+f}{t+a_f}} \cdot \left((t-\cf)^2-\cf\left(\sqrt{t+a_f}-\sqrt{\cf+a_f}\right)^2\right) & \text{ for } \cf \leq t \leq y(1)
   \end{cases}
\end{align*}






This implies that there exists an algorithm obtaining a competitive ratio of $\crf = \frac{(1+f)(\sqrt{f(2-f)}+1)}{(1+f)\sqrt{f(2-f)}+3f+1}$ in the regime of $f \in [\frac13, 1]$. The interval of $[\frac13, 1]$ is determined by the fact that for $f < 1/3$, we need to iterate the procedure more than two times to obtain a solution on $[\cf, 1]$.





\paragraph{The general case.}
For general $f>0$ (in particular $0 < f < \frac13$), we do not provide an explicit solution of the differential equation, and we also do not have an explicit formula for the optimal competitive ratio $\crf$. Nevertheless, we prove that a solution exists for a suitable choice of $\theta$ and we describe the solution in a somewhat implicit form. Here we sketch a high-level argument for the existence of a monotonic solution such that $y(\cf) = 0$.

Consider a fixed buyback factor $f>0$, and different choices of $\theta \in [\frac12, 1]$ (which also determine the value $y(1)$ and the initial segment of the solution). By solving the differential equation (as far as possible), we obtain a solution $y^\theta(t)$. Depending on our choice of $\theta$, several options could occur:
\begin{itemize}
\item The solution can be extended to the full interval $[\cf,1]$ so that it is monotonically increasing and $y^\theta(\cf) = 0$. This is the ideal option---it means that we have found a value of $\theta$ for which a $\theta$-competitive algorithm can be designed.
\item A monotonic solution exists on $[\cf,1]$, but it fails to satisfy $y(\cf) = 0$. 
\item A monotonic solution does not exist on $[\cf,1]$---either because we obtain $y'(t) = 0$ for $t > \cf$ which leads to division by $0$ in the differential equation, or we obtain $y'(t) < 0$ which means that the solution is not monotonically increasing. 
\end{itemize}
The figure below illustrates the different options that might happen.
\begin{center}
\begin{tikzpicture}[scale=1]
\centering
  \begin{axis}[
    domain=0:1, 
    samples=300, 
    axis lines=middle, 
    enlargelimits, 
    clip=false,
  ymin=0.055, ymax=0.61,
  label style={font=\footnotesize},   
  tick label style={font=\footnotesize},
  ]

    \addplot [gray, thick, dashed] [domain=0.2:0.5] {(x-0.25)*(x-0.25) - 0.05 };
    \addplot [gray, thick, dashed] [domain=0.5:0.75] {(x-0.25)*(x-0.25) + 0.1*(x-0.5) - 0.05 };
    \addplot [gray, thick, dashed] [domain=0.75:1.00] {0.275 + 1.2*(x-0.75) - 0.05 };

    \addplot [gray, thick , dashed] [domain=0.2:0.5] {(x-0.25)*(x-0.25) - 0.025 };
    \addplot [gray, thick, dashed ] [domain=0.5:0.75] {(x-0.25)*(x-0.25) + 0.1*(x-0.5) - 0.025 };
    \addplot [gray, thick, dashed ] [domain=0.75:1.00] {0.275 + 1.2*(x-0.75) - 0.025 };

    \addplot [thick] [domain=0.25:0.5] {(x-0.25)*(x-0.25) };
    \addplot [thick] [domain=0.5:0.75] {(x-0.25)*(x-0.25) + 0.1*(x-0.5) };
    \addplot [thick] [domain=0.75:1.00] {0.275 + 1.2*(x-0.75) };

    \addplot [gray, thick, dashed ] [domain=0.20:0.5] {(x-0.25)*(x-0.25) + 0.025 };
    \addplot [gray, thick, dashed ] [domain=0.5:0.75] {(x-0.25)*(x-0.25) + 0.1*(x-0.5) + 0.025 };
    \addplot [gray, thick, dashed ] [domain=0.75:1.00] {0.275 + 1.2*(x-0.75) + 0.025 };

    \addplot [gray, thick, dashed ] [domain=0.2:0.3] {4*(x-0.3)*(x-0.3) + 0.2*(x-0.3) + 0.05 };
    \addplot [gray, thick, dashed ] [domain=0.3:0.8] {(x-0.3)*(x-0.3) + 0.2*(x-0.3) + 0.05 };
    \addplot [gray, thick, dashed ] [domain=0.8:1.00] {0.4 + 1.1*(x-0.8 };

    \addplot [gray, thick, dashed ] [domain=0.25:0.4] {10*(x-0.4)*(x-0.4) + 0.1*(x-0.4) + 0.12};
    \addplot [gray, thick, dashed ] [domain=0.4:0.6] {2*(x-0.4)*(x-0.4) + 0.1*(x-0.4) + 0.12 };
    \addplot [gray, thick, dashed ] [domain=0.6:1.00] {0.22 + 1.05*(x-0.6 };

    \addplot [gray, thick, dashed ] [domain=0.35:0.5] {20*(x-0.5)*(x-0.5) + 0.3*(x-0.5) + 0.2 };
    \addplot [gray, thick, dashed ] [domain=0.5:0.6] {3*(x-0.5)*(x-0.5) + 0.3*(x-0.5) + 0.2 };
    \addplot [gray, thick, dashed ] [domain=0.6:1.00] {0.26 + 1.0*(x-0.6 };

    \draw (60, 60) -- (60, 40);
    \node at (60, -50) {$c_f$};
\draw[->, thick]
  (axis cs:1.03, 0.5) -- (axis cs:1.03, 0.67)
  node[midway, right] {\footnotesize{$\theta~~\text{increases}$}};
  \draw[-, thick, draw=cornellred]
  (axis cs:0.26, -0.07) -- (axis cs:0.26, 0.6);


    \end{axis}

 \node at (0.5,5.8) {$y^\theta(t)$};
  \node at (6.7,-0.3) {$t$};
  \node at (6,2.9) {$f=\tfrac{1}{3}$};
  \node at (6,2.1) {$c_f=\tfrac{f}{1+f}=0.25$};
  \node at (6,1.3) {$\theta=?$}; 
\end{tikzpicture}
\end{center}
The crux of the proof is that the solution $y^\theta(t)$ as a function of the parameter $\theta$ behaves {\em continuously in} $\theta$, for a range of $\theta$ such that the solution $y^\theta(t)$ can be defined on the interval $[\cf,1]$. We remark that due to the non-linear nature of the differential equation, this is by no means an obvious fact. In fact, the behavior of $y^\theta(t)$ can be rather unstable for larger values of $\theta$, where the solution cannot be extended to the interval $[\cf,1]$, and the solution can break down or diverge in various ways. 

A key property that we prove is that {\em if a solution $y^\theta(t)$ is well-defined on $[\cf,1]$, and $y(\cf) \leq 0$, then $y^\theta(t)$ must be convex}. This leads to other desirable properties of $y^\theta(t)$ (for example, the derivative $y'(t)$ for $t > \cf$  satisfies $y'(t) \geq \frac{y(t)}{t-\cf}$ which means that quantities in the differential equation do not diverge). Finally, a continuity argument proves that since there exist solutions such that $y^\theta(\cf) < 0$, and also $y^{\theta'}(\cf) > 0$, there must be a critical value $\theta^*$ such that $y^{\theta^*}(\cf) = 0$, and this is the desired value $\theta=\theta_f$ for any given $f>0$. The rigorous proof of this fact is quite technical and is deferred to \Cref{sec:diff-eq-soln}. We next show the actual optimality of this factor via a concrete worst-case instance.


\section{Worst-Case Instances}
\label{sec:lower-bound}
Previously, we showed that the competitive ratio $\crf = \frac{1}{2-\yf(1)}$ can be achieved by an optimal online algorithm for any $f>0$. Our goal in this section is to prove that this result is tight.
\begin{theorem}[Optimal Competitive Ratio -- Upper bound]
\label{thm:WorstCase}
No online algorithm, either deterministic or randomized, can achieve a competitive ratio better than $\crf = \frac{1}{2-\yf(1)}$.
\end{theorem}

Recall the simple two-variable example showing that no competitive ratio better than $\frac{1+f}{1+2f}$ can be achieved by an online algorithm: $X_1=1$ (deterministic) and $X_2=(1+f)$ w.p. $\frac{1}{(1+f)}$ and $0$  otherwise. By comparing with the competitive ratio guarantee of our optimal online algorithm for $f\ge 1$, we conclude that this bound is tight for $f\ge 1$. Note that in this instance, the algorithm is \emph{indifferent} between accepting or rejecting the first random variable. In \Cref{sec:three-variable-example}, we build on this idea to construct a harder instance with $n=3$ random variables for $0<f<1$  where any algorithm achieves an even worse competitive ratio (in fact, exactly the competitive ratio established in \eqref{eq:cr-3} for $f\in[\tfrac{1}{3},1]$. Then, in Sections~\ref{sec:indifference} and \ref{sec:general-instance}, we construct a general worst-case instance for every $f$ using similar ideas and our solution to the differential equation, completing the proof of \cref{thm:WorstCase}.


\subsection{Warm-up: Hard instance with three random variables}
\label{sec:three-variable-example}
Let $0<f<1$. Consider the following instance with $n=3$ random variables and parameter $x \geq (1+f)$:
\begin{align*}
    X_1 = 1, \;\;\;\;\;\; X_2= \left\{
     \begin{array}{@{}l@{\thinspace}l}
       x \;\;\;\;\; & \text{w.p.} \;\; \frac{1}{x} \\
       0 & \text{o.w.} \\
     \end{array}
    \right. ,
     \;\;\;\;\;\; X_3 = \left\{
     \begin{array}{@{}l@{\thinspace}l}
       x(1+f) \;\;\;\;\; & \text{w.p.} \;\; \frac{x-1-f}{(1+f)(x-1)} \\
       0 & \text{o.w.} \\
     \end{array}
   \right. ~.
\end{align*}
The values and probabilities are chosen so that the optimal online algorithm remains indifferent between accepting $X_i$ or rejecting it, assuming it has already accepted $X_{i-1}$ for $i \in \{2,3\}$. The expected value of the prophet benchmark in this instance can be calculated as follows:
$$\EX[X_{\max}] = \frac{f}{1+f} + \frac{xf}{(1+f)(x-1)} + \frac{x(x-1-f)}{(x-1)} = x+ f\frac{x-1-xf}{(x-1)(1+f)}.$$
To determine the expected reward for the optimal online algorithm, we can solve the recurrence described in \eqref{eq:recurr} to show that $\Phi_0(0) = x-f$. By setting $x\geq (1+f)$ to minimize the competitive ratio, we conclude that no online algorithm can exceed the competitive ratio
\begin{align}\label{eq:cr3variables}
    \min_{x\geq (1+f)} \frac{x-f}{x+ f\frac{x-1-xf}{(x-1)(1+f)}} = \frac{(1+f)(\sqrt{f(2-f)}+1)}{(1+f)\sqrt{f(2-f)}+3f+1}~,
\end{align}
attained at $x^* = \frac{f+2+\sqrt{f(2-f)}}{2}$ (Note that $x^*>(1+f)$ if $f<1$). The corresponding instance is:
\begin{align*}
X_1 = 1, &
\;\;X_2 = \left\{
\begin{array}{ll}
\frac{f+2+\sqrt{f(2-f)}}{2} & \text{w.p.}\;\; \frac{2}{f+2+\sqrt{f(2-f)}} \\
0 & \text{o.w.}
\end{array}
\right.,\;\;
X_3 = \left\{
\begin{array}{ll}
\frac{f+2+\sqrt{f(2-f)}}{2}(1+f) & \text{w.p.}\;\; \frac{1-\sqrt{f(2-f)}}{(1+f)(1-f)} \\
0 & \text{o.w.}
\end{array}
\right.~.
\end{align*}
Notice that this instance is well defined only when $\frac{1-\sqrt{f(2-f)}}{(1+f)(1-f)} \geq 0$, which holds when $f < 1$. One can verify that \eqref{eq:cr3variables} is strictly less than our previous $\frac{1+f}{1+2f}$ bound when $f < 1$, and matches $\frac{1+f}{1+2f}$ at $f=1$. Comparing this upper bound with the competitive ratio guarantee of our algorithm in \eqref{eq:cr-3} confirms the tightness of the competitive ratio in \eqref{eq:cr3variables} for $\frac{1}{3} \leq f \leq 1$.

The above instances with 2 and 3 random variables suggest that as $f$ decreases, we may need more random variables to enforce the worst-case competitive ratio, and that all such instances might satisfy a certain indifference property. Expanding on these observations, we now turn to instances with potentially more than three variables and aim to find the hardest possible instance for every $f$ that enforces a similar indifference condition to that of the instances with $n=2$ or $n=3$ random variables.


\subsection{Instances with indifference}
\label{sec:indifference}
Let $X_1,...,X_n$ be $n\in\mathbb{N}$ scaled Bernoulli random variables, where $X_i=x_i\,\mathrm{Ber}(p_i)$ for $x_i\geq 0 $ and $p_i\in[0,1]$. To imitate the example with $n=3$ variables, we would like to construct instances with the property that, holding a random variable, the optimal online algorithm is indifferent between accepting and rejecting the next. Informally, suppose the algorithm is currently holding $x_{i-1}$. Accepting $x_i$ is beneficial only if all future variables are inactive, and the resulting benefit is:
\[
(x_i - (1+f)\,x_{i-1}) \prod_{j=i+1}^n (1-p_j).
\]
On the other hand, accepting $x_i$ is costly if any subsequent variable is active; in that case, the cost is:
\begin{align*}
\left(1- \prod_{j=i+1}^n (1-p_j)\right) f \cdot x_i.
\end{align*}
Intuitively, if these two quantities are equal, we obtain an instance of the desired type. This idea is formalized in the following lemma, with the proof deferred to \Cref{apx:proof:lem-indifference}. 

\begin{lemma}[Indifference Condition] \label{lem:indifference}
Consider $n \geq 2$ scaled Bernoulli random variables $X_1,...,X_n$ with $X_i=x_i\,\mathrm{Ber}(p_i)$ for $x_i\geq 0 $ and probability $p_i\in[0,1]$. Letting $x_0= 0$, suppose for $i \in [1, n]$ that:
\begin{align}
\label{eq:indifference-condition}
    (x_i - (1+f)x_{i-1})\prod_{j=i+1}^n (1-p_j)  = \left(1- \prod_{j=i+1}^n (1-p_j)\right) f \cdot x_i.
\end{align} Then, the optimal online algorithm, when holding $x_{i-1}$, is indifferent  between accepting and rejecting $x_i$, i.e., $\Phi_i(x_{i-1}) = \Phi_i(x_i) - f x_{i-1}$ for $i \in [1, n]$. Furthermore, the optimal online algorithm, in expectation, obtains exactly $x_{n} - f \sum_{i=1}^{n-1} x_i$. 
\end{lemma}

\subsection{General hard instance: Proof of \Cref{thm:WorstCase}}
\label{sec:general-instance}
Let function $y_f(\cdot)$ be the solution of \ref{eq:diff-eq}. Recall that \cref{thm:diff-eq-existence} guarantees that for any $f>0$, one can always find a finite integer $n$ and a sequence of points $\left(k_0,k_1,\ldots,k_{n-1}\right)$, such that $k_0 = 0$, $y_f (k_{i+1}) = k_i$ for $0\leq i\leq n-2$,  and $k_{n-1} \geq y_f(1)$\footnote{\revcolor{Recall the role of the parameter $k_1$ in \Cref{defn:differentialequation}.}}. By convention, we also set $k_n=1$. Based on this sequence $\{k_i\}$, and by using \Cref{lem:indifference}, we provide a problem instance on which no online algorithm can achieve a competitive ratio better than $\crf=\frac{1}{2-y_f(1)}$. 


Consider $n$ scaled Bernoulli random variables $X_1,...,X_n$ where $X_i$ is equal to $x_i$ with probability $p_i$ and is $0$ otherwise. Furthermore, let $p_1=1$,  normalize the random variables so that $x_1 = 1$, and implicitly let $x_0 = 0$. Let us choose $p_i$ so that the probability that $X_i$ is the maximum random variable is exactly $k_i - k_{i-1}.$ In this case, it follows that
$p_i = \frac{k_i - k_{i-1}}{k_i}$ for all $1\leq i\leq n.$ Following \cref{lem:indifference}, let us impose the condition \eqref{eq:indifference-condition} on $x_i$'s and $p_i$'s.
Hence, 
\begin{align}\label{xii-1Connection}
x_1=1~~~~~~,~~~~~~2\leq i \leq n:~x_i =\frac{k_i}{k_i-k_1}x_{i-1} = \prod_{j=2}^i \frac{k_j}{k_j-k_1},
\end{align}   
and under this instance, we have:
\begin{align*}
    \text{OPT} = \sum_{i=1}^n \left(k_i- k_{i-1}\right)x_i,~~~\text{ and }~~~\text{ALG} = x_{n} -f\sum_{i=1}^{n-1}x_i~,
\end{align*}
where we use $\textrm{OPT}$ to denote $\EX[X_{\textrm{max}}]$ and $\textrm{ALG}$ to denote the expected (net) reward of the optimal online algorithm. We are now prepared to prove \cref{thm:WorstCase}.
\smallskip

\begin{proofof}{Proof of \Cref{thm:WorstCase}} Since the function $y_f$ is a proper solution to the differential equation \ref{difEq}, as described in \Cref{thm:diff-eq-existence}, and recalling the proof of \Cref{thm:alg-existence}, there must exist a function $\varphi:(0,1]\rightarrow\mathbb{R}$ satisfying all six properties listed in \Cref{sec:reduction-to-diff-eq} when we set $\Theta = \theta_f$. 

First, recall \ref{eq:property6}. Setting $t = k_i$ for $i \in [1,n-1]$ in this property and noting that $\varphi(s)\tau(s) = \frac{f\varphi(s) + \theta_f}{1+f}$ (because of \ref{eq:coverage-property5}), we get:
\begin{align}\label{recurrence}
    k_i^2 \varphi(k_i) = k_1^2\varphi(k_1) + \int_{0}^{k_{i-1}}\frac{f\varphi(s) + \theta_f}{1+f}ds
\end{align}
Moreover, while we cannot set $i = n$, by plugging in $t=1$ in \ref{eq:property6} we get:
\begin{align}\label{eq:extraRecurrence}
    \theta_f (1+k_{n-1}-y(1)) &= k_1^2\varphi(k_1) + \int_{0}^{y(1)}\frac{f\varphi(s) + \theta_f}{1+f}ds + \int_{y(1)}^{k_{n-1}}\frac{f\varphi(s) + \theta_f}{1+f}ds \nonumber \\ &= k_1^2\varphi(k_1) + \int_{0}^{k_{n-1}}\frac{f\varphi(s) + \theta_f}{1+f}ds
\end{align}
From \cref{xii-1Connection}, we know that $-\frac{fx_i}{k_i}=(1+f)(x_{i-1}-x_{i})$. For $i \in [1,n-1]$, let us multiply the left hand side of  \cref{recurrence} by $-\frac{fx_i}{k_i}$, and the right hand side of \cref{recurrence} by $(1+f)(x_{i-1}-x_{i})$. Next, multiply the left hand side of \cref{eq:extraRecurrence} by $x_{n}$ and the right hand side by an equal value of $(1+f)x_{n-1}$. Then, sum all of the resulting equations together. Recalling $y(1) = 2-\frac{1}{\theta_f} = 1+k_n-\frac{1}{\theta_f}$, we get:
\begin{eqnarray*}
    && \theta_f \left(k_{n-1}-k_n+\frac{1}{\theta_f}\right)x_{n} - f\sum_{i=1}^{n-1} k_i \varphi(k_i)x_i \\
    &=& x_{n-1}\int_{0}^{k_{n-1}}(f\varphi(s) + \theta_f) ds - \sum_{i=1}^{n-1} (x_i-x_{i-1})\int_{0}^{k_{i-1}}(f\varphi(s) + \theta_f) ds  \\
    &=& \sum_{i=1}^{n-1} x_i \int_{k_{i-1}}^{k_{i}} (f\varphi(s) + \theta_f) ds \overset{(a)}{=}\sum_{i=1}^{n-1} x_i \left(f-fk_{i}\varphi(k_i)+ \theta_f (k_i-k_{i-1})\right),
\end{eqnarray*}
where in (a), we used \ref{eq:phi-after-k1-property4}. We conclude the following, :
\[    \text{ALG} = x_n - f\sum_{i=1}^{n-1} x_i = \theta_f \sum_{i=1}^n x_i(k_i-k_{i-1}) = \theta_f \cdot \text{OPT},\]
which finishes the proof. \end{proofof}

\section{A Simple Order-Agnostic Optimal Competitive Algorithm}
\label{sec:algorithm}
In the previous sections, we essentially developed an algorithm for the prophet inequality with buyback problem on monotonic scaled-Bernoulli instances with small activation probabilities. This algorithm relies on our solution to the problem of maximizing the effective probability of accepting Bernoulli random variables---the analog of the online contention resolution scheme problem for the buyback model---on a continuum (\Cref{alg:quantile-selection}). By applying our ``variable splitting'' reduction (detailed in the proof of \Cref{lem:continuous-reduction}), we could extend this algorithm to obtain another algorithm that works for any monotonic scaled-Bernoulli instance (which could alternatively be derived by directly solving \ref{eq:discrete-LP}). Moreover, using the sequence of reductions from \Cref{sec:monotone reduction}, we could further transform the algorithm into one that applies to general instances. 

As a result of these reductions, in addition to the optimal online policy---which is optimal competitive by definition, but requires full knowledge of the sequence of distributions---we obtain another optimal competitive online algorithm. However, it remains unclear whether this algorithm can be implemented explicitly in polynomial time using the constructs from \Cref{sec:upper-bound}. Moreover, it is not yet clear whether such an implementation can be made \emph{order-agnostic}, relying only on the set of distributions rather than their arrival order.

In this section, we describe a simple and order-agnostic polynomial-time algorithm, \cref{alg:order-oblivious}, that obtains the optimal competitive ratio for a general instance of the problem. Throughout this section, we assume that the instances consist of a sequence of distributions $\cD_1,\ldots,\cD_n$ of independent nonnegative random variables $X_1,\ldots,X_n$, with $\cD_i$ having a cumulative distribution function $F_i$. We also denote the CDF of $X_{\max} = \max_{1 \leq i \leq n} X_i$ by $F(x) = \prod_{i=1}^{n} F_i(x)$.

\subsection{Intuition behind the design of \cref{alg:order-oblivious}}

An issue with general instances---even when the random variables are Bernoulli---is that the values may be observed in an arbitrary order, in which case our analysis of \Cref{alg:quantile-selection} from \cref{sec:upper-bound} does not apply. Intuitively, seeing a higher value earlier is advantageous, since we immediately know that any lower value appearing afterwards is useless. This suggests that designing an algorithm for general instances should not be harder than for monotonic instances, a fact formally captured by the reductions in \Cref{sec:monotone reduction}. However, implementing this idea requires substantial technical care.

To illustrate, consider a sequence of non-monotonic Bernoulli variables. If we observe a variable $X_i = v_i$ earlier than it would appear in the monotonic order, we can \emph{simulate} what the algorithm would have done had it seen a monotonic sequence before observing this variable. It might have chosen to accept a lower value $X_j = v_j < v_i$, though we do not yet know if such a variable will appear. Assuming that we are currently holding value $z < v_j$, we can simulate accepting $X_j$ by accepting $X_i = v_i$ with probability $\frac{v_j - z}{v_i - z}$. In expectation, this mirrors the behavior of the original algorithm.


While this idea can handle the case of non-monotonic Bernoulli variables, to handle general instances (possibly with continuous distributions), we need a technique to appropriately \emph{embed} realized values from any sequence of distributions into a sequence of quantiles. For this, we use the formalism of the \emph{Poisson point process (PPP)}, a probabilistic model in which points on the real line are activated according to a given density. In our framework, the PPP represents the observed values, or possibly their quantiles with respect to $X_{max}$.

\smallskip
Our design of \cref{alg:order-oblivious}, formalizes and combines the above high-level ideas in several steps:
\begin{itemize}
    \item \textbf{\Cref{sec:alg-overview-ppp}:} We start by interpreting \cref{alg:quantile-selection} as an algorithm that acts on a particular kind of non-homogeneous PPP; one that captures the quantiles of the observed variables.
    \item \textbf{\Cref{sec:alg-overview-embedding}:} Next, we explain how based on an observed \revcolor{sequence} of random variables $X_i$, we can generate an ``online PPP'' consistent with the distribution of the full instance, such that at each time the maximum value in the PPP is the maximum value $X_i$ observed so far.
    \item \textbf{\Cref{sec:alg-overview-threshold}:} Next, we explain a modification to \cref{alg:quantile-selection} that ensures that it always makes decisions to accept points based only on a comparison with a certain threshold (note that this is true for \cref{alg:quantile-selection} as stated when $t \geq k_1$, but not for $t<k_1$). While this modification is not strictly speaking necessary, it simplifies the design of \cref{alg:order-oblivious}.
    \item \textbf{\Cref{sec:alg-overview-flagging}:} Finally, we describe \cref{alg:order-oblivious} which can be applied to a general instance of the prophet inequality with buyback problem, by using our online PPP to simulate the behavior of \cref{alg:quantile-selection} (with the modification described just above).
\end{itemize}

\subsubsection{Scale-invariant Poisson point process}
\label{sec:alg-overview-ppp}
We begin by recalling that in our previous discussion of \cref{alg:quantile-selection}, we worked with a continuous limit of Bernoulli random variables defined on the interval $(0, 1]$, where a random variable corresponding to the interval $[t, t + dt]$ is active independently with probability $\frac{dt}{t}$. We can formalize this by viewing the set of active points as a Poisson point process (PPP) on $(0,1]$ with rate $\lambda(t) = \frac{1}{t}$, or equivalently, with intensity measure $\Lambda((t,1]) = -\log t$. Such a PPP is referred to as a scale-invariant Poisson point process~\citep{arratiascaleinvariant}.\footnote{This process has been considered previously in optimal stopping problems; see, for instance, the work of \citet{gnedin2021beat1estrategybestchoice}.} \revcolor{A formal discussion of PPPs is provided in \cref{sec:apx-PPP}. We would like to reiterate that an interpretation of this PPP is the activation process of a monotonic sequence of Bernoulli random variables with vanishing (i.e., infinitely small) activation probabilities, as described in \Cref{sec:reduction-to-diff-eq}. Here, we map each arriving Bernoulli to its corresponding quantile in the distribution of $X_{max}$ and add a point if the arriving Bernoulli turns out to be active. (Note that in this way, the distribution of the last active quantile $t$ would be uniform, which is consistent with the quantile of a sample from $X_{max}$'s distribution).}

\revcolor{With the above interpretation, \cref{alg:quantile-selection} can now be viewed as a procedure that observes points arising from a scale-invariant PPP on $[0,1]$ and accepts or swaps these points in such a way that each active point $t$ is accepted with an effective probability equal to $\Theta$ times the probability that $t$ is the maximum point of the PPP. (Recall that the effective probability refers to the probability of accepting the point minus $(1+f)$ times the probability of discarding the point.)}

\subsubsection{Online PPP generation}
\label{sec:alg-overview-embedding}
An algorithm that aims to use \cref{alg:quantile-selection} to solve a general instance of the problem must find a way to embed its sequential observations of random variables $X_i$ (arriving in an arbitrary order) into a process whose points are observed in a monotonic order. To achieve this, we generate a PPP in an online fashion such that it globally agrees with the distribution of $X_{max}$ and its maximum at each time captures the maximum observed value so far, $\max_{1 \leq j \leq i} X_j$. This embedding allows us to make accept-or-swap decisions based on this evolving maximum value and thereby leverage \cref{alg:quantile-selection}, which expects to observe values in a monotonic order. 

\revcolor{Intuitively, the algorithm treats the `peaks' in the sequence of $X_i$s as fixed `skeleton points' and dynamically samples auxiliary points from the PPP to fill the gaps between them. Together, the combined set of discrete observations and generated points reconstruct the active elements of a continuous instance, enabling the direct deployment of our continuous quantile-selection framework developed in \Cref{sec:upper-bound}.}

A technical detail is that, while previously we described PPPs supported on the ``quantile space'' $(0,1]$, it is more convenient in the formulation of \cref{alg:order-oblivious} to work directly with the ``value space'' $(0,\infty)$ (noting the straightforward correspondence between these spaces via the CDF $F(x)$). More formally, we embed the random variables $X_i$ into a PPP with intensity measure $\Lambda((t, +\infty)) = -\log F(t)$ on $(0, +\infty)$, where $F$ is the CDF of $X_{\max}$. By the mapping theorem for PPPs \citep{kingman-poisson-processes}, there is a direct correspondence between this PPP and a scale-invariant PPP, allowing points generated in the value-space PPP to be converted to a scale-invariant PPP on $(0, 1]$ by applying the CDF $F$.

The remaining question is how to generate a value-space PPP by adding points in an online fashion, ensuring that at each time the current maximum is consistent with the maximum $X_i$ observed so far, and that the PPP at the end has the desired intensity measure $-\log F(x)$. The cornerstone of our approach is to simultaneously track the current maximum observed value $\hat{x}$ among $X_1, X_2, \ldots, X_{i-1}$, as well as a PPP supported on the interval $(0, \hat{x}]$. When a new random variable $X_i$ arrives, we simply ignore it if $X_i \leq \hat{x}$. On the other hand, if $X_i > \hat{x}$, we explicitly add $X_i$ to the PPP to ensure that $X_{\max}$ appears in the PPP. We then need to generate the portion of the PPP between $\hat{x}$ and $X_i$ by filling in additional points to maintain the intensity measure $-\log F(t)$ of the  point process, as desired.

An initial attempt might be to generate this portion by sampling from a PPP on $(\hat{x}, X_i]$ with intensity measure $\Lambda((t, +\infty)) = -\log F(t) = -\sum_{r=1}^n \log F_r(t)$. However, this approach would clearly generate too many points, exceeding the desired intensity measure (since we already explicitly added $X_i$). Moreover, conceptually speaking, this approach incorrectly ignores the partial information that we have already observed: namely, that $X_j \leq \hat{x}$ for all $j < i$. Intuitively, the random variables $X_j$ for $j < i$ should not contribute to generating the PPP between $\hat{x}$ and $X_i$. The correct approach is therefore to sample from a PPP on the interval $(\hat{x}, X_i]$ with intensity measure $\Lambda((t, +\infty)) = -\sum_{r=i}^n \log F_r(t)$, as formally established in \cref{lem:generated-samples}. This adjustment prevents artificial inflation of the PPP's density and accurately reflects our partial knowledge from prior observations.

\subsubsection{Making \cref{alg:quantile-selection} threshold-based}
\label{sec:alg-overview-threshold}
We note that while \cref{alg:quantile-selection} makes deterministic acceptance decisions for quantiles $t \geq k_1$, it uses randomized decisions for quantiles $t<k_1$: it accepts each active quantile $t<k_1$ with probability $\frac{t\varphi(t)}{1- \int_{0}^t\varphi(s)ds}$. However, for convenience, we prefer to modify the algorithm so that it makes deterministic, threshold-based decisions for $t<k_1$ as well.

One way to convert \cref{alg:quantile-selection} into a threshold-based algorithm is by introducing a \textit{randomized} threshold $\tau(0) \in (0, k_1)$ and then deterministically accepting the first active quantile after this threshold. A calculation shows that the appropriate choice of the randomized threshold $\tau(0)$—to ensure each quantile $t$ is accepted with probability $\varphi(t)dt$—is to sample it from a distribution with CDF $G(t) \coloneqq t\varphi(t)+\int_{0}^t\varphi(s)ds$. See the proof of \cref{lem:quantile-acceptance} for more details; the fact that $G(t)$ is a valid CDF follows from \cref{cor:shrinkingGap}.

\subsubsection{Flagging points and simulating \cref{alg:quantile-selection}}
\label{sec:alg-overview-flagging}
Given the online PPP constructed above, our algorithm (\Cref{alg:order-oblivious}) aims to simulate \cref{alg:quantile-selection} (with an initial randomized threshold $\tau(0)$) on these points. To achieve this, \Cref{alg:order-oblivious} first identifies the values in the value-space PPP corresponding to quantiles that \cref{alg:quantile-selection} would accept. At this point, if our algorithm could simply accept all these values, we would be done; \revcolor{this} holds because the effective acceptance probability for each quantile of $X_{\max}$ would match that of \cref{alg:quantile-selection}, and linearity of expectation would imply a competitive ratio equal to this effective acceptance probability.

However, \cref{alg:order-oblivious} cannot directly accept every PPP point because many of them are ``fictitious'' points added by our own sampling procedure, rather than being ``real'' values corresponding to actual observed variables $X_i$. On the positive side, we know that the maximum PPP point at any given time is a real value, as it corresponds to the maximum observed variable $X_i$ thus far. Therefore, we can simulate accepting a fictitious point by instead accepting a real variable $X_i$ with an appropriate probability. Formally, our algorithm \emph{flags} points corresponding to quantiles accepted by \cref{alg:quantile-selection}, and then simulates the acceptance of these flagged points via accepting actual observed values $X_i$.

As previously discussed (and formalized in \cref{lem:splitting}), if we currently hold a value $z<X_i$ and want to simulate the acceptance of a value $z' < X_i$, we can instead accept $X_i$ with probability $\frac{z'-z}{X_i - z}$ and act in the future as if we had accepted the value $z'$. This simulation approach is valid because our objective is linear in the variables $X_i$ (that we are swapping or holding), and
\[
\mathbb{E}\left[\textrm{value held in future}\right]=
\left(1 - \frac{z' - z}{X_i - z}\right)z  + \left(\frac{z' - z}{X_i - z}\right) X_i =
z + \left( \frac{z' - z}{X_i - z} \right) (X_i - z) = z'.
\]

\vspace{-3mm}
\begin{algorithm}[ht]
\raggedright
\caption{Order-agnostic algorithm}
\label{alg:order-oblivious}
Initialize $\hat{x} \leftarrow 0$ {\small\color{\commentcolor}\tcc{$\hat{x}$ tracks the peak value the algorithm has observed so far}}
Initialize $\hat{z} \leftarrow 0$ {\small\color{\commentcolor}\tcc{$\hat{z}$ tracks the last value the algorithm has flagged at any time}}
Initialize threshold $T$ as $F^{-1}(\tau(0))$ with $\tau(0)$ sampled from a distribution with CDF equal to $G(t) := t\varphi(t)+\int_{0}^t\varphi(s)ds$
{\small\color{\commentcolor}\tcc{$T$ determines whether we flag a new value. This is a well defined CDF on $[0, k_1]$ because of \Cref{cor:shrinkingGap}.}}
\For{$i\in[1:n]$}{
\If{$X_i > \hat{x}$}{

Update $z \leftarrow \hat{z}$

Generate a set of samples $S_i$ from a PPP on $(\hat{x},X_i]$ with intensity measure $-\sum_{r=i}^n\log F_r(t)$

\While{$X_i \geq T$}{
Update $\hat{z} \leftarrow \min\{x\in S_i \cup \{X_i\} | x \geq T\}$ {\small\color{\commentcolor}\tcc{the value $\hat{z}$ takes has now been \textit{flagged}}}

Update $T \leftarrow F^{-1}(\tau(F(\hat{z})))$
}

Accept $X_i$ with probability $\frac{\hat{z}-\sample}{X_i-\sample}$.

Update $\hat{x} \leftarrow X_i$
}
}
\end{algorithm}

\subsection{Analysis of \cref{alg:order-oblivious}}
Now, let us formally analyze \cref{alg:order-oblivious}.
First, we would like to show that the samples generated by the algorithm are distributed in exactly the same way as the quantiles observed by \cref{alg:quantile-selection}: \begin{lemma} \label{lem:generated-samples}
The set of samples generated by \cref{alg:order-oblivious}, together with the observed peaks, when considered as quantiles of $X_{\max}$, form a scale-invariant Poisson point process on $(0, 1]$, i.e., a PPP with intensity measure $\Lambda\left((t, 1]\right) = -\log t$.
\end{lemma}
\revcolor{To aid the reader, we present an intuitive sketch for this result, and defer the technical proof to \Cref{sec:generated-samples} at the end of this section.} Here, we will provide an intuitive proof sketch.

\smallskip
\begin{proofof}{Proof sketch}
First, note that a PPP with intensity measure $-\sum_{r=i}^n \log F_r(t)$ can be alternatively generated by separately generating independent PPPs with intensity measures $-\log F_r(t)$ for each $r \in \{i, \dots, n\}$, and then taking the union of these generated points. Thus, we can fix an $r$ and focus only on the samples generated from the $-\log F_r(t)$ portion of the PPP during a run of \cref{alg:order-oblivious}. We then show that these points have the same distribution as a PPP $S_r$ on $(0, \infty)$ with intensity measure $\Lambda((t, +\infty)) = -\log F_r(t)$.

During a run of \cref{alg:order-oblivious}, the samples from the $-\log F_r(t)$ portion of the PPP are generated in two stages: In the first stage, a PPP with intensity $-\log F_r(t)$ is generated on the interval $(0, \hat{x}]$, where $\hat{x} = \max(X_1,\dots,X_{r-1})$. In the second stage, $X_r$ is observed. If $X_r \leq \hat{x}$, it is ignored. If $X_r > \hat{x}$, we explicitly add the point $X_r$ along with a PPP on the interval $(\hat{x}, X_r)$ with intensity $-\log F_r(t)$ to the set of points generated in the first stage.

Observe that the distribution of $X_r$ matches the distribution of the last sampled point $\tilde{X}_r$ from $S_r$. Consequently, the probability that $X_r > \hat{x}$ equals the probability that $\tilde{X}_r > \hat{x}$. Therefore, we can couple samples of $S_r$ where $\tilde{X}_r \leq \hat{x}$ with the points generated during the first stage when $X_r \leq \hat{x}$. Likewise, we can couple samples of $S_r$ where $\tilde{X}_r > \hat{x}$ with the combined points generated in stages one and two when $X_r = \tilde{X}_r$. The desired conclusion follows.
\end{proofof}

Next, we would like to formally establish the link between the values \cref{alg:order-oblivious} actually flags, and the quantiles \cref{alg:quantile-selection} accepts. \revcolor{We say quantile ``$t$ is sampled'' if $F^{-1}(t)$ is one of the values added by the \Cref{alg:order-oblivious} to the PPP it generates in an online fashion.} For each value $v$ \cref{alg:order-oblivious} flags, we can consider the associated quantile of $X_{\max}$, $F(v)$. Then:
\begin{lemma} \label{lem:quantile-acceptance}
In the quantile space of $X_{\max}$, the probability that the algorithm flags a  quantile $t$ conditioned on that quantile being sampled by the algorithm is exactly $t\varphi(t)$.
\end{lemma}
\begin{proofof}{Proof}
Throughout the proof, we often refer to quantities such as \[\Pr (\text{flagging quantile } t| t \text{ is sampled}),\]
which correspond to \Cref{alg:order-oblivious}. Here, we are conditioning on an event with zero probability, but this is well-defined for our particular point processes. Formally speaking, these probabilities are \textit{Palm probabilities} related to a \textit{reduced Palm distribution}~\citep{palm1943variation}, but an understanding of this theory is not necessary to follow our argument. As another technical note, the distribution of a scale-invariant PPP conditioned on $t$ being sampled is the same as the distribution of a scale-invariant PPP. Informally, this is because in a PPP, each point in the interval belongs to the PPP independently. Formally speaking, this is a consequence of the \textit{Slivnyak-Mecke theorem} \citep{slivnyak1962some,mecke1967stationare}, but once again, knowing this theorem is not necessary to follow our proof.\footnote{We refer the curious reader to \cite{penrosepoisson}, and a related tutorial by \citet{coeurjolly2017tutorial}, for more details around these mathematical formalities.} 

Now, note that for $t \in (0, k_1]$ we have that $\Pr (\text{flagging quantile } t| t \text{ is sampled})$ is equal to: 
\[ 
    \int_{0}^t\Pr \left(\text{no sampled $ s \in (\tau_0,t)$ }|t \text{ is sampled}, \tau(0) = \tau_0\right)
    dG(\tau_0).\]
(Recall that our initial threshold $\tau(0)$ is a sample from a distribution with CDF $G(t) = t\varphi(t)+\int_{0}^t\varphi(s)ds$.) The probability that no quantile in  $(\tau_0, t)$ is sampled for a scale-invariant PPP is exactly $\frac{\tau_0}{t}$, and $\tau(0)$ is independent of the realization of the PPP. Therefore, the above is equal to:
\begin{align*}
    \int_{0}^t \frac{\tau_0}{t}dG(\tau_0)
     =& 
    \int_{0}^t \frac{\Pr (t>\tau(0)>\tau_0)}{t}d \tau_0
    \
    =
    \int_{0}^t \frac{ t\varphi(t)-\tau_0\varphi(\tau_0)+\int_{\tau_0}^t\varphi(s)ds}{t}d \tau_0 \\
    =& 
    t\varphi(t) -\int_{0}^t \frac{\tau_0\varphi(\tau_0)}{t}d \tau_0+\frac{\int_{0}^t\int_{\tau_0}^t\varphi(s)ds d \tau_0 }{t}
     =
    t\varphi(t) -\int_{0}^t \frac{\tau_0\varphi(\tau_0)}{t}d \tau_0+\frac{\int_{0}^t\int_{0}^s\varphi(s)d \tau_0 d s }{t} \\
    =& 
    t\varphi(t) -\int_{0}^t \frac{\tau_0\varphi(\tau_0)}{t}d \tau_0+\int_{0}^t\frac{s\varphi(s)}{t}ds\\ 
    =& t\varphi(t).
\end{align*}
If we define $t\varphi_0(t) = \Pr (\text{flagging quantile } t| t \text{ is sampled}) $, we observe that $\varphi_0(t) = \varphi(t)$ for $t \in (0, k_1]$. Furthermore, for $t \in [k_1,1]$ we have:
\begin{align*}
t\varphi_0(t) &= \Pr (\text{flagging quantile } t| t \text{ is sampled})\\
&= \Pr \left(\text{not flagging any quantile } s \in (y(t),t)\right | t \text{ is sampled}) \\
&= 1 - \Pr \left(\text{flagging a quantile } s \in (y(t),t) | t \text{ is sampled}\right) \\
&= 1 - \Pr \left(\text{flagging a quantile } s \in (y(t),t)\right) \\
&= 1 - \int_{y(t)}^t \Pr \left(\text{flagging quantile } s | s \text{ is sampled} \right) \frac{ds}{s}  
= 1 - \int_{y(t)}^t \varphi_0(s)ds.
\end{align*}
The  second-to-last equation comes from the fact that the events of flagging any quantile within the interval $(y(t),t)$ are mutually exclusive, and the fact that the rate of the PPP is $\lambda(t) = 1/t$.

Notice that this equation is exactly the same as the equation $\varphi(t)$ satisfies, and we need to show that $\varphi_0(t) = \varphi(t)$ for $t \in (0, 1]$. We already know that $\varphi_0(t) = \varphi(t)$ for $t \in (0, k_1]$. Suppose by way of contradiction that the supremum of $t$ such that $\varphi_0(s) = \varphi(s)$ for all $s \in (0, t]$ is $t_0 < 1$. Now if $y(t_0 + \epsilon) < t_0$, then, for $s < t_0 + \epsilon$, we have
\[s\varphi_0'(s) + 2\varphi_0(s) = \varphi_0(y(s))y'(s) = \varphi(y(s))y'(s) = s\varphi'(s) + 2\varphi(s),\]
so $\frac{d}{ds}\left(s^2(\varphi_0(s) - \varphi(s))\right) = 0$, from which we immediately conclude that $\varphi_0(s) = \varphi(s)$ for $s < t+ \epsilon$. This is a contradiction, which finishes the proof.
\end{proofof}
Since we have now shown that the values \cref{alg:order-oblivious} flags are essentially the same as the quantiles accepted by \cref{alg:quantile-selection} (after applying $F$), we must next link the performance of the algorithm to the values it flags. Establishing this link is the goal of \cref{lem:alg-expectation}.

\begin{lemma}\label{lem:alg-expectation}
$\mathbb{E}[\text{ALG}] \geq \mathbb{E}\left[\max_{z \in S'} z\right]-f\,\mathbb{E}\left[\sum_{z \in S'} z - \max_{z \in S'} z \right]$, where $\text{ALG}$ is the realized net reward of the algorithm and $S'$ is the (random) set of values the algorithm flags.  
\end{lemma}
\begin{proofof}{Proof}
First, we define the random variables $\hat{X}_i$ and $\hat{Z}_i$. The variable $\hat{X}_i$ represents the largest value accepted by the algorithm before observing $X_i$, and $\hat{Z}_i$ represents the largest value flagged by the algorithm before observing $X_i$. (If no values have been accepted or flagged, respectively, these variables equal $0$.) For convenience, we set $\hat{X}_0 = \hat{Z}_0 = 0$.

Now, note that:
\[\text{ALG} = \hat{X}_n - f\sum_{i=0}^{n-1}\indic(\hat{X}_{i+1} > \hat{X}_i)\hat{X}_i \geq \hat{X}_n - f\sum_{i=0}^{n-1}\indic(\hat{Z}_{i+1} > \hat{Z}_i)\hat{X}_i~,\]
since if $\hat{Z}_{i+1} = \hat{Z}_i$, then $\hat{X}_{i+1} = \hat{X}_i$. Next, note that in case we have $\hat{Z}_{i+1} > \hat{Z}_i$, then we have:
\begin{align*}
    \hat{X}_{i+1} &= \begin{cases} 
    X_{i+1} & \text{ with probability } \frac{\hat{Z}_{i+1}-\hat{Z}_i}{X_{i+1}-\hat{Z}_i} \\
     \hat{X}_{i}  & \text{ otherwise}~.
   \end{cases}
\end{align*}
Now fix an instantiation of the randomness related to the values of $X_i$, the numbers generated by the algorithm in the PPP, and $\tau(0)$, and call this instantiation $\mathcal{F}$. Conditioned on $\mathcal{F}$, the $X_i$ and $\hat{Z}_i$ are all deterministic. By induction on $j$, we will show that for $j \in[0, n]$  we have that $\mathbb{E}[\hat{X}_j | \mathcal{F}] = \hat{Z}_j$. The base case is true since $\hat{X}_0=\hat{Z}_0 = 0$. Let us now assume that the statement is true for $j = i$. If $\hat{Z}_{i+1} = \hat{Z}_{i}$, we already said that $\hat{X}_{i+1} = \hat{X}_i$, so there is nothing to prove. Otherwise, we have:
    \begin{align*}
        \mathbb{E}[\hat{X}_{i+1} | \mathcal{F}]  =& \mathbb{E}\left[\frac{\hat{Z}_{i+1}-\hat{Z}_i}{X_{i+1}-\hat{Z}_i}X_{i+1} + \left(1-\frac{\hat{Z}_{i+1}-\hat{Z}_i}{X_{i+1}-\hat{Z}_i}\right)\hat{X}_{i}\middle| \mathcal{F}\right] = \frac{\hat{Z}_{i+1}-\hat{Z}_i}{X_{i+1}-\hat{Z}_i}X_{i+1} + \left(1-\frac{\hat{Z}_{i+1}-\hat{Z}_i}{X_{i+1}-\hat{Z}_i}\right)\hat{Z}_{i} = \hat{Z}_{i+1}.
    \end{align*}
We conclude that:
\[
    \mathbb{E}[\text{ALG}|\mathcal{F}] \geq  \hat{Z}_n - f \sum_{i=0}^{n-1} \mathbbm{1}(\hat{Z}_{i+1} > \hat{Z}_{i})\hat{Z}_{i}.
\]
Our result simply follows by taking the expectation of the above inequality over $\mathcal{F}$.
\end{proofof}

We are finally ready to establish the competitive ratio of \cref{alg:order-oblivious}:
\begin{theorem}
    [Competitive Ratio of \cref{alg:order-oblivious}]
 \label{thm:alg-CR}
    \cref{alg:order-oblivious} obtains a competitive ratio of at least $\theta_f = \frac{1}{2-\yf(1)}$, and therefore is an optimal competitive online algorithm.
\end{theorem}
\begin{proofof}{Proof}
Note that by \cref{lem:alg-expectation}, it follows that:
\begin{eqnarray*}
\mathbb{E}[\text{ALG}] &\geq& \mathbb{E}[\max_{z \in S'} z]-f\,\mathbb{E}\left[\sum_{z \in S'} z - \max_{z \in S'} z\right]  \\
&=& \mathbb{E}\left[\sum_{z \in S'} z\right]-(1+f)\,\mathbb{E}\left[\sum_{z \in S'} z - \max_{z \in S'} z\right]  \\
&=& \int_{0}^1F^{-1}(t)\Pr (\text{flagging quantile } t| t \text{ is sampled}) \frac{dt}{t}\\ 
&&- (1+f)\int_{0}^1F^{-1}(t) \Pr (\text{unflagging quantile } t| t \text{ is sampled}) \frac{dt}{t}.
\end{eqnarray*}
The last equation follows from the fact that the probability of flagging each quantile can be expressed by conditioning on the event of it being sampled. It follows from \cref{lem:generated-samples} that the probability  of there being a sampled quantile $>\tau(t)$ conditional on the quantile $t$ being sampled is exactly $1-\tau(t)$. Using \cref{lem:quantile-acceptance} and \eqref{eq:coverage-property5}, the above is then equal to:
\begin{eqnarray*}
&& \int_{0}^1F^{-1}(t)\left(t\varphi(t)- (1+f) t\varphi(t)(1 - \tau(t))\right) \frac{dt}{t} = \int_{0}^1 \theta_f F^{-1}(t)dt = \theta_f \mathbb{E}[X_{\max}],
\end{eqnarray*}
as desired.
This completes our analysis of \cref{alg:order-oblivious}. 
\end{proofof}
Finally, we note that \Cref{alg:order-oblivious} can be implemented in polynomial time.  The discussion of the running time and this implementation is deferred to \Cref{sec:running-time}.

\subsection{\revcolor{Proof of \cref{lem:generated-samples}}}
\label{sec:generated-samples}
First of all, we more formally describe the set of samples and peaks generated by \cref{alg:order-oblivious} as the set $S$ returned by \cref{alg:PPP-generation}.

\begin{algorithm}[ht]
\raggedright
\caption{PPP generation}
\label{alg:PPP-generation}
Initialize $\hat{x} \leftarrow 0$ {\small\color{\commentcolor}\tcc{$\hat{x}$ tracks the peak value the algorithm has observed at any time}}
Initialize $S \leftarrow \emptyset$ {\small\color{\commentcolor}\tcc{$S$ is the PPP we are creating}}
\For{$i\in[1:n]$}{
\If{$X_i > \hat{x}$}{
Generate a set of samples $S_i$ from a PPP on $(\hat{x},X_i]$ with intensity measure $-\sum_{r=i}^n\log F_r(t)$

Update $S \leftarrow S \cup \{X_i\} \cup S_i$.

Update $\hat{x} \leftarrow X_i$

}
}
Return $S$
\end{algorithm}

Second, we have the following simple algebraic proposition, useful in the proof of \cref{lem:generated-samples}.

\begin{proposition}\label{prop:nestedGeometric}
For all $m,n \geq 0$ and formal variables $x_{r, j}$ (for $r \in [1:n]$ and $j \in [1:m]$) we have:
    \begin{align*}
        \sum_{0 \leq i_1 \leq \cdots \leq i_m \leq n} \prod_{j=1}^m \left(\left(1-\indic (i_{j+1} > i_j) \propPar_{i_j+1,j}\right) \prod_{r=1}^{i_j}\propPar_{r,j} \right) = 1,   
    \end{align*}
with the interpretation that $i_{m+1} = n$.
\end{proposition}
\begin{proofof}{Proof}
We prove this by induction on $m$. The base case, $m=0$, holds vacuously. Suppose the statement is true for $m = k -1$. To prove the statement for $m = k$, notice:
\begin{align*}
        & \sum_{0 \leq i_1 \leq \cdots \leq i_k \leq n} \prod_{j=1}^k \left(\left(1-\indic (i_{j+1} > i_j) \propPar_{i_j+1,j}\right) \prod_{r=1}^{i_j}\propPar_{r,j} \right) \\ 
        =& \sum_{i_k = 0}^n \left( \left(\left(1-\indic (n > i_k) \propPar_{i_k+1,k}\right) \prod_{r=1}^{i_k}\propPar_{r,k} \right) \sum_{0 \leq i_1 \leq \cdots \leq i_k} \prod_{j=1}^{k-1} \left(\left(1-\indic (i_{j+1} > i_j) \propPar_{i_j+1,j}\right) \prod_{r=1}^{i_j}\propPar_{r,j} \right) \right) \\ 
        =& \sum_{i_k = 0}^n \left(\left(1-\indic (n > i_k) \propPar_{i_k+1,k}\right) \prod_{r=1}^{i_k}\propPar_{r,k} \right) \cdot 1 \\
        =& \sum_{i_k = 0}^n \prod_{r=1}^{i_k}\propPar_{r,k} - \sum_{i_k = 1}^n \prod_{r=1}^{i_k}\propPar_{r,k} = 1,
\end{align*}
where the penultimate equation comes from the fact that the inner summation is equal to $1$ using the induction hypothesis with $n = i_k$ and $ m= k-1$.
\end{proofof}
We are ready to prove \cref{lem:generated-samples}, that when the points in $S$ are converted to quantiles of $X_{\max}$ we obtain a PPP on $(0, 1]$ with intensity measure $\Lambda(t) = -\log t$:

\begin{proofof}{Proof of \cref{lem:generated-samples}}
First, note that by the mapping theorem for PPPs \revcolor{\citep[][Section~2.3]{kingman-poisson-processes}}, this is the same as showing that the distribution of samples generated during a run of \cref{alg:PPP-generation} is exactly the same as the distribution of a PPP on $(0, \infty)$ with intensity measure $\Lambda((t, \infty)) = -\log F(t)$. Now suppose we are given $0 < s_1 < t_1 < s_2 < t_2 < \cdots < s_k < t_k$. We can calculate the probability that a PPP with intensity measure $-\log F(x)$ generates no points in $B = (s_1, t_1]\cup \cdots \cup (s_k, t_k]$ as follows:
\[\Pr(\text{no points in } B) = \prod_{j=1}^ke^{-\log F(t_j) + \log F(s_j)} = \prod_{j=1}^k\frac{F(s_j)}{F(t_j)}.\]

Now by Rényi's theorem for PPPs \revcolor{\citep[][Section~3.4]{kingman-poisson-processes}}, it suffices to show that the probability that $S$ contains no points in $B$ is exactly $\prod_{j=1}^k\frac{F(s_j)}{F(t_j)}$.\footnote{There is one little wrinkle--as stated in the cited book, the result only applies to intensity measures that are finite on bounded sets. But we can generate a scale-invariant PPP from a homogeneous PPP by applying the map $t \to e^{-t}$, so the result we need follows from the result for the homogeneous case.} The probability that a PPP on $(x',x]$ with $x' < s < t < x$ and intensity measure $-\sum_{r=i}^n\log F_r(t)$ generates no points in $(s, t]$ is exactly $\prod_{r = i}^n \frac{F_r(s)}{F_r(t)}$. Hence:
\begin{align}\label{eq:noPointinB}
    \Pr(S \cap B = \emptyset) = \sum_{0 \leq i_1 \leq \cdots \leq i_k \leq n} \Pr \left[\max(X_1,...,X_{i_j}) < s_j < t_j < X_{i_j+1} \;\forall j \right]\prod_{j=1}^k \prod_{r=i_j+1}^n \frac{F_r(s_j)}{F_r(t_j)},
\end{align}
and we would like to show this equals $\prod_{j=1}^k\prod_{r = 1}^n \frac{F_r(s_j)}{F_r(t_j)}$. Multiplying the right hand side of \ref{eq:noPointinB} by $\prod_{j=1}^k \prod_{r=1}^n \frac{F_r(t_j)}{F_r(s_j)}$, we get:
\begin{align*}
    & \prod_{j=1}^k \prod_{r=1}^n \frac{F_r(t_j)}{F_r(s_j)} \sum_{0 \leq i_1 \leq \cdots \leq i_k \leq n} \Pr \left[\max(X_1,...,X_{i_j}) < s_j < t_j < X_{i_j+1} \;\forall j \right]\prod_{j=1}^k \prod_{r=i_j+1}^n \frac{F_r(s_j)}{F_r(t_j)} \\
     =& \prod_{j=1}^k \prod_{r=1}^n \frac{F_r(t_j)}{F_r(s_j)} \sum_{0 \leq i_1 \leq \cdots \leq i_k \leq n} \prod_{j=1}^k  \left(\left(1-\indic(i_{j+1} > i_j)\frac{F_{i_j+1}(t_j)}{F_{i_j+1}(s_{j+1})} \right) \left(\prod_{r=i_{j-1} +1}^{i_j} F_r(s_j) \right) \right)
     \prod_{j=1}^k \prod_{r=i_j+1}^n \frac{F_r(s_j)}{F_r(t_j)}\\
     =& \sum_{0 \leq i_1 \leq \cdots \leq i_k \leq n} \prod_{j=1}^k \left(1-\indic(i_{j+1} > i_j)\frac{F_{i_j+1}(t_j)}{F_{i_j+1}(s_{j+1})} \right)  \left(\frac{\prod_{r=1}^{i_j} F_r(t_j)}{\prod_{r=1}^{i_{j-1}} F_r(s_j)} \right)\\
     =& \sum_{0 \leq i_1 \leq \cdots \leq i_k \leq n} \prod_{j=1}^k \left(1-\indic(i_{j+1} > i_j)\frac{F_{i_j+1}(t_j)}{F_{i_j+1}(s_{j+1})} \right) \left(\frac{\prod_{r=1}^{i_j} F_r(t_j)}{\prod_{r=1}^{i_{j}} F_r(s_{j+1})} \right) = 1,
\end{align*}
where the last equation follows directly from \cref{prop:nestedGeometric} and setting $\propPar_{r,j} = \frac{F_r(t_j)}{F_r(s_{j+1})}$. Hence:
\begin{align*}
    \Pr(S \cap B = \emptyset) = \prod_{j=1}^k\prod_{r = 1}^n \frac{F_r(s_j)}{F_r(t_j)},
\end{align*}
as desired.
\end{proofof}

\section{Conclusion}
\label{sec:conclusion}
In this paper, we introduced and studied a version of the prophet inequality problem in which cancellation is possible by paying a penalty $f$ times the value of the discarded item. We constructed an order-agnostic and optimally competitive algorithm for this problem for all $f\geq 0$, and calculated its competitive ratio through the solution to a particular differential equation, which we derived explicitly. We also used the same differential equation to find worst-case instances for our problem.
We finally obtained explicit solutions for the differential equation for $f\geq \frac13$, providing simpler and explicit forms for the algorithm and its competitive ratio in this regime.

Our algorithm, which allows cancellations and takes into account distributional information about the
input, obtains a better competitive ratio than the algorithms in the classical prophet inequality setting, and the algorithms in the setting with adversarial inputs. As $f \to 0$, our algorithm's competitive ratio goes to $1$ significantly faster than the competitive ratio of algorithms with adversarial input \revcolor{(see \Cref{fig:competitive-ratios}); when $f$ is near $0$, our algorithm has competitive ratio $1 + \mathcal{O}(f \log f$) while algorithms with adversarial input have competitive ratio $1 - \mathcal{O}(\sqrt{f})$}. Furthermore, our algorithm, unlike the optimal online algorithm, does not need to know the order in which the random variables will arrive in the future.

\paragraph{Future research:} There are several interesting directions for future research. Firstly, is there a closed-form formula for the optimal competitive ratio $\alpha(f)$ for every $f>0$? We conjecture that this is not the case, but perhaps there is a more explicit description of the behavior of $\alpha(f)$ than the one provided in this paper. \revcolor{Secondly, future research could extend our model to more complicated cancellation cost functions, particularly after normalizing the random variables and letting $\mathbb{E}[X_{\max}] = 1$. In this case, we conjecture our techniques extend to affine cancellation costs, where canceling a random variable $X_j$ costs $fX_j +c.$ Thirdly, another interesting direction is to adapt our model to the multi-item prophet prophet inequality, prophet inequality with a matching constraint, matroid constraint, or to the prophet inequality for combinatorial auctions.}  We believe that this model, and our technical framework, might find other interesting applications in revenue management.

\setlength{\bibsep}{0.0pt}
\bibliographystyle{plainnat}
\OneAndAHalfSpacedXI
{\footnotesize
\bibliography{refs}}

\renewcommand{\theHchapter}{A\arabic{chapter}}
\renewcommand{\theHsection}{A\arabic{section}}
\newpage
\ECSwitch
\ECDisclaimer
\section{Solving the Differential Equation}
\label{sec:diff-eq-soln}
Given parameters $f>0$ and $\theta \in [\frac{1}{2}, \frac{1+f}{1+2f}]$, let us denote $k_1 = \cf = \frac{f}{1+f}$.
We aim to solve the following differential equation. 

\begin{itemize}
\item For $t \in [2-\frac{1}{\theta}, 1]$, $y(t)$ is given explicitly as
\begin{equation}
\label{eq:init-segment}
 y(t) = 1 - \frac{1}{\theta} + \kk + \frac{(t-\kk)^2}{1-\kk} 
\end{equation}
 Note that $y(1) = 2 - \frac{1}{\theta}$.
\item For every $t<1$, if $y(t) \geq k_1$ and $y'(t)$ exists, then
\begin{equation}
\label{eq:diff-eq}
y'(y(t)) = \frac{y(t) - \kk}{t - \kk} \left( 2 - \frac{y(t)}{y'(t) (t - \kk)} \right)
\end{equation}
\end{itemize}

We remark that this differential equation demands that for points approaching $2-\frac{1}{\theta}$ from the left, the derivative $y'(t)$ (given by (\ref{eq:diff-eq})) is different from that for points approaching $2-\frac{1}{\theta}$ from the right (given by (\ref{eq:init-segment})). Hence, the one-sided derivatives are not equal; $y'_+(2-\frac{1}{\theta}) \neq y'_-(2-\frac{1}{\theta})$ and $y(t)$ is not differentiable at $t = 2-\frac{1}{\theta}$. More generally, $y(t)$ is not going to be differentiable at points in the form $r_1 = 2-\frac{1}{\theta}, r_2 = y(r_1), r_3 = y(r_2)$, etc. At these points, we can interpret (\ref{eq:diff-eq}) as holding for the left and right derivatives separately; this is implied by taking a one-sided limit in the equation (\ref{eq:diff-eq}). This will not cause any substantial issues but we need to keep in mind that the derivative changes discontinuously at these points.

Since the differential equation is in the form above, where $y'(y(t))$ depends only on $y(t)$ and $y'(t)$, and not on the value $y(y(t))$, we can write the solution (at least on some interval) explicitly as an integral. As above, let $r_0 = 1$ and $r_1 = y(r_0) = 2 - \frac{1}{\theta}$. For $t \in [r_1, r_0]$, we are given the initial segment
$$ y(t) = 1 + \kk - \frac{1}{\theta} + \frac{(t-\cf)^2}{1-\cf} $$
and we can also write
$$ y(y(t)) = y(r_1) - \int_{y(t)}^{r_1} y'(x) dx = r_2 - \int_{t}^{r_0} y'(y(x)) y'(x) dx $$
which determines $y(t)$ for $t \in [r_2, r_1]$ where $r_2 = y(r_1)$.
Note that the integrand is provided by the differential equation (\ref{eq:diff-eq}):
$$ y'(y(x)) y'(x) = \frac{y(x) - \kk}{x - \kk} \left( 2 y'(x) - \frac{y(x)}{x - \kk} \right) $$
Since \revcolor{$r_1\leq t \leq x \leq r_0$, both} $y(x)$ and $y'(x)$ \revcolor{are} given explicitly by \eqref{eq:init-segment}\revcolor{, which allows us to directly calculate} $y(y(t))$.

Inductively, as long as $y(t)$ is increasing and differentiable for $t \in (r_{k+1}, r_k)$, we  define $r_{k+2} = y(r_{k+1}) < y(r_k) = r_{k+1}$ and 
$$ y(y(t)) = r_{k+2} - \int_{t}^{r_k} \frac{y(x) - \kk}{x - \kk} \left( 2 y'(x) - \frac{y(x)}{x - \kk} \right) dx $$
for $t \in [r_{k+1}, r_{k}]$, which defines $y(t)$ for $t \in [r_{k+2}, r_{k+1}]$. We can continue inductively (going backwards in terms of $t$) as long as $y(t)$ is increasing and $y(t) \geq \kk$.

\begin{figure}
\centering
\begin{tikzpicture}[scale=1.2]

  \begin{axis}[
    domain=0:1, 
    samples=300, 
    axis lines=middle, 
    enlargelimits, 
    clip=false,
  ymin=0.055, ymax=0.61,
  label style={font=\footnotesize},   
  tick label style={font=\footnotesize},
  ]

    \addplot [gray, thick, dashed] [domain=0.2:0.5] {(x-0.25)*(x-0.25) - 0.05 };
    \addplot [gray, thick, dashed] [domain=0.5:0.75] {(x-0.25)*(x-0.25) + 0.1*(x-0.5) - 0.05 };
    \addplot [gray, thick, dashed] [domain=0.75:1.00] {0.275 + 1.2*(x-0.75) - 0.05 };

    \addplot [gray, thick , dashed] [domain=0.2:0.5] {(x-0.25)*(x-0.25) - 0.025 };
    \addplot [gray, thick, dashed ] [domain=0.5:0.75] {(x-0.25)*(x-0.25) + 0.1*(x-0.5) - 0.025 };
    \addplot [gray, thick, dashed ] [domain=0.75:1.00] {0.275 + 1.2*(x-0.75) - 0.025 };

    \addplot [thick] [domain=0.25:0.5] {(x-0.25)*(x-0.25) };
    \addplot [thick] [domain=0.5:0.75] {(x-0.25)*(x-0.25) + 0.1*(x-0.5) };
    \addplot [thick] [domain=0.75:1.00] {0.275 + 1.2*(x-0.75) };

    \addplot [gray, thick, dashed ] [domain=0.20:0.5] {(x-0.25)*(x-0.25) + 0.025 };
    \addplot [gray, thick, dashed ] [domain=0.5:0.75] {(x-0.25)*(x-0.25) + 0.1*(x-0.5) + 0.025 };
    \addplot [gray, thick, dashed ] [domain=0.75:1.00] {0.275 + 1.2*(x-0.75) + 0.025 };

    \addplot [gray, thick, dashed ] [domain=0.2:0.3] {4*(x-0.3)*(x-0.3) + 0.2*(x-0.3) + 0.05 };
    \addplot [gray, thick, dashed ] [domain=0.3:0.8] {(x-0.3)*(x-0.3) + 0.2*(x-0.3) + 0.05 };
    \addplot [gray, thick, dashed ] [domain=0.8:1.00] {0.4 + 1.1*(x-0.8 };

    \addplot [gray, thick, dashed ] [domain=0.25:0.4] {10*(x-0.4)*(x-0.4) + 0.1*(x-0.4) + 0.12};
    \addplot [gray, thick, dashed ] [domain=0.4:0.6] {2*(x-0.4)*(x-0.4) + 0.1*(x-0.4) + 0.12 };
    \addplot [gray, thick, dashed ] [domain=0.6:1.00] {0.22 + 1.05*(x-0.6 };

    \addplot [gray, thick, dashed ] [domain=0.35:0.5] {20*(x-0.5)*(x-0.5) + 0.3*(x-0.5) + 0.2 };
    \addplot [gray, thick, dashed ] [domain=0.5:0.6] {3*(x-0.5)*(x-0.5) + 0.3*(x-0.5) + 0.2 };
    \addplot [gray, thick, dashed ] [domain=0.6:1.00] {0.26 + 1.0*(x-0.6 };

    \draw (60, 60) -- (60, 40);
    \node at (60, -65) {$\cf$};
\draw[->, thick]
  (axis cs:1.03, 0.5) -- (axis cs:1.03, 0.67)
  node[midway, right] {\footnotesize{$\theta~~\text{increases}$}};
  \draw[-, thick, draw=cornellred]
  (axis cs:0.26, -0.08) -- (axis cs:0.26, 0.6);


    \end{axis}

 \node at (0.4,5.8) {$y(t)$};
  \node at (6.5,0.1) {$t$};
  \node at (6,2.9) {$f=1/3$};
  \node at (6,2.1) {$\cf=\tfrac{f}{f+1}=1/4$};
  \node at (6,1.3) {$\theta=?$};
  
\end{tikzpicture}
 \caption{\centering Different possible outcomes of the construction, for different values of $\theta$}
    \label{fig:diff-equ}
\end{figure}

\paragraph{Classification of outcomes.}
Depending on the choice of $\theta$, this process could terminate in several possible ways:
\begin{itemize}
\item The solution $y(t)$ is defined for all $t \in [\kk, 1]$, such that $y(\kk) \leq 0$, $y(t)$ is continuous and increasing on $[\kk,1]$, 
and there is a finite sequence of ``breakpoints'' $1 = r_0 > r_1 > r_2 > \ldots > r_n$, such that $r_{i+1} = y(r_i)$, $r_n < \kk$.

\item The solution is not in the form above, i.e. either it can be defined for $t \in [\kk,1]$ but not in a continuous increasing way, or $y(\kk) > 0$, or the solution breaks down at some point and cannot be extended to $[\kk,1]$.
\end{itemize}

There are several reasons why the solution might break down: At some point, we might obtain $y'(t) = 0$ (note that we divide by $y'(t)$ in the differential equation). 
Or possibly, there is an infinite sequence of breakpoints $r_{i+1} = y(r_i) < r_i$ which never reaches $r_n < \kk$. At this point we do not aim to classify the reasons why the solution might break down; all these cases are considered a failure.

The first case is what we call a ``complete successful solution''. As we prove later, this means that there is an algorithm for the buyback problem achieving a factor of $\theta$.

\begin{defn}
\label{def:complete-successful-sol}
We call $y(t)$ a {\em complete successful solution}, if $y(t)$ is continuous and strictly increasing on $[\kk, 1]$, with $y(\kk) \leq 0$, satisfying (\ref{eq:init-segment}) for $t \in [2-\frac{1}{\theta},1]$  and (\ref{eq:diff-eq}) whenever $y(t) \geq \kk$, except for a finite number $n$ of breakpoints $r_0=1, r_{i+1} = y(r_i) < r_i$ such that $r_n < \kk$.
\end{defn}

A formally weaker notion (which we show in fact to be equivalent), is the following.

\begin{defn}
\label{def:successful-sol}
We call a solution $y(t)$ {\em successful}, if $y(t)$ is continuous and strictly increasing on $[t_0, 1]$, with $t_0 \geq \kk$, $y(t_0) \leq 0$, satisfying (\ref{eq:init-segment}) for $t \in [2-\frac{1}{\theta},1]$  and (\ref{eq:diff-eq}) whenever $y(t) \geq \kk$, except for a finite number $n$ of breakpoints $r_0=1, r_{i+1} = y(r_i) < r_i$ such that $r_n < t_0$.
\end{defn}



\subsection{Characterization of successful solutions}

The key lemmas that help us distinguish between success and failure are the following.

\begin{lemma}
\label{lem:success-criterion}
Let $\theta \in [\frac12, \frac{1+f}{1+2f}]$ be such that on every interval $[t_1,1]$ where the solution $y(t)$ can be defined, if $y'(t)$ is defined then
$$ y'(t) \geq \frac{y(t)}{t-\kk}. $$
Then $y(t)$ can be defined on $[\kk,1]$ so that it is a complete successful solution. In particular, the recurrence $r_0=1$, $r_{i+1} = y(r_i)$ reaches $r_n \leq \kk$ in $n \leq 1 + \frac{1}{f}$ steps.
\end{lemma}

We note that we always have $y'(t) \geq \frac{y(t)}{t-\kk}$ on the initial segment, $t \in (r_1,r_0)$, because
$$ y'(t) (t-\kk) - y(t) = \frac{(t-\kk)^2}{1-\kk} - 1 - \kk + \frac{1}{\theta} \geq -\frac{1+2f}{1+f} + \frac{1}{\theta} \geq 0. $$
So the assumption is always satisfied at least for $t_1 = r_1 = 2-\frac{1}{\theta}$. However, the solution can typically be extended to $t < t_1$, and that's where the assumption plays a non-trivial role.

\begin{proofof}{Proof of Lemma~\ref{lem:success-criterion}}
Let us inductively define a solution $y(t)$, starting from the initial segment $[r_1,r_0]$ where $r_0 = 1$, $r_1 = y(1) = 2 - \frac{1}{\theta}$. If $2-\frac{1}{\theta} \leq \kk$, we are done since the initial segment already covers the interval $[\kk,1]$ and we have a complete successful solution. (This happens for $f \geq 1$ and $\theta = \frac{1+f}{1+2f}$.)

Otherwise, as long as $y(t)$ satisfies $y'(t) \geq \frac{y(t)}{t-\kk}$ on an interval $(r_{i+1}, r_i)$, $r_{i+1} = y(r_i) > \kk$, we can use the differential equation (\ref{eq:diff-eq}) to extend the solution to another interval $[r_{i+2}, r_{i+1}]$, where $r_{i+2} = y(r_{i+1})$. Note that the assumption together with $(\ref{eq:diff-eq})$ implies that $y'(y(t)) \geq \frac{y(t) - \kk}{t - \kk} > 0$, so for $y(t) > \kk$ we obtain $y'(y(t)) > 0$, and the function is strictly increasing. By induction, we can continue this process and define $y(t)$ on $[r_{i+1}, r_i]$ as long as $r_i > \kk$. 

Inductively, we prove furthermore that $\frac{y(t)}{t-\kk}$ is non-decreasing in $t$, as long as $t>\kk$. On each interval $(r_{i+1}, r_i)$, we can differentiate with respect to $t$:
$$ \frac{d}{dt} \left( \frac{y(t)}{t-\kk} \right) = \frac{y'(t) (t-\kk) - y(t)}{(t-\kk)^2} \geq 0 $$
by assumption. Therefore $\frac{y(t)}{t-\kk}$ is non-decreasing for every $t>\kk$ where the solution can be defined. In particular, we have 
\begin{equation}
\label{eq:slope}
\hspace{100pt}   \frac{y(t)}{t-\kk} \leq \frac{y(1)}{1-\kk} = (1+f) \left( 2-\frac{1}{\theta} \right) \leq 1 
\end{equation}
using the assumption $\theta \leq \frac{1+f}{1+2f}$ in the last inequality. Therefore, for every $t$ where the solution can be defined, we have $y(t) \leq t - \kk$. In at most $\lfloor \frac{1}{\kk} \rfloor = \lfloor \frac{1+f}{f} \rfloor$ steps, we reach a point $r_n \leq \kk$ and we terminate. The union of intervals $[r_n, r_{n-1}], \ldots, [r_1,r_0]$ covers $[\kk,1]$ and hence the solution is well-defined and strictly increasing on $[\kk,1]$. 

Finally, we prove that $y(\kk) \leq 0$:  If $y(\kk) > 0$, then by continuity of $y(t)$, $\lim_{t \to \kk^+} \frac{y(t)}{t-\kk} = +\infty$. However, this contradicts the fact that $\frac{y(t)}{t-\kk}$ is non-decreasing for $t>\kk$.
\end{proofof}

In other words, if we try to extend the solution as far as possible, then the solution can fail only if at some point, $y'(t) < \frac{y(t)}{t-\kk}$. 

The next lemma shows a converse to this in some sense: If a solution is successful, then it must indeed satisfy $y'(t) \geq \frac{y(t)}{t-\kk}$; otherwise if at some point we obtain $y'(t) < \frac{y(t)}{t-\kk}$, then the solution fails because the curve at this point is aiming for a value $y(\kk) > 0$ and cannot reach $y(\kk) \leq 0$ anymore. This is essentially because if $y(t)$ is successful, then it must be convex, as the lemma also states. This will be crucial in the analysis.

\begin{lemma}
\label{lem:success-properties}
If $y(t)$ is a successful solution on $[t_0, 1]$, $t_0 \geq \kk$,
then $y(t)$ can be extended to a complete successful solution on $[\kk,1]$, which is moreover convex, satisfies $y(t) < t$ for all $t \in [\kk,1]$, and
$$ y'(t) \geq \frac{y(t)}{t - \kk} $$
for all $t \in (\kk, 1)$ where $y(t) > 0$ and $y'(t)$ exists.
\end{lemma}

Note that this lemma implies that a successful solution for $\theta$ exists if and only if a complete successful solution exists; the notions only differ in the domain where $y(t)$ is defined.

The proof of Lemma~\ref{lem:success-properties} is more involved. We proceed in a sequence of simpler propositions.

\subsection{Monotonicity and convexity properties}

The first proposition states that the assumption of monotonicity and finitely many breakpoints imply that $y(t) < t$.

\begin{proposition}
\label{prop:y-monotonicity}
If $y(t)$ is an increasing function on $[t_0,1]$, with a finite number of breakpoints $r_0 = 1, r_{i+1} = y(r_i) < r_i$ such that $r_n < t_0$, then $y(t) < t$ for all $t \in [t_0,1]$.
\end{proposition}

\begin{proofof}{Proof}
Since $1 > r_1 > r_2 > \ldots t_0 > r_n$, every $t \in [t_0,1]$ is in some interval $(r_{i+1}, r_i]$, and by the monotonicity of $y(t)$,
$$ y(t) \leq y(r_i) = r_{i+1} < t.$$
\end{proofof}

The second proposition states that the solution has monotonic derivatives, when we compare them at $t_2$ and $t_1 = y(t_2) < t_2$.

\begin{proposition}
\label{prop:disc-convex}
If a solution $y(t)$ is defined on $[t_0, 1]$, $\kk \leq t_0 < t_1 < t_2 \leq 1$,  $t_1 = y(t_2)$, $y$ is differentiable at $t_2$, and $y'(t_2) > 0$, then it's also differentiable at $t_1$ and
$$ y'(t_1) \leq y'(t_2). $$
\end{proposition}

\begin{proofof}{Proof}
Since $t_2 > \kk$ and $t_1 = y(t_2) \geq \kk$, we have $\frac{t_1}{t_2 - \kk} > 0$; we also have $y'(t_2) > 0$ and we can write
$$ y'(t_2) = \frac{1}{1-\lambda} \cdot \frac{t_1}{t_2 - \kk} $$
for some $\lambda < 1$. Using the differential equation,
$$ y'(t_1) = y'(y(t_2)) 
= \frac{y(t_2) - \kk}{t_2 - \kk} \cdot \left( 2 - \frac{y(t_2)}{y'(t_2) (t_2 - \kk)} \right)
= \frac{t_1 - \kk}{t_2 - \kk} \cdot \left( 2 - \frac{t_1}{y'(t_2) (t_2 - \kk)} \right). $$
Now we use the expression $y'(t_2) = \frac{1}{1-\lambda} \frac{t_1}{t_2 - \cf}$, to obtain
$$ y'(t_1) = \frac{t_1 - \kk}{t_2 - \kk} (2 - (1-\lambda)) \leq (1+\lambda) \frac{t_1}{t_2 - \kk}. $$
Finally, we observe that $1+\lambda \leq \frac{1}{1-\lambda}$ for all $\lambda < 1$, and hence $$ y'(t_1) \leq \frac{1}{1-\lambda} \cdot \frac{t_1}{t_2 - \kk} = y'(t_2).$$
\end{proofof}

The third proposition states that as long as $y$ is twice differentiable and convex at $t$, and it satisfies the criterion $y'(t) \geq \frac{y(t)}{t - \kk}$, then $y$ is also increasing and convex at $y(t)$. 

\begin{proposition}
\label{prop:second-diff}
If a solution $y(t)$ is defined on $[t_0, 1]$, $\kk \leq t_0 < y(t) < t < 1$, $y$ is twice differentiable at $t$, and 
$$ y'(t) \geq \frac{y(t)}{t - \kk}, \ \ \ y''(t) \geq 0, $$
then $y$ is also twice differentiable at $y(t)$, $y'(y(t)) > 0$ and $y''(y(t)) > 0$.
\end{proposition}

\begin{proofof}{Proof}
Note that our assumption \revcolor{implies} in particular that $y(t) > \kk$ and $y'(t) > 0$. 
By (\ref{eq:diff-eq}), we further have
$$ y'(y(t)) = \frac{y(t) - \kk}{t - \kk} \left( 2 - \frac{y(t)}{y'(t) (t -\kk)} \right) \geq \frac{y(t) - \kk}{t - \kk}  > 0. $$
Also, the right-hand side is differentiable (since $y$ is twice differentiable at $t$), and so $y$ is twice differentiable at $y(t)$. We have
\begin{eqnarray*}
\frac{d}{dt} y'(y(t)) & = & \left( \frac{y'(t)}{t - \kk} - \frac{y(t) - \kk}{(t - \kk)^2} \right) \left( 2 - \frac{y(t)}{y'(t) (t - \kk)} \right) \\
& & + \frac{y(t) - \kk}{t - \kk} \left( -\frac{y'(t)}{y'(t) (t - \kk)} + 
 \frac{y(t) y''(t)}{(y'(t))^2(t - \kk)} + \frac{y(t)}{y'(t) (t-\kk)^2}  \right) \\
& = & \frac{2 y'(t)}{t-\kk} - \frac{y(t)}{(t-\kk)^2} - \frac{y(t) - \kk}{(t-\kk)^2}
\left( 3 - \frac{2 y(t)}{y'(t) (t-\kk)} - \frac{y(t) y''(t)}{(y'(t))^2} \right).
\end{eqnarray*}
Let us denote $\alpha = y'(t)$ and $y(t) = \alpha (t - \ell)$ for some $\ell<t$ (this is possible since we assume $\alpha = y'(t)>0$ and $y(t)>0$). Note also that $\alpha = y'(t) \geq \frac{y(t)}{t-\kk} = \frac{\alpha(t-\ell)}{t-\kk}$ which means that $\ell \geq \kk$. 
Then we can write
\begin{eqnarray*}
\frac{d}{dt} y'(y(t)) & = & \frac{2\alpha}{t-\kk} - \frac{\alpha(t-\ell)}{(t-\kk)^2} - \frac{\alpha(t-\ell)-\kk}{(t-\kk)^2} \left( 3 - \frac{2(t-\ell)}{t - \kk} - \frac{(t-\ell) y''(t)}{\alpha}  \right) \\
& \geq & \frac{1}{(t-\kk)^3} \left( 2 \alpha (t-\kk)^2 - \alpha (t-\ell)(t-\kk) - (\alpha(t-\ell)-\kk)(3(t-\kk) - 2(t-\ell)) \right)
\end{eqnarray*}
using the assumption that $y''(t) \geq 0$, so we dropped the last term.
From here,
\begin{eqnarray*}
\frac{d}{dt} y'(y(t)) & \geq & \frac{1}{(t-\kk)^3} \left( 2\alpha(t-\kk)^2 - 4\alpha(t-\ell)(t-\kk) + 2\alpha(t-\ell)^2  + \kk (t + 2\ell - 3\kk)  \right) \\
& = & \frac{1}{(t-\kk)^3} \left( 2\alpha(\ell-\kk)^2 + \kk (t + 2\ell - 3\kk)  \right) > 0
\end{eqnarray*}
since $t > \ell \geq \kk$, as discussed above.
This means that $\frac{d}{dt} y'(y(t)) = y''(y(t)) y'(t) > 0$, which using the fact that $y'(t) > 0$ implies $y''(y(t)) > 0$.
\end{proofof}

From here, we obtain an important statement: As long as $y'(t) \geq \frac{y(t)}{t-\kk}$, the function is convex. In fact, it is convex on ``one additional block'' in the sense that if the condition holds above $t_2$, $y$ is convex above $t_1 = y(t_2)$.

\begin{proposition}
\label{prop:convexity}
Assume that $y(t)$ is a solution defined on $[t_1, 1]$, where $t_1 \geq \max \{\kk, y(t_2)\}$, $t_2 \in (t_1,1)$, we have a finite number of breakpoints $r_{i+1} = y(r_i) < r_i$ such that $r_n < t_1$, and 
$$ y'(t) \geq \frac{y(t)}{t-\kk} $$
whenever $y'(t)$ is defined and $t > t_2$. Then $y(t)$ is convex for $t \in [t_1, 1]$.
\end{proposition}

\begin{proofof}{Proof}
Consider $r_0 = 1$ and $r_{i+1} = y(r_i)$ as before (for $i=1,\ldots,k+1$ such that $r_{n-1} \geq t_1, r_n = y(r_{n-1}) < t_1$).
We proceed by induction on $k$, proving that $y(t)$ is convex on each interval $[r_{i+1}, r_i]$, and also the one-sided derivatives at $r_i$ satisfy $y'_-(r_i) \leq y'_+(r_i)$. 

The initial segment is $$y(t) = 1 - \frac{1}{\theta} + \kk + \frac{(t-\kk)^2}{1-\kk}$$ on $[r_1,r_0]$ which is a convex function. We have $y(1) = 2-\frac{1}{\theta}$ and $y'(1) = 2$. 

Let us first consider the case where $2 - \frac{1}{\theta} \leq \kk$. In this case, the initial segment covers the interval $[\kk, 1]$ and hence there is no need to solve the differential equation. The conclusion that $y(t)$ is convex is satisfied trivially for $t \in [\kk, 1]$.

In the following, we can assume that $2 - \frac{1}{\theta} > \kk$. 
The right derivative at $r_1 = 2-\frac{1}{\theta}$ is:
$$ y'_+(r_1) = \frac{2(r_1-\kk)}{1-\kk} = \frac{2 (2 - \frac{1}{\theta} - \kk)}{1-\kk} > 0 $$
while the left derivative $y'_-(r_1) = \lim_{t \to r_1^-} y'(t)$ is given by the differential equation:
$$ y'_-(r_1) = y'_-(y(1)) = \frac{y(1) - \kk}{1-\kk} \left( 2 - \frac{y(1)}{y'(1)(1-\kk)} \right) 
= \frac{2-\frac{1}{\theta}-\kk}{1-\kk} \left( 2 - \frac{2-\frac{1}{\theta}}{2(1-\kk)} \right). $$
Comparing the expressions for $y'_+(r_1)$ and $y'_-(r_1)$, since $2-\frac{1}{\theta} > \kk$, we have $y'_-(r_1) \leq y'_+(r_1)$.

Now we proceed by induction, using the differential equation and Proposition~\ref{prop:second-diff}.
We assume that $y'(t) \geq \frac{y(t)}{t-\kk}$ on the interval $[t_2,1]$. 
This implies that $y(t)$ is strictly increasing, the sequence of points $r_0 > r_1 > r_2 > \ldots$ is well-defined, and assuming inductively that $y''(t) \geq 0$ on the interval $(r_i, r_{i-1})$, we also get $y''(t) \geq 0$ on $(r_{i+1}, r_i)$ from Proposition~\ref{prop:second-diff}. Since the sequence of breakpoints reaches $r_n < t_1$ after a finite number of iterations, this implies that $y''(t) \geq 0$ for all $t \geq t_1$ except for the breakpoints $r_i$ (note that Proposition~\ref{prop:second-diff} applies to the second derivative at $y(t)$, given the assumptions at $t$). 

Hence $y''(t) \geq 0$ at each point of $(t_1,1)$ except for the points $r_1, r_2, \ldots$ where $y$ is not  differentiable. Here we need to check the one-sided derivatives: By induction, $y'_-(r_i) \leq y'_+(r_i)$, using the differential equation and the inductive hypothesis that $y'_-(r_{i-1}) \leq y'_+(r_{i-1})$. Therefore, $y(t)$ is convex on $[t_1, 1]$. 
\end{proofof}



Finally, we can prove Lemma~\ref{lem:success-properties}.

\begin{proofof}{Proof of Lemma~\ref{lem:success-properties}}
Suppose $y(t)$ is a solution on an interval $[t_0, 1]$ such that $y(t_0) \leq 0$, $y(t)$ is continuously increasing on $[t_0,1]$, with a finite number of breakpoints $r_0=1, r_1=y(r_0), r_{i+1} = y(r_i), r_n < t_0$. Suppose for a contradiction that $y'(t) < \frac{y(t)}{t - \kk}$ for some $t \in (t_0,1)$ where $y(t) > 0$. 
Let $$ t_2 = \sup \left\{t \in (t_0, 1): y(t) > 0, y'(t) < \frac{y(t)}{t - \kk} \right\}.$$
Since $y'(t) \geq \frac{y(t)}{t-\kk}$ for all $t \in (t_2, 1]$, wherever $y'(t)$ is defined, we obtain from Proposition~\ref{prop:convexity} that $y(t)$ is convex on $[t_1,1]$ where $t_1 = \max\{ t_0, y(t_2) \}$, and in particular $y'(t) < \frac{y(t_2)}{t_2-\kk}$ for $t \in (t_1, t_2)$ wherever $y$ is differentiable.
Furthermore, for every $t \in (t_1,t_2)$, if $y(t) \geq t_0$ and $y'(t)$ is defined, we have $y'(y(t)) \leq y'(t)$ by Proposition~\ref{prop:disc-convex} (here we need the fact that $y'(t)>0$). By induction, since $y(y(\ldots(t_2)))$ reaches a value $\leq t_0$ in a finite number of steps, we have $y'(t) < \frac{y(t_2)}{t_2-\kk}$ for all $t \in (t_0, t_2)$ where $y$ is differentiable (which is all but finitely many points). By integration, $y(t_2) - y(t_0) = \int_{t_0}^{t_2} y'(t) dt < (t_2 - t_0) \frac{y(t_2)}{t_2 - \kk} \leq y(t_2)$, which means $y(t_0) > 0$, a contradiction.

Thus we conclude that $y'(t) \geq \frac{y(t)}{t-\kk}$ for every $t \in (t_0,1)$ where $y(t) > 0$ and $y$ is differentiable. This also implies that we can extend $y(t)$ to the interval $[\kk,1]$ by solving the differential equation for $t \in [\kk,t_0]$; the condition $y'(t) \geq \frac{y(t)}{t-\kk}$ implies that $y(t)$ is still going to be strictly increasing for $t \in [\kk,t_0]$. 
Hence, $y(t)$ is strictly increasing on $[\kk,1]$.

Further, Proposition~\ref{prop:convexity} in this case implies that $y(t)$ is convex on $[\kk,1]$:
Assume that $y(t_2) = \kk$ for some $t_2 \in (\kk,1]$ (if not, then $y(1) < \kk$, which means that $y(t)$ on $[\kk,1]$ is defined by the initial segment, which we know to be convex). Then we apply Proposition~\ref{prop:convexity} with this value of $t_2$ and $t_1 = y(t_2) = \kk$: We have $y'(t) \geq \frac{y(t)}{t-\kk}$ for all $t \in [t_2,1]$ where $y$ is differentiable, and therefore $y(t)$ is convex on $[\kk,1]$.

Finally, Proposition~\ref{prop:y-monotonicity} implies that $y(t) < t$ for all $t \in [\kk,1]$. 
\end{proofof}

\subsection{Stability properties}

In the following, we consider solutions parameterized by $\theta$, denoted by $y^{(\theta)}(t)$. (Recall that $\theta$ defines the initial segment, and then the solution $y^{(\theta)}(t)$ is uniquely defined on any interval $[t_1,1]$ where it can be found.) 

Our main goal is to prove that a solution $y^{(\theta)}(t)$, on whatever interval it is defined, behaves continuously in the parameter $\theta$. To prove this, we parameterize the solution in a different way: We introduce a variable $x$ such that the interval $[0,1]$ corresponds to the initial segment $[r_1,r_0]$, $[1,2]$ corresponds to the next segment $[r_2,r_1]$, etc. Formally, the relationship between the variables $x$ and $t$ is
\begin{itemize}
\item $t(x) = 1 - (\frac{1}{\theta} - 1) x$ for $x \in [0,1]$,
\item $t(x) = y(t(x-1))$ for $x > 1$.  
\end{itemize}
(We caution the reader that in the following, $t(\cdots)$ denotes a function rather than multiplication of $t$ and $(\cdots)$.)

Note that while $y(t)$ is increasing in a successful solution, $t(x)$ should be decreasing:
On the initial segment $[0,1]$, we have $t'(x) = -(\frac{1}{\theta} - 1) < 0$ by definition, and 
as long as $y'(t(x)) > 0$, we have $t'(x+1) = y'(t(x)) t'(x) < 0$ by induction.

Our problem can now be formulated as follows: (using the identities $t(x+1) = y(t(x))$, $t(x+2) = y(y(t(x))$,
$t'(x+1) = y'(t(x)) t'(x)$, $t'(x+2) = y'(y(t(x))) t'(x+1)$):
\begin{itemize}
\item Initial parameterization: $t(x) = 1 - (\frac{1}{\theta} - 1) x$ for $x \in [0,1]$,
\item Initial function segment: $t(x+1) = y(t(x)) = 1 - \frac{1}{\theta} + \kk + \frac{(1 - (\frac{1}{\theta} - 1) x -\kk)^2}{1-\kk} $ for $x \in [0,1]$,
\item Differential equation: for $x>0$,
\begin{equation}
\label{eq:t-equation}
 t'(x+2) = \frac{t(x+1) - \kk}{t(x) - \kk} \left( 2 t'(x+1) - \frac{t(x+1)}{t(x) - \kk} t'(x) \right). 
\end{equation}
\end{itemize}


Conversely, given $t(x)$, we can reconstruct $y(t)$ as follows: Assuming that $t(x)$ is strictly decreasing on $[0,\xi-1$], if $t_1 \in [t(\xi-1),1]$ then there is a unique point $x_1 \in [0,\xi-1]$ such that $t(x_1) = t_1$. Then, we have $y(t_1) = y(t(x_1)) = t(x_1 + 1)$. We can also compute the derivative of $y$ from $t'(x+1) = y'(t(x)) t'(x)$:
$$ y'(t(x)) = \frac{t'(x+1)}{t'(x)}.$$
This is well-defined for non-integer points in $(0,\xi-1)$; the integer points correspond to the breakpoints where $y$ and $t$ are not differentiable.

\begin{defn}
We call a solution $t(x)$ for $x \in [0,\xi]$ {\em strongly monotonic}, if there is $\alpha>0$ such that for all $x \in (0,\xi) \setminus \ZZ$, $t'(x) \leq -\alpha$.
\end{defn}

Strong monotonicity of $t(x)$ is implied by a similar property for $y(t)$, as the following statement shows.

\begin{proposition}
\label{prop:t'-bound}
If $y'(t) \geq \beta > 0$ for $t \in [t_0,1]$ except for finitely many breakpoints $r_0=1$, $r_{i+1} = r_i$, $r_n < t_0$, 
then the corresponding function $t(x)$ defined as above is strongly monotonic on $[0,\xi]$ such that $t(\xi) = y(t_0)$.
\end{proposition}

\begin{proofof}{Proof}
Suppose that $y'(t) \geq \beta > 0$ everywhere except for the breakpoints; we can assume w.l.o.g.~that $\beta \in (0,1)$. 

Let $\xi > 0$ be such that $t(x) = y(t(x-1))$ defined as above satisfies $t(\xi) = y(t_0)$
(this construction is valid because we reach $t(\xi) = y(t_0)$ after finitely many breakpoints).
For $x \in (0,1)$, we have $t'(x) = -(\frac{1}{\theta} - 1) < 0$ by definition. 
For $x > 0$, we have 
$$ t'(x+1) = y'(t(x)) t'(x).$$
Hence, by induction, assuming that $y'(t(x)) \geq \beta$, we obtain $t'(x+1) \leq \beta t'(x)$ for every non-integer $x \in (0,\xi-1)$. We assumed that $\beta \in (0,1)$; so if $n = \lceil \xi \rceil$, we have $t'(x) \leq -\beta^n (\frac{1}{\theta} - 1) < 0$. Hence we can set $\alpha = \beta^n (\frac{1}{\theta} - 1)$.
\end{proofof}

Now we can prove our stability result for $t(x)$.

\begin{lemma}
\label{lem:t-stability}
Fix $f>0$. Let $\theta \in [\frac12,1]$ be such that for some $\alpha, \xi > 0$ the system above has a strongly monotonic solution $t(x)$ defined on $[0,\xi]$. Then for fixed $f$, the solution depends continuously on $\theta$: For any $\epsilon>0$ there is $\delta>0$ such that if $|\tilde{\theta} - \theta| < \delta$, there is a solution $\tilde{t}(x)$ for parameter $\tilde{\theta}$ such that for all $x \in [0,\xi]$, $|\tilde{t}(x) - t(x)| < \epsilon$ and for all $x \in (0,\xi) \setminus \ZZ$, $|\tilde{t}'(x) - t'(x)| < \epsilon$. In particular, for any $\tilde{\theta}$ sufficiently close to $\theta$, the solution $\tilde{t}(x)$ is still strongly monotonic.
\end{lemma}

\begin{proofof}{Proof}
We prove inductively that for $x \in [i, i+1]$, the solution $t(x)$ depends continuously on $\theta$:
For any $\epsilon>0$, there is $\delta_i>0$ such that if $|\tilde{\theta}-\theta| < \delta_i$, $t(x)$ is the solution corresponding to $y^{(\theta)}(t)$ for parameter $\theta$, and $\tilde{t}(x)$ is the solution corresponding to $y^{(\tilde{\theta})}(t)$, then $|\tilde{t}(x) - t(x)| < \epsilon$ and $|\tilde{t}'(x) - t'(x)| < \epsilon$.

For $x \in [0,1]$ and $x \in [1,2]$, we have explicit expressions for which one can verify directly that $t(x)$ and $t'(x)$ behave continuously in $\theta$, hence the statement is valid in the base case.

For the inductive step, we appeal to the differential equation (\ref{eq:t-equation}).
By the inductive hypothesis, $t(x), t(x+1), t'(x)$ and $t'(x+1)$ depend continuously on $\theta$.
We note that as long as $t(x+1) \geq \kk$, we have $t(x) \geq \kk + \beta$ for some $\beta>0$.
This implies that the behavior is continuous in the denominator of (\ref{eq:t-equation}).
Hence $t'(x+2)$ also depends continuously on $\theta$: for every $\epsilon>0$ there is $\delta_i>0$ such that
if $|\tilde{\theta} - \theta| < \delta_i$ then $|\tilde{t}'(x+2) - t'(x+2)| < \epsilon$ for non-integer $x$, and by integration
also $|\tilde{t}(x+2) - t(x+2)| < \epsilon$. By induction, this holds for all $x \leq \xi-2$.

If we take $\delta = \min_{1 \leq i \leq \lceil \xi \rceil} \delta_i$, we obtain the statement of the lemma:
For given $|\tilde{\theta} - \theta| < \delta$, we get $|\tilde{t}(x) - t(x)| < \epsilon$ for all $x \in [0,\xi]$ and $|\tilde{t}'(x) - t'(x)| < \epsilon$ for all $x \in (0, \xi) \setminus \ZZ$.
In particular, since $t(x)$ is strongly monotonic, i.e.~$t'(x) \leq -\alpha$ for $x \in (0,\xi) \setminus \ZZ$,
for $\epsilon = \frac12 \alpha$, we obtain $\tilde{t}(x)$ which is still strongly monotonic.
\end{proofof}

\subsection{Existence of an optimal solution}

Our goal is to prove that there is $\theta^*$ for which there is a complete successful solution with $y^{(\theta^*)}(\kk) = 0$. First, let us prove that \revcolor{if} $y^{(\theta)}(\kk) < 0$, then $\theta$ is suboptimal.

\begin{lemma}
\label{lem:theta-continuous}
If $\theta \in [\frac12,\frac{1+f}{1+2f}]$ is such that there is a complete successful solution for $\theta$ such that $y^{(\theta)}(\kk) < 0$, then there is $\tilde{\theta} > \theta$ for which there is also a successful solution $y^{(\tilde{\theta})}(t)$.
\end{lemma}

\begin{proofof}{Proof}
Consider a complete successful solution $y(t)$ for parameter $\theta$, and the corresponding solution $t(x)$ as above. 
By Lemma~\ref{lem:success-properties}, $y(t)$ is strictly increasing and convex on $[\kk,1]$, with a positive derivative everywhere except for finitely many breakpoints, and $y(\kk) < 0$. Pick some point $t_0 > \kk$ such that $y(t_0) < 0$, and $t_0$ is not a breakpoint (which exists by the continuity of $y(t)$). We claim that the derivative $y'(t_0)$ is strictly positive:  Again by continuity, there is a point $t_1 \in (t_0, 1]$ such that $y(t_1) = t_0 > \kk$. By Lemma~\ref{lem:success-properties}, we have $y'(t_1) \geq \frac{y(t_1)}{t_1 - \kk}$. The differential equation (\ref{eq:diff-eq}) implies that $y'(y(t_1)) = y'(t_0) = \beta > 0$. Moreover, $y(t)$ is convex for $t \in [t_0, 1]$, so wherever the derivative is defined for $t \geq t_0$, we have $y'(t) \geq \beta$.

By Proposition~\ref{prop:t'-bound}, this means that $t(x)$ is strongly monotonic on $[0,\xi]$, where $t(\xi-1) = t_0$ and $t(\xi) = y(t_0) < 0$. That is, $t'(x) \leq -\alpha < 0$ for all $x \in (0,\xi)$ where $t'(x)$ is defined.
By Lemma~\ref{lem:t-stability}, for any $\epsilon>0$, for $\tilde{\theta}$ sufficiently close to $\theta$, we obtain a solution $\tilde{t}(x)$ such that $|\tilde{t}(x) - t(x)| < \epsilon$ and $|\tilde{t}'(x) - t'(x)| < \epsilon$ wherever defined. In particular, for $\epsilon = \frac12 \min \{ \alpha, |t(\xi)|, t_0 - \kk \}$, $\tilde{t}(x)$ is still strongly monotonic.

Translating back to the formalism of $y(t)$, we obtain a solution $\tilde{y}(t)$ defined on $[\tilde{t}_0,1]$ where $\tilde{t}_0 = \tilde{t}(\xi-1) \geq t_0 - \epsilon > \kk$, and $y(\tilde{t}_0) = \tilde{t}(\xi) \leq t(\xi) + \epsilon < 0$. Hence, this solution is successful.
\end{proofof}

We prove next that if $\theta^*$ is a limit point of values $\theta$ for which a successful solution exists, then it also exists for $\theta^*$. 

\begin{lemma}
\label{lem:theta-limit}
If $\theta^* \in [\frac12,\frac{1+f}{1+2f}]$ is such that for every $\delta>0$, there is $\theta \in (\theta^*-\delta,\theta^*)$ such that $y^{(\theta)}(t)$ is a successful solution, then $y^{(\theta^*)}(t)$ is a successful solution.
\end{lemma}

\begin{proofof}{Proof}
Suppose for a contradiction that $y^{(\theta^*)}(t)$ is not a complete successful solution. However, it is well-defined on the initial segment $[2-\frac{1}{\theta}, 1]$, and it satisfies $y'(t) \geq \frac{y(t)}{t-\kk}$. Therefore, the solution can be extended to $[r_2,r_1]$, where $r_1 = 2-\frac{1}{\theta}, r_2 = y(r_1)$, using the differential equation. As long as $y'(t) \geq \frac{y(t)}{t-\kk}$ on an interval $(r_{i+1}, r_i)$, we can use the differential equation to extend the solution to $[r_{i+2}, r_{i+1}]$.
However, at some point we obtain $y'(t) < \frac{y(t)}{t-\kk}$, otherwise the solution is successful by Lemma~\ref{lem:success-criterion}.
Let $(r_{k+1},r_k)$ be the first interval on which this happens; this implies that $y'(t) \geq \frac{y(t)}{t-\kk}$ for all $t > r_k$, and this in turn implies that $y'(t) > 0$ for all $t > r_{k+1}$, using the differential equation.

Let us consider now the corresponding function $t(x)$ defined on $[0,k]$, such that $t(k) = r_k$, and $y(t(x)) = t(x+1)$ for all $x \in [0,k]$. By Proposition~\ref{prop:t'-bound}, $t(x)$ is strongly monotonic. Therefore, by Lemma~\ref{lem:t-stability}, the solution $\tilde{t}(x)$ for $\tilde{\theta}$ sufficiently close to $\theta$ is still strongly monotonic and arbitrarily close to $t(x)$. Hence, the corresponding solution $\tilde{y}(t)$ still satisfies $\tilde{y}'(t) < \frac{\tilde{y}(t)}{t-\kk}$ for some $t > \kk$ and hence it cannot be successful by Lemma~\ref{lem:success-properties}. So there exists $\delta > 0$ such that no solution corresponding to $\theta \in (\theta^*-\delta, \theta^*)$ is successful, which is a contradiction with the assumptions of the lemma.
\end{proofof}

Finally, we can prove that there exists $\theta^*$ such that there is a complete successful solution with $y^{(\theta^*)}(\kk) = 0$, which implies our optimal algorithm. 

\begin{theorem}
For every $f>0$, there exists $\theta^* \in [\frac12, \frac{1+f}{1+2f}]$ such that there is a complete successful solution $y^{(\theta^*)}(t)$ on $[\kk,1]$ such that $y^{(\theta^*)}(\kk) = 0$.
\end{theorem}

\begin{proofof}{Proof}
Given $f>0$, define $S \subseteq [\frac12, 1]$ to be the set of values $\theta$ such that there is a successful solution $y^{(\theta)}(t)$. This set is nonempty, because for every $f>0$, a successful solution exists for $\theta=1/2$. Also, $\theta \leq \frac{1+f}{1+2f}$ for every $\theta \in S$, because we know that a successful solution cannot exist for $\theta > \frac{1+f}{1+2f}$. In particular, $S$ is bounded away from $1$.

Consider $\theta^* = \sup S$. By the above, we have $\theta^* \in [\frac12, \frac{1+f}{1+2f}]$. 
We claim that $\theta^* \in S$: If $\theta^* \notin S$, then for every $\delta>0$, there is $\theta \in (\theta^*-\delta, \theta^*)$, $\theta \in S$, i.e.~there is a successful solution for $\theta$. By Lemma~\ref{lem:theta-limit}, there is also a successful solution for $\theta^*$, hence $\theta^* \in S$. 

By Lemma~\ref{lem:success-properties}, we can assume that $y^{(\theta^*)}(t)$ is a complete successful solution on $[\kk,1]$, which means $y^{(\theta^*)}(\kk) \leq 0$. Let us assume for a contradiction that $y^{(\theta^*)}(\kk) < 0$. Then by Lemma~\ref{lem:theta-continuous}, there exists $\theta' > \theta^*$ such that there is also a successful solution for $\theta'$. But this means that $\theta' \in S$ which contradicts $\theta^* = \sup S$. 

Hence we have a complete successful solution for $\theta^*$ such that $y^{(\theta^*)}(\kk) = 0$.
\end{proofof}

\section{Omitted Proofs}
\label{apx:proofs}

\subsection{Proof of \Cref{lem:discretize}}
\label{apx:lem:discrete}
We prove the lemma by contradiction. Suppose that there exists an instance with distributions $\cD_1,\ldots,\cD_n$, such that the best achievable competitive ratio is strictly less than $\alpha(f)$; say less than $\alpha(f)-3\delta$ for some $\delta>0$. We define discrete distributions $\cD'_1,\ldots,\cD'_n$ as follows: Given $X_i \sim \cD_i$ and $X_{\textrm{max}} = \max_{1 \leq i \leq n} X_i$, let $M \triangleq \EX[X_{\textrm{max}}]$ and $M' \triangleq \inf \{ v: \EX[(X_{\textrm{max}} - v)_+] \leq \delta M \}$. Since $\lim_{v \to \infty} \EX[(X_{\textrm{max}} - v)_+] = 0$, the value $M'$ is well defined. We also define $\cV \triangleq \{ \delta (1+\delta)^k M: 0 \leq k \leq \lceil \log_{1+\delta} \frac{M'}{\delta M} \rceil \}$,
and $X'_i \triangleq \min \{ v \in \cV: v \geq \min(X_i,M')\}$; in other words, $X'_i$ is obtained by rounding $X_i$ up to the nearest value in $\cV$ (and rounding down to $\max \{v\in \mathcal{V}\}$ if it is greater than $M'$). It is obvious that $X'_i$ is a discrete random variable. Let us denote the distribution of the discrete random variable $X'_i$ by $\cD'_i$.

We claim that $X'_{\max} = \max_{1 \leq i \leq n} X'_i$ does not differ significantly from  $X_{\textrm{max}}$.
On the one hand, capping each variable $X_i$ at $\max \{v\in \mathcal{V}\}$ decreases $\EX[X_{\textrm{max}}]$ by at most $\EX[(X_{\textrm{max}} - M')_+] 
\leq \delta M$, or equivalently shrinks $\EX[X_{\textrm{max}}]$ by a multiplicative factor of ${1-\delta}$. On the other hand, rounding up each $X_i$ to a value in the form of $\delta(1+\delta)^k M$ implies that $\EX[X'_{\textrm{max}}]\leq (1+\delta)\EX[X_{\textrm{max}}]+\delta M\leq (1+2\delta)\EX[X_{\textrm{max}}]$. We therefore conclude that
$$(1+2\delta)\,\EX[X_{\textrm{max}}]\geq \EX[X'_{\max}] \geq (1-\delta)\, \EX[X_{\textrm{max}}]$$


Assume now that there is a policy that achieves a competitive ratio of $\alpha(f)$ for any discretized instance. We further assume that this policy never swaps a held value $x$ with $y$ if $y<(1+f)x$. This is a property of any meaningful policy for the problem; If not satisfied, just consider a modified (and weakly better) policy in this argument that exactly follows the original policy but it does not swap the held value $x$ with the current revealed value $y$ if $y<(1+f)x$. Then we can apply this policy to any instance, with values discretized to $\cV$ as above. Assume that we eventually accept the value $X'_\tau$ and pay the buyback cost $c'$ in the discretized instance, so that $\EX[X'_\tau - c'] \geq \alpha(f) \EX[X'_{\max}]$. Therefore, we eventually accept the value $X_\tau$ and pay the buyback cost $c$ in the original instance. Since values were rounded up by at most a factor of $1+\delta$ and an additive error of $\delta M$, we have $X'_\tau \leq  (1+\delta)X_\tau+\delta M$.  The buyback cost in the discretized instance satisfies $c' \geq c$, because the original values are smaller, except for values above $M'$---but a variable of value above $M'$ will never be bought back in the discretized instance. Hence, the reward in the original instance is:
\begin{align*}
\EX[X_\tau - c] &\geq \EX[\frac{1}{1+\delta} (X'_\tau - \delta M) - c'] \geq \EX[X'_\tau - c'] - \frac{\delta}{1+\delta} \EX[X'_\tau] - \frac{\delta}{1+\delta} M \\
&\geq \alpha(f) \EX[X'_{\max}] -  \frac{\delta}{1+\delta}(1+2\delta) \EX[X_{\textrm{max}}] - \frac{\delta}{1+\delta} \EX[X_{\textrm{max}}] \\
&\geq \left(\alpha(f)(1-\delta) -\frac{\delta}{1+\delta}(2+2\delta)\right)\EX[X_{\textrm{max}}] \geq (\alpha(f) - 3 \delta) \EX[X_{\textrm{max}}]
\end{align*}
This contradicts the assumption that this factor cannot be achieved for the original instance.\qed

\subsection{Proof of \Cref{lem:monotonize}}
\label{sec:lem:monotonize}

Consider an instance $I_1$ with variables $X_i = v_i Be(q_i)$ and $X_{i+1} = v_{i+1} Be(q_{i+1})$, $v_i < v_{i+1}$, and the same instance with the ordering of variables $X_i, X_{i+1}$ swapped, which we call $I_2$.

Similar to the proof of \Cref{lem:splitting}, let us couple the two instances so that the instantiations of all the random variables are the same, \textit{except for the instantiations of} $X_i$. Let us attempt to define the behavior of a policy $\cP_2$ on the instance $I_2$ given a policy $\cP_1$ on the instance $I_1$.

First of all, we let $\cP_2$ copy $\cP_1$ until the random variables $X_{i+1}$ and $X_{i}$ arrive. Suppose both polices hold value $x$ before observing the variables $X_i$ and  $X_{i+1}$.   

\paragraph{Case 1:} In instance $I_2$, we observe $X_{i+1} = 0$. We proceed to $X_{i}$ and let $\cP_2$ copy $\cP_1$. Note the performance of the policy $\cP_2$ exactly matches the performance of the policy $\cP_1$.

\paragraph{Case 2:} In instance $I_2$, we observe $X_{i+1} = v_{i+1}$. In this case, we let $\cP_2$ simulate what would have happened if $X_i$ appeared first. $\cP_2$ generates a duplicate sample $X’_i$, and observes the last acceptance of $\cP_1$ on $X'_i$ and $X_{i+1}$. If this last accepted value is $v$ (equaling either $v_i$ or $v_{i+1}$, or perhaps it equals $x$ signifying that $\cP_1$ does not accept the values $v_i$ or $v_{i+1}$), $\cP_2$ accepts $X_{i+1}$ with probability $\frac{v - x}{v_{i+1}-x}$. After this decision, $\cP_2$ ignores $X_i$ and proceeds to copy $\cP_1$ once again, assuming it is holding $v$.

We claim that this yields a policy which performs at least as well as $\cP_1$. To prove this, we couple the instantiation of $X’_i$ in $I_2$ with $X_i$ in $I_1$. Then $\cP_2$ clearly outperforms $\cP_1$ if the last accepted value $v$ of $\cP_1$ is either $x$ or $v_{i+1}$, as they obtain exactly the same net reward (including the buyback cost) in the future and $\cP_1$ may only pay an additional buyback cost compared to $\cP_2$ during processing $X'_i$ and $X_{i+1}$. 

The only non-trivial case is when $X'_i = v_i$, $X_{i+1} = v_{i+1}$, and $\cP_1$ accepts $X'_i$ and then rejects $X_{i+1}$. But in this case, after $X_i$ and $X_{i+1}$, $\cP_1$ and $\cP_2$ hold the same value in expectation, because:
\[x\cdot \left(1 - \frac{v - x}{v_{i+1}-x}\right) + v_{i+1}\cdot \frac{v - x}{v_{i+1}-x} = \frac{x(v_{i+1}-v) +v_{i+1}(v-x) }{v_{i+1}-x} = v.\]
Since the expected net reward (including the buyback cost) is linear in the accepted values, the linearity of expectation implies that the expected net reward of $\cP_2$ equals the expected net reward of $\cP_1$ in this case, both during processing these two random variables and in the future.\qed

\subsection{Proof of \Cref{lem:LP-3}}
\label{apx:lem:duality}
Starting from the LP in \ref{eq:disc-primal-LP}, we introduce a dual variable $y_{i,t}$ for the constraint $\Delta_{i,t} \geq q_t \sum_{j=t+1}^{n} (\Delta_{t,j} - \Delta_{i,j}) + v_t - (1+f) v_i$, for $0 \leq i < t \leq n$, and $\Theta$ for the constraint $\sum_{i=1}^{n} v_i q_i \prod_{j=i+1}^{n} (1-q_j) = 1$. We write down constraints for each $0 \leq s < t \leq n$ corresponding to $\Delta_{s,t}$ and for each $1 \leq t \leq n$, corresponding to $v_t$. Note that the constraint corresponding to $\Delta_{0,t}$ looks somewhat different from the others, because $\Delta_{0,t}$ appears in the objective and hence the RHS is nonzero. Summarizing these steps, we obtain the following dual LP:
\begin{align*}
\underset{y_{s,t}\geq 0 \,,\, \Theta\geq 0}{\max} \;\;\;\;\;\;  & \Theta
     && \text{s.t.} 
    \\
    &y_{0,t} + \sum_{i=1}^{t-1} q_i y_{0,i}  \leq 1
     && 
     1 \leq t \leq n
     \\
     & y_{s,t} - \sum_{i=0}^{s-1} q_s y_{i,s} + \sum_{j=s+1}^{t-1} q_j y_{s,j}  \leq  0
     && 1 \leq s < t \leq n \\
     & -\sum_{i=0}^{t-1} q_t y_{i,t} + (1+f) \sum_{j=t+1}^{n} q_j y_{t,j} + \Theta q_t \prod_{j=t+1}^{n} (1-q_j)  \leq 0
     && 1 < t \leq n
\end{align*}
We remark that this dual LP has the same optimal objective value compared with \ref{eq:disc-primal-LP} by applying strong LP duality; hence, its optimal value characterizes the worst-case competitive ratio of the optimal online policy due to \Cref{lem:LP-1}. Lastly, we rewrite this LP in a slightly cleaner form by using the substitution $x_{s,t} = q_t\,y_{s,t}$. This completes the proof. \qed

\subsection{Proof of \Cref{lem:continuous-reduction}}
\label{apx:proof-cont-red}
We will show that the solution to \ref{integrals} can be used to construct a feasible solution for \ref{eq:discrete-LP} with given $q_1,...,q_n \in [0,1]$. Define $r_t = \sum_{s = 1}^{t} \hat{q}_s = \prod_{s=t+1}^n (1-q_s)$ and let $x_{0,t}  = \int_{r_{t-1}}^{r_t}h(r)dr$ and $x_{s,t}  = \int_{r_{s-1}}^{r_s}\int_{r_{t-1}}^{r_t}g(r,u)dudr$. First, observe that:
\[\frac{d}{dr}\left(r\left(1 - \int_{0}^rh(x)dx\right)\right) = 1 - rh(r) - \int_{0}^rh(x)dx \geq 0,\]
where the last inequality holds because of \eqref{OIfromzero}. This implies that $r\left(1 - \int_{0}^rh(x)dx\right)$ is an increasing function. Consequently, we conclude that:
\begin{eqnarray*}
& r_{t-1}\left(1 - \int_{0}^{r_{t-1}}h(r)dr\right) &\leq r_t\left(1 - \int_{0}^{r_{t}}h(r)dr\right).
\end{eqnarray*}
Dividing both sides by $r_t$ and rearranging terms yields:
\begin{eqnarray*}
\int_{r_{t-1}}^{r_{t}}h(r)dr &\leq (1 - \frac{r_{t-1}}{r_{t}})\cdot\left(1 - \int_{0}^{r_{t-1}}h(r)dr\right).
\end{eqnarray*}
Using the fact that $q_t = 1 - \frac{r_{t-1}}{r_{t}}$ , it follows that: $x_{0,t} \leq q_t \left( 1 - \sum_{i=1}^{t-1} x_{0,i} \right),$ as desired.

In an analogous manner, by using \eqref{OI}, we conclude that:
\[\int_{r_{t-1}}^{r_t}g(r, u)du \leq (1 - \frac{r_{t-1}}{r_t})\cdot\left(h(r) + \int_0^r g(x, r)dx - \int_r^{r_{t-1}} g(r, x)dx\right),\] and hence, it follows that:
\begin{eqnarray*}
\frac{x_{s,t}}{q_t} & = & \frac{r_t}{r_t-{r_{t-1}}} \cdot \int_{r_{s-1}}^{r_s}\int_{r_{t-1}}^{r_t}g(r, u)du dr \\
& \leq & \int_{r_{s-1}}^{r_s}\left\{h(r) + \int_0^r g(x, r)dx - \int_r^{r_{t-1}} g(r, x)dx\right\}dr \\
& = &\int_{r_{s-1}}^{r_s}\left\{h(r) + \int_0^{r_{s-1}} g(x, r)dx - \int_{r_s}^{r_{t-1}}g(r, x)dx + \int_{r_{s-1}}^r g(x, r)dx - \int_{r}^{r_{s}} g(r, x)dx\right\}dr \\
& = & \int_{r_{s-1}}^{r_s}h(r)dr + \int_{r_{s-1}}^{r_s}\int_0^{r_{s-1}} g(x, r)dx dr - \int_{r_{s-1}}^{r_s} \int_{r_s}^{r_{t-1}}g(r, x)dx dr \\ 
&&+ \int_{r_{s-1}}^{r_s}\int_{r_{s-1}}^r g(x, r)dx dr - \int_{r_{s-1}}^{r_s}\int_{r}^{r_{s}} g(r, x)dx dr \\
& = &\sum_{i=0}^{s-1} x_{i,s} - \sum_{j=s+1}^{t-1} x_{s,j} + \int_{r_{s-1}}^{r_s}\int_{r_{s-1}}^r g(x, r)dx dr - \int_{r_{s-1}}^{r_s}\int_{r}^{r_{s}} g(r, x)dx dr  \\
& = &\sum_{i=0}^{s-1} x_{i,s} - \sum_{j=s+1}^{t-1} x_{s,j},
\end{eqnarray*}
where the last line is a consequence of swapping the order of integrals (Fubini's theorem). Finally, note that by integrating \eqref{coverage} from ${r_{t-1}}$ to ${r_t}$, we get that:
\begin{eqnarray*}
\Theta \hat{q}_t  & = & \Theta (r_t -r_{t-1}) \\
& \leq & \int_{r_{t-1}}^{r_t}\left\{h(u) + \int_0^u g(r, u)dr  - (1+f) \int_{u}^{1} g(u, r)dr \right\}du \\
& = & \int_{r_{t-1}}^{r_t}\left\{h(u) + \int_0^{r_{t-1}} g(r, u)dr  - (1+f) \int_{r_t}^{1} g(u, r)dr +  \int_{r_{t-1}}^{u} g(r, u)dr  - (1+f) \int_{u}^{r_t} g(u, r)dr \right\}du \\
& = & \int_{r_{t-1}}^{r_t}h(u) du + \int_{r_{t-1}}^{r_t} \int_0^{r_{t-1}} g(r, u)dr du  - (1+f) \int_{r_{t-1}}^{r_t} \int_{r_t}^{1} g(u, r)dr du \\ &&+ \int_{r_{t-1}}^{r_t} \int_{r_{t-1}}^{u} g(r, u)dr du  - (1+f) \int_{r_{t-1}}^{r_t} \int_{u}^{r_t} g(u, r)dr du \\
& = & \sum_{i=0}^{t-1} x_{i,t} - (1+f) \sum_{j=t+1}^{n} x_{t,j} + \int_{r_{t-1}}^{r_t} \int_{r_{t-1}}^{u} g(r, u)dr du  - (1+f) \int_{r_{t-1}}^{r_t} \int_{u}^{r_t} g(u, r)dr du \\
& = & \sum_{i=0}^{t-1} x_{i,t} - (1+f) \sum_{j=t+1}^{n} x_{t,j}  - f \int_{r_{t-1}}^{r_t} \int_{u}^{r_t} g(u, r)dr du \\
& \leq & \sum_{i=0}^{t-1} x_{i,t} - (1+f) \sum_{j=t+1}^{n} x_{t,j},
\end{eqnarray*}
and our proof is complete.
\qed

\subsection{Proof of \Cref{prop:from-diffeq-to-phiy}}
\label{apx:proof:phiy-existence}
Let $y$ be the existing solution $y_f$ to \ref{difEq} as in \Cref{thm:diff-eq-existence}. Moreover, let $\tau(t) = y^{-1}(t)$ for $t\in(0,y(1)]$, and as a convention, let $\tau(t)\triangleq 1$ for $t\in [y(1),1]$. Now define:
\begin{align*}
    \forall\,t\in(0,1]:~~\varphi(t) = 
        \frac{\Theta}{(1+f)\tau(t)-f},
\end{align*}
with $\Theta = \theta_f = \frac{1}{2-y(1)}$. Notice that for $t \in [y(1),1]$ we have $\varphi(t) = \Theta$. 

We now show that $\{\varphi,y\}$ satisfy all the Properties~I-VI. First, \ref{eq:y-property-1}, \ref{eq:y-property-2}, and \ref{eq:coverage-property5} are direct consequences of our definition. Now substituting $t$ with $\tau(t)$ in \ref{difEq} we get that for $t \in [k_1,y(1)]$\revcolor{:}
\begin{align}\label{corollary}
    y'(t) = \frac{t-\cf}{\tau(t)-\cf}\left(2- \frac{t\tau'(t)}{\tau(t)-\cf}\right), 
\end{align}
This equation is also correct for $t \geq y(1)$ using the boundary condition in \ref{difEq}. We use \cref{corollary} to establish \ref{eq:phi-after-k1-property4} and \ref{eq:property6}. Notice that a necessary and sufficient condition to establish these properties is to verify the equation at a boundary value and then demonstrate that the derivative of the left-hand side is equal to the derivative of the right-hand side. For \ref{eq:phi-after-k1-property4}, notice that for  $t = 1$: 
\begin{align*}
 \Theta = 1- (1-y(1))\Theta = 1- \int_{y(1)}^{1} \varphi(s) ds
\end{align*}
Differentiating both sides of \ref{eq:phi-after-k1-property4} for $t\in[k_1,1]$, we must prove:
\begin{align*}
t\varphi'(t) + \varphi(t) = -\varphi(t) + \varphi(y(t))y'(t),
\end{align*}
which is equivalent to:
\begin{align*}
    -\frac{t\tau'(t)}{(1+f)(\tau(t)-\cf)^2}+\frac{1}{(1+f)(\tau(t)-\cf)}= - \frac{1}{(1+f)(\tau(t)-\cf)} + \frac{y'(t)}{(1+f)(t-\cf)},
\end{align*}
but this is simply a rearrangement of \cref{corollary}. To show \ref{eq:property6}, for $s \geq k_1$, we now know that
\begin{align*}
    s\varphi'(s) + \varphi(s) = -\varphi(s) + \varphi(y(s))y'(s) \Rightarrow s^2\varphi'(s)+2s\varphi(s)=s\varphi(y(s))y'(s)~.
\end{align*}
Using this equation, we have that: 
\begin{align*}
    \frac{d}{ds}\left(s^2\varphi(s)\right) = s^2\varphi'(s)+ 2s\varphi(s)
    = s\varphi(y(s))y'(s)= \frac{d}{ds}\left(\int_0^{y(t)} \varphi(s) \tau (s)\,ds\right)~.
\end{align*}
Integrating both sides from $k_1$ to $t$ implies \ref{eq:property6}. 

Finally, we would like to establish \ref{eq:OIfromzero-phi-property3}, that is, we want to prove for $t \in (0,k_1]$:
\begin{align*}
   \frac{t\; \Theta}{(1+f)(\tau(t) - \cf)} \leq 1-\int_0^t \frac{\; \Theta}{(1+f)(\tau(s) - \cf)} \,ds
\end{align*}
But notice that right-hand side is decreasing in $t$ and the left-hand side is increasing in $t$ by \cref{thm:diff-eq-existence}. Thus the gap is decreasing and proving the inequality for all $t$ reduces to proving it for $k_1$. However, we have already proved that both sides are equal for $t \in[ k_1,1]$, which completes the proof of \ref{eq:OIfromzero-phi-property3}. $\hfill\Box$
\begin{corollary} \label{cor:shrinkingGap}
    We have established a stronger result than \eqref{eq:OIfromzero-phi-property3}: not only is \( 1 - \int_{0}^{t} \varphi(s) \, ds - t\varphi(t) \) always positive in \( t \in [0, k_1] \), but it is also decreasing in that interval and equals zero at \( t = k_1 \). 
\end{corollary}

\subsection{Proof of \Cref{prop:from-phiy-to-integrals}}
\label{apx:proof:properties}
Recall the definition of $\tau(t) = y^{-1} (t)$ (which is well-defined because of \ref{eq:y-property-1}).
First, we would like to propose assignment for the functions $g$ and $h$ in terms of $\varphi$ and $y$. To discover these assignments, we argue informally using our interpretation of the problem from \cref{sec:reduction-to-diff-eq}:
\begin{itemize}
\item Clearly, $g(s,t)=0$ whenever $t\leq \tau(s)$. For $t>\tau(s)$, we have:
\begin{align}
    g(s,t)\,ds\,dt 
    &= \Pr\Bigl[\text{\Cref{alg:quantile-selection} selects } t \text{ while holding }s\,\big|\,t \text{ is non-zero}\Bigr]\frac{dt}{t}\nonumber\\
    &= \Pr\Bigl[\bigl(\text{\Cref{alg:quantile-selection} selects }s\bigr)\;\mathbin{\scalebox{1.5}{\ensuremath{\cap}}}\;\bigl(\text{all }r\in[\tau(s),t)\text{ are zero}\bigr)\Bigr]\frac{dt}{t}\nonumber\\
    &\label{eq:g-from-phi-y}
     = \varphi(s)ds\,\frac{\tau(s)}{t}\,\frac{dt}{t} = \frac{\varphi(s)\tau(s)}{t^2}\,ds\,dt
\end{align}
Here, the third equality holds because the probability that all $r\in[\tau(s),1)$ are zero ($\equiv \tau(s)$) is equal to the probability that all $r\in[t,1)$ are zero ($\equiv t$), times the probability that all $r\in[\tau(s),t)$ are zero. Therefore:
\begin{equation}
\label{eq:g-phi-y}\tag{\texttt{g-function}}
 g(s,t) = \begin{cases}
      0 & t \leq \tau(s),\\
      \displaystyle \frac{\varphi(s)\,\tau(s)}{t^2} & t>\tau(s).
   \end{cases}
\end{equation}
\item Furthermore, for $t\in(0,k_1]$, we clearly have $h(t)=\varphi(t)$. For $t\in (k_1,1]$ similar to above we have:
\begin{align}
    h(t)\,dt
    &= \Pr\bigl[\text{\Cref{alg:quantile-selection} selects }t\text{ as its first pick}\,\big|\,t \text{ is non-zero}\bigr]\frac{dt}{t}\nonumber\\
    &= \Pr\Bigl[\text{\Cref{alg:quantile-selection} picks nothing in }(0,k_1] \mathbin{\scalebox{1.5}{\ensuremath{\cap}}} \text{(all }s\in[k_1,t)\text{ are zero)}\Bigr]\frac{dt}{t}\nonumber\\
    &= \Bigl(1-\int_0^{k_1}\varphi(s)\,ds\Bigr)\frac{k_1}{t}\frac{dt}{t}
    = \frac{k_1^2}{t^2}\,dt.
\end{align}
The final equality uses \ref{eq:phi-after-k1-property4} with $t=k_1$. Combining the two cases yields:
\begin{equation}
\label{eq:h-from-phi-y}\tag{\texttt{h-function}}
     h(t) = \begin{cases}
     \varphi(t) & t \leq k_1,\\
     \displaystyle \frac{k_1^2\,\varphi(k_1)}{t^2}  & t> k_1.
   \end{cases}
\end{equation}
\end{itemize}
Before we verify that these functions $h$ and $g$ form a feasible solution for \ref{integrals}, we start by noticing that for all $t\in(0,k_1]$ we have $h(s)=\varphi(s)$ by construction. Moreover, for any $t\in[k_1,1]$,
\begin{align*}
h(t) + \int_0^t g(s,t)ds=\frac{k_1^2\varphi(k_1)}{t^2}+\int_0^{y(t)}\frac{\varphi(s)\tau(s)}{t^2}ds=\varphi(t)~,
\end{align*}
where the last equality holds because of \ref{eq:property6}. Therefore:
\begin{equation}
\label{usefulequation}
\varphi(t)=h(t) + \int_0^t g(s,t)ds,~~~\forall t\in(0,1]
\end{equation}
Furthermore, by integrating $g(s,\cdot)$ from $s$ to $t$, we have $\int_{s}^t g(s,r)dr=0$ when $t \leq \tau(s)$. For $t>\tau(s)$,
\begin{align*}
    \int_s^t g(s, r)dr = \int_{\tau(s)}^t \frac{\varphi(s) \tau (s)}{r^2}dr = \varphi(s) \tau (s)\left(\frac{1}{\tau(s)}-\frac{1}{t}\right) = \varphi(s) \left(1-\frac{\tau(s)}{t}\right).
\end{align*}
and therefore we conclude that:
\begin{align}\label{usefulequation2}
& \int_s^t g(s, r)dr = 
\begin{cases} 
      0 & t \leq \tau(s)\\
     \displaystyle\varphi(s)\left(1 - \frac{\tau(s)}{t}\right) & t> \tau(s) 
   \end{cases}~~~, 
   && \forall 0 < s \leq t \leq 1~.
\end{align}

Now, we will verify that the above solution fulfills each inequality in \ref{integrals} individually. Specifically, we will show \eqref{OIfromzero}, \eqref{OI}, and \eqref{coverage}.

To show \eqref{OIfromzero}, we consider two case. First, notice that the inequality is true for $t \leq k_1$ by \ref{eq:OIfromzero-phi-property3}. Second, the inequality for $t \geq k_1$ follows since the derivative of both sides is equal to $-\frac{\varphi(k_1)k_1^2}{t^2} = -h(t)$ for $t>k_1$, and that at $t=k_1$ both sides are equal because:
$$
k_1h(k_1)=k_1\varphi(k_1)=1-\int_{0}^{k_1}\varphi(s)\,ds=1-\int_{0}^{k_1} h(s)\,ds~,
$$
which holds due to \ref{eq:phi-after-k1-property4} and the fact that $y(k_1)=0$.

To show \eqref{OI}, we also consider two cases. If $t \leq \tau(s)$ the left hand side is $0$ and the right hand side has only positive terms therefore is positive. For $t > \tau(s)$, using \eqref{usefulequation} and \eqref{usefulequation2}, we have:
\begin{align*}
    h(s) + \int_0^s g(r, s)\,dr  - \int_{s}^{t} g(s, r)\,dr &= \varphi(s) - \varphi(s)\left(1-\frac{\tau(s)}{t}\right) = t \frac{\varphi(s)\tau(s)}{t^2} = tg(s,t) ,
\end{align*} and therefore the two sides of \eqref{OI} are indeed equal.

Finally, to show \eqref{coverage}, by using \eqref{usefulequation} and \eqref{usefulequation2} we have:
\begin{align*}
    h(s) + \int_0^s g(r, s)dr  - (1+f) \int_{s}^{1} g(s, r)dr &= \varphi(s) - (1+f)\varphi(s)\left(1-\tau(s)\right) \\
    &= (1+f)\varphi(s)\left(\tau(s)-c_f\right) = \Theta=\frac{1}{2-y(1)}~,
\end{align*}
where in the last two equalities we use \ref{eq:coverage-property5}. This completes the proof.\qed

\subsection{Proof of \Cref{lem:indifference}}
\label{apx:proof:lem-indifference}
Let us use the notation $s_i = \prod_{j=i+1}^n(1-p_i)$. We will prove by induction the stronger statement that $\Phi_i(x) = \Phi_i(0) + (s_i(1+f)-f)x$ for $x \in [0, x_i]$. Clearly the statement is true for $i = n$. Now suppose it is true for $i = t \geq 2$, and let us prove the statement for $i = t-1$.
We know from \eqref{eq:recurr} that:
\[\Phi_{t-1}(x) = \max \{ \Phi_t(x), (1-p_t) \Phi_t(x) + p_t (\Phi_t(x_t) - f x) \}.\]
Now we already know by the induction hypothesis that $\Phi_t(x)$ is linear when $x \in [0, x_t]$, and the two terms in the maximum are equal when $x = x_{t-1}$. We conclude that:
\[\Phi_{t-1}(x) = (1-p_t) \Phi_t(x) + p_t (\Phi_t(x_t) - f x) \text{ for } x \in [0, x_{t-1}].\]
Hence, $\Phi_{t-1}(x)$ is linear in the required range, and its slope is:
\[(1-p_t)(s_t(1+f)-f)-p_tf = s_{t-1}(1+f)-f,\]
as desired. Finally, plugging the fact that $\Phi_i(x_{i-1}) = \Phi_i(x_i) - f x_{i-1}$ for $i \in [1, n]$ into the Bellman recurrence \eqref{eq:recurr}, we conclude 
\[\Phi_{i-1}(x_{i-1}) = \max \{ \Phi_i(x_{i-1}), (1-p_i) \Phi_i(x_{i-1}) + p_i (\Phi_i(x_i) - f x_{i-1}) \}
= \Phi_i(x_i) - f x_{i-1},\]
and an easy induction gives that $\Phi_0(0) = \Phi_{n}(x_{n}) -f\sum_{i=1}^{n-1}x_i = x_{n} -f\sum_{i=1}^{n-1}x_i,$ as needed.\qed

\section{Omitted Technical Details \revcolor{of \Cref{sec:algorithm}}}\label{sec:technical_details}

\subsection{Scale-invariant Poisson Point Process}
\label{sec:apx-PPP}

A Poisson point process (PPP) $\eta$ on an interval $S \subseteq \mathbb R$ with intensity function $\lambda$ is a random set of countably many points in $S$ such that:
\begin{enumerate}
    \item For any Borel-measurable set $I \subset S$, $|\eta \cap I|$ is a Poisson random variable with mean $\Lambda(I) = \int_I \lambda(t) dt$. $\Lambda$ is referred to as the intensity measure, or the mean measure.
    \item For any finite collection of disjoint Borel-measurable sets $I_1, I_2, \ldots, I_k$ each contained in $S$, the random variables $|\eta \cap I_j|$ are independent of each other.
\end{enumerate}

The theory of Poisson processes tells us that such an $\eta$ exists, and is unique, as long as $\Lambda$ is an $s$-finite measure. For a full discussion, see \citet{penrosepoisson}. 

Since we aim to formalize the continuous limit in which a random variable at location $t$ is active with probability $\frac{dt}{t}$, we will only be interested in the case where $\lambda(t) = \frac{1}{t}$ on the interval $(0, 1]$ and $\Lambda((t, 1]) = -\log t$, the so-called scale-invariant Poisson point process. In this case, we can provide an explicit construction for $\eta$. Define independent random variables $Y_n$ where $Y_n$ has the same distribution as $U_1U_2 \cdots U_{n}$ where the $U_i$ are independent, uniform random variables on $[0,1]$. Then, if we define
$Q_i = 1 - \sum_{j=1}^i Y_j$, 
\[\eta = \{Q_i\}_{i=1}^{\infty}\]
is a PPP with intensity function $\lambda(t) = \frac{1}{t}$ \revcolor{\citep[][Section~6]{arratiascaleinvariant}}. Note that $Q_i$ is the $i$th largest point in $\eta$. Furthermore, $Y_1 = 1 - Q_1$, and $Y_i = Q_{i-1}-Q_i$, so the $Y_i$ are independent spacings between the $Q_i$. 

\subsection{Running time of \cref{alg:order-oblivious}}
\label{sec:running-time}

 \cref{alg:order-oblivious} can be implemented efficiently in polynomial time. First of all, we would like to understand the number of sampled values generated by the algorithm because generating these values is one of the primary processing steps the algorithm performs, other than simple comparisons between the various sampled values and the $X_i$ values it observes. 
 
 Note that the algorithm never needs to actually generate any values below its initial threshold $T$ because the algorithm ignores them anyway. Moreover, the algorithm only needs to generate the smallest value between $T$ and $F^{-1}(k_1)$, since this is the only one that could potentially influence the algorithm's behavior.

The number of quantiles above $k_1$ generated during any run of the algorithm is in fact upper bounded by a Poisson random variable, which is concentrated about the mean number of samples generated. The mean number of samples generated by the sampler above $k_1$ is at most $-\log k_1,$ since the intensity measure of the Poisson process the algorithm generates is $-\log F(t)$. Thus, the number of samples generated by the sampler, overall, is finite.

Moreover, we can sample a Poisson process on an interval $(x', x]$ with intensity measure $-\log F$ (and $x' >0$) efficiently using the following mapping based approach  \revcolor{\citep[][Section~2.3]{kingman-poisson-processes}}: sample a scale-invariant Poisson process with intensity function $\frac{dt}{t}$ and for each point $z$ in this process, return $F^{-1}(z)$. Note that simulating a scale-invariant Poisson process is straightforward: the probability density function of the next active quantile after a number $\tau_0$ is $\frac{\tau_0}{t^2}$ (since the probability that a scale-invariant PPP generates no points in $(\tau_0, t)$ is exactly $\frac{\tau_0}{t}$).

One other concern is the necessity of knowing $\tau$ in order for the algorithm to make decisions. But $\tau$ itself can be computed by solving \ref{difEq} as described in \cref{sec:diff-eq-soln}. Note that many potential values of $\theta$ might need to be tried to find $\tau$, but a simple binary search quickly finds the approximately optimal $\theta$. It follows that with appropriate discretization if necessary, the algorithm can be implemented efficiently. Note that for any particular fixed $f$, this last step is merely preprocessing, and does not  not have to be repeated for every fresh instance of the problem.

\section{Small buyback regime}
\label{sec:smallbuybackregime}
We would like to investigate the setting with $f\to 0$, and understand the asymptotic behavior of $\crf$.

\subsection{Construction of a multistage counterexample for small buyback factor}
\label{sec:multistage}

In this section, we are interested in the behavior of the competitive ratio for $f \to 0$. 
We present a counterexample which proves the following.

\begin{theorem}
\label{thm:multistage-example}
Given a buyback factor $0<f<\frac{1}{256}$, there is no algorithm for the buyback problem which achieves a competitive ratio better than $1 - \frac{1}{2} f \log_2 \frac{1}{256f}$.  
\end{theorem}

This demonstrates an interesting phenomenon: Recall that for $f=0$, we can achieve a competitive ratio of $1$, simply by performing buyback whenever possible. For small $f>0$, we might aim to achieve a competitive ratio of $1 - c f$ for some fixed constant $c>0$. However, the counterexample shows that this is impossible, and the optimal competitive ratio $\alpha(f)$ as a function of the buyback factor $f$ satisfies  $\alpha'(0) = -\infty$.
As we prove in Theorem~\ref{thm:small-factor}, the behavior of the optimal competitive ratio for small $f>0$ is indeed $1 - \Theta(f \log \frac{1}{f})$.

\begin{proofof}{Proof}
Let us construct an instance with $X_i = a_i$ with probability $\frac{1}{2}$, and $0$ otherwise, $n = \lceil \log_2 \frac{1}{f} \rceil$, and $\Phi_t(a_{t-1}) = \Phi_t(a_t) - f a_{t-1},$ by choosing $a_n = 1$ and let $a_{n-k-1} = (1 - \frac{f}{1+f} 2^k) a_{n-k}$ for $k \in [0, n-2]$. \cref{lem:indifference} immediately tells us that the performance of the optimal algorithm is $\textsc{ALG} = a_n - f \sum_{k=1}^{n-1}a_{n-k} =  1-  f \sum_{k=1}^{n-1}a_{n-k}$. We observe that for $k \in [1, n-1]$: $$ a_{n-k} = \left(1 - \frac{f}{1+f}\right) \left(1 - \frac{2f}{1+f}\right) \cdots\left(1 - \frac{2^{k-1} f}{1+f}\right) \geq 1 - \frac{2^k f}{1+f}.$$
Hence, recalling that $n = \lceil \log_2 \frac{1}{f} \rceil$, we have:
$$\textsc{ALG} \leq 1 - f \sum_{k=1}^{n-1} \left(1 - \frac{2^k f}{1+f} \right)
\leq 1- f \left( n - 1 - \frac{2^n f}{1+f} \right) \leq 1 - f \left(\log_2 \frac{1}{f} - 3 \right) =  1 + 3f - f\log_2 \frac{1}{f}.$$
On the other hand, we know that:
\begin{align*}
\textsc{OPT} = \sum_{k=0}^{n-1} a_{n-k} \Pr[X_{\max} = a_{n-k}] 
= \sum_{k=0}^{n-1} \frac{1}{2^{k+1}} a_{n-k} &\geq \frac{1}{2} + \sum_{k=1}^{n-1} \frac{1}{2^{k+1}} \left(1 - \frac{2^k f}{1+f}\right) = 1 - \frac{1}{2^n}- \frac{(n-1)f}{2(1+f)}.
\end{align*}
Hence, $\text{OPT} \geq 1 - f - \frac{1}{2}f\log_2 \frac{1}{f}.$
Therefore, $\textsc{ALG} \leq \textsc{OPT} + 4f - \frac12 f \log_2 \frac{1}{f}$, and hence $\textsc{ALG} < \left(1 - \frac12 f \log_2 \frac{1}{256f}\right) \textsc{OPT}$, considering that $\textsc{OPT} < 1$.
\end{proofof}

In the next section, we show how a simple algorithm achieves the asymptotic optimal competitive ratio of $1-\Theta\left(f\log\tfrac{1}{f}\right)$, establishing that the bound in \Cref{thm:multistage-example} is (almost) tight.

\subsection{Threshold based algorithm}
\label{sec:thresholdalgorithm}
In this section, we present a simple algorithm which achieves a competitive ratio of $1 - \Theta(f \log \frac{1}{f})$ for small $f>0$, matching the example presented in Section~\ref{sec:multistage} up to constant factors. This algorithm also gives 
a constant improvement over $1/2$, for any fixed buyback factor $f > 0$. The algorithm is a simple extension of the standard threshold algorithm: After setting an initial threshold, and accepting some variable above the threshold, we buy back and replace the current variable if and only if this brings positive profit.

\paragraph{The algorithm.}
\begin{itemize}
\item Set an initial threshold $T$ (to be determined).
\item If a variable $X_i$ arrives and $X_i \geq T$, then select $X_i$.
\item Given a currently selected value $X_i$, if a variable $X_j$ arrives and $X_j > (1+f) X_i$, then discard $X_i$ and select $X_j$.
\end{itemize}

An advantage of this algorithm is that is order oblivious, and as will see $T$ will depend only on the distribution of $X_{\textrm{max}}$.

\subsubsection{General analysis approach}

Let us first develop some general formulae to analyze this algorithm.

\begin{lemma}
Suppose that the algorithm selects at least $k$ variables and the $k$-th selected variable is $X^{(k)}$. Then the profit of the algorithm is at least $X^{(k)} / (1+f)^{k-1}$.
\end{lemma}

\begin{proofof}{Proof}
By induction: If the first selected variable is $X^{(0)}$, then the profit at that point is clearly $X^{(0)}$. Observe also that profit never goes down after future buybacks, so the profit at the end is at least $X^{(0)}$.

Assuming that the $k$-th selected variable is $X^{(k)}$ and the variable selected just before that was $X^{(k-1)}$, by induction our profit after selecting $X^{(k-1)}$ was $p_{k-1} \geq X^{(k-1)} / (1+f)^{k-2}$. After the last buyback, since we select $X^{(k)}$ only if $X^{(k)} > (1+f) X^{(k-1)}$, and we pay $f \cdot X^{(k-1)}$ for discarding $X^{(k-1)}$, our profit is 
$$ p_k = p_{k-1} + (X^{(k)} - (1+f) X^{(k-1)}) \geq \frac{X^{(k-1)}}{(1+f)^{k-2}} + \frac{1}{(1+f)^{k-1}} (X^{(k)} - (1+f) X^{(k-1)}) = \frac{X^{(k)}}{(1+f)^{k-1}}, $$ as desired.
\end{proofof}
Next, consider the classical analysis of the prophet inequality, which uses the fact that our algorithm achieves expected value at least $T \cdot \Pr[X_{\max} \geq T] + \EX[(X_{\max}-T)_+] \cdot \Pr[X_{\max} < T]$. We claim the following extension of this bound.

\begin{lemma}
\label{lemma:buyback-bound}
The expected profit of our algorithm is
$$ \EX[ALG] \geq T \cdot \Pr[X_{\max} \geq T] + 
 \EX\left[ \left( \frac{X_{\max}}{(1+f)^{S'}} - T \right)_+ \right],$$
 where $S' = \sum_{j} \indic(X_j' \geq T)$ and $X_j'$ is an independent copy of $X_j$.
\end{lemma}

\begin{proofof}{Proof}
    First, let us note that if the maximum is $X_i$, and the algorithm makes $\theta_i$ picks before $X_i$, \cref{lemma:buyback-bound} implies that in the case that the maximum is above $T$, the algorithm obtains at least
    $$\max \left(T, \frac{X_i}{(1+f)^{\theta_i}}\right).$$
It follows that 
\begin{eqnarray*}
 \EX[ALG] & \geq & T \cdot \Pr[X_{\max} \geq T] + 
\sum_{i=1}^{n} \EX\left[ \left( \frac{X_{i}}{(1+f)^{\theta_i}} - T \right)_+ \indic(X_i = X_{\max}) \right] \\
& \geq & T \cdot \Pr[X_{\max} \geq T] + \sum_{i=1}^{n} \EX\left[ \left( \frac{X_{i}}{(1+f)^{\sum_{j \neq i} \indic(X_j \geq T)}} - T \right)_+ \indic(X_i = X_{\max}) \right],
\end{eqnarray*}
since clearly, $\theta_i \leq \sum_{j \neq i} \indic(X_j \geq T)$.
Observe that conditioned on $X_i$, $\left( \frac{X_{i}}{(1+f)^{\sum_{j \neq i} \indic(X_j \geq T)}} - T \right)_+$ is a decreasing function of $(X_j: j \neq i)$, and $\indic(X_i = X_{\max})$ is also a decreasing function of $(X_j: j \neq i)$. So we can apply the FKG inequality on the product space of $(X_j: j \neq i)$, for every fixed value of $X_i$, and replace the variables $(X_j: j \neq i)$ in the first expression by independent copies $(X'_j: j \neq i)$. Using the notation $S' = \sum_{j} \indic(X_j' \geq T)$, we obtain
\begin{eqnarray*}
\EX[ALG] 
& \geq & T \cdot \Pr[X_{\max} \geq T] + \sum_{i=1}^{n} \EX\left[ \left( \frac{X_{i}}{(1+f)^{\sum_{j \neq i} \indic(X_j' \geq T)}} - T \right)_+ \indic(X_i = X_{\max}) \right] \\
& \geq & T \cdot \Pr[X_{\max} \geq T] + \sum_{i=1}^{n} \EX\left[ \left( \frac{X_{i}}{(1+f)^{S'}} - T \right)_+ \indic(X_i = X_{\max}) \right] \\
 & = & T \cdot \Pr[X_{\max} \geq T] + 
\EX\left[ \left( \frac{X_{\max}}{(1+f)^{S'}} - T \right)_+ \right].
\end{eqnarray*}
\end{proofof}

We will also use the following useful comparison between summations of independent Bernoulli variables and a Poisson variable. 

\begin{lemma} 
\label{lemma:Poisson-compare}
If $S' = \sum_{j=1}^{n} \indic(X_j' \geq T)$, for independent random variables $X'_1,\ldots,X'_n$, and $\Pr[\forall j; X'_j < T] = e^{-\lambda}$, then $S'$ is stochastically dominated by a Poisson random variable with mean $\lambda$. In particular, for any $\ell \geq 1$,
$$ \Pr[S' \leq \ell] \geq e^{-\lambda} \sum_{k=0}^{\ell} \frac{\lambda^k}{k!}.$$
\end{lemma}

\begin{proofof}{Proof}
    Let $p_i = \Pr[X_i' \geq T]$, and let $Y_i$ be a Poisson random variable with mean $\lambda_i = -\log(1-p_i)$, where $Y_i = 0$ whenever $X_i'<T$. (Note that $\Pr[Y_i=0] = e^{-\lambda_i} = 1-p_i = \Pr[X'_i<T]$.) Hence, $Y_i \geq \indic(X'_i \geq T)$, and $S' = \sum_{j=1}^{n} \indic(X_j' \geq T) \leq \sum_{j=1}^{n} Y_j$  with probability $1$. Finally, note that $Y = \sum_{j=1}^{n} Y_j$ is also a Poisson random variable, with mean $\lambda = -\sum_{j=1}^{n} \log (1-p_j) = -\log \prod_{j=1}^{n} (1-p_j) = -\log \Pr[\forall j; X'_j < T]$.
\end{proofof}

\subsubsection{Analysis for small buyback factors}

Consider now the case of small $f$ (e.g. $f \in (0,1/2)$).
We can use Lemmas~\ref{lemma:buyback-bound} and \ref{lemma:Poisson-compare} to get a good bound on the algorithm's performance in this regime. We note that this bound beats the optimal randomized algorithm for the buyback problem with no information, for all $f$ \citep{ashwinkumar2009randomized}.\footnote{The bound presented here is worse than the classical prophet inequality when $f \to \infty$ (in which case our bound tends to $1/3$) as we are aiming primarily to match the asymptotic upper bound of $1-\Theta\left(f\log(\tfrac{1}{f})\right)$  we obtained earlier.} 

\begin{theorem}
\label{thm:small-factor}
There is an algorithm which achieves a competitive ratio of at least 
$$\max_{0 \leq x \leq 1}\frac{(1-x)\cdot x^{\frac{f}{1+f}}}{1-x+x^{\frac{2f+1}{f+1}}} \geq \frac{1}{\frac{f}{1+f} + \left(2 + \frac{1}{f}\right)^{\frac{f}{1+f}}}.$$
\end{theorem}

We remark that the last expression in the theorem follows from plugging in $x = \frac{f}{1+2f}$. For $f \in (0,\frac12)$, it is easy to see that this is $1-\Theta\left(f\log(\tfrac{1}{f})\right)$, indeed,  $(2 + \frac{1}{f})^{\frac{f}{1+f}} < (\frac{2}{f})^f = e^{f \ln \frac{2}{f}} < 1 + 2f \ln \frac{2}{f}$, so the competitive ratio is at least $\frac{1}{\frac{f}{1+f} + 1 + 2f \ln \frac{2}{f}} > \frac{1}{1 + 2f \ln \frac{4}{f}}> 1 - 2f \ln \frac{4}{f}$.

\begin{proofof}{Proof}
From Lemmas~\ref{lemma:buyback-bound} and \ref{lemma:Poisson-compare}, we get:
If $Z$ is a Poisson random variable with mean $\lambda$ (with $e^{-\lambda} = \Pr[X_{\max}< T]$), and $0\leq A \leq 1$ is any random variable (possibly correlated with $X_1,\ldots,X_n$), then 
\begin{eqnarray*}
\EX[ALG]  & \geq & T \cdot \Pr[X_{\max} \geq T] + 
\EX\left[ \left( \frac{X_{\max}}{(1+f)^{S'}} - T \right)_+ \right] \\
& \geq & T \cdot \Pr[X_{\max} \geq T] + 
\EX\left[  A \cdot\left( \frac{X_{\max}}{(1+f)^{Z}} - T \right) \right] \\
& = & T \cdot (\Pr[X_{\max} \geq T] -\EX[A]) + 
\EX\left[  A \cdot\frac{X_{\max}}{(1+f)^{Z}} \right]. 
\end{eqnarray*}
Recall that $\Pr[X_{\max} \geq T] = 1 - e^{-\lambda}$. Now, take $A = c \cdot \indic(X_{\max} \geq T)$ to get that 
\begin{eqnarray*}
\EX[ALG]  & \geq & T \cdot (\Pr[X_{\max} \geq T] -\EX[A]) + 
\EX\left[  A \cdot\frac{X_{\max}}{(1+f)^{Z}} \right] \\
& = & T (1 - e^{-\lambda})(1-c) + 
c \cdot \EX\left[\frac{1}{(1+f)^{Z}} \right]\cdot \EX[X_{\max}\indic(X_{\max} \geq T)] \\
& = & T(1 - e^{-\lambda})(1-c) + 
c e^{\frac{-\lambda f}{1+f}}\cdot \EX[X_{\max}\indic(X_{\max} \geq T)],
\end{eqnarray*}
where $\EX[\frac{1}{(1+f)^Z}] = e^{\frac{-\lambda f}{1+f}}$ follows from an elementary computation for the Poisson random variable $Z$.
Since 
$$\EX[X_{\max}] \leq T \cdot \Pr[X_{\max} < T] + \EX[X_{\max}\indic(X_{\max} \geq T)]
= T e^{-\lambda} + \EX[X_{\max}\indic(X_{\max} \geq T)],$$
setting 
$$ (1 - e^{-\lambda})(1-c) = ce^{\frac{-\lambda f}{1+f}} e^{-\lambda}$$
guarantees a competitive ratio of 
$ce^{\frac{-\lambda f}{1+f}}.$
Letting $x = e^{-\lambda}$, and solving for $c$, we get 
$$c = \frac{1-x}{1-x+x^{\frac{2f+1}{f+1}}},$$
which means we are guaranteed a competitive ratio of 
$$\max_{0 \leq x \leq 1}\frac{(1-x)\cdot x^{\frac{f}{1+f}}}{1-x+x^{\frac{2f+1}{f+1}}},$$
as desired. 
\end{proofof}
\

We remark that with $A = \indic(X_{\max} \geq T) \ \indic(Z \leq 1)$, we could have obtained another, different lower bound on the competitive ratio, better for medium values of $f$. Another reasonable option is $A =\indic(X_{\max} \geq T) c_2^Z$, which can recover approximately the correct result near both $f= 0 $ and $f = \infty$.


\section{Numerical Experiment}
\label{sec:numerical}
In this section, we numerically evaluate the performance of \Cref{alg:order-oblivious} and compare it against several different benchmarks studied in the literature.

\medskip
\noindent\textbf{Settings:}
Each instance of the experiment consists of $n=7$ distributions, where each distribution follows a Generalized Pareto Distribution (GPD). The hyperparameters of these distributions, location $\mu$, scale $\sigma$, and shape $\psi$, are sampled independently at the beginning of each simulation from uniform distributions (to generate enough heterogeneity in our randomized instances). In particular
\begin{align*}
    \mu\sim \textrm{Unif}(0,16),\qquad \sigma\sim \textrm{Unif}(0,4),\qquad \psi\sim \textrm{Unif}(-1,0.5)
\end{align*}
Once the hyperparameters of the GPD distributions are sampled and fixed for a given instance, we generate $500$ independent realizations of the seven random variables to run a Monte Carlo simulation. \Cref{alg:order-oblivious} and benchmarks are then applied to each realization. Since both the sampled values and the decision-making process of the algorithms involve randomness, multiple runs per instance are necessary to obtain robust performance estimates. To evaluate performance, we compute the ratio of the average net reward of each algorithm (the final value minus the buyback cost) to the average of $X_{\max}$ over the $500$ realizations. We then aggregate these results across 100 sampled instances and visualize them using box-and-whisker plots for different values of $f$ in \Cref{fig:numerical}.

\medskip
\noindent\textbf{Benchmarks:} In the graphs, five algorithms are compared, as explained below from left to right:
\begin{itemize}
\item \underline{\textsc{ALG}}: This is basically \Cref{alg:order-oblivious}, the main algorithm of this paper, with the theoretical guarantee establishing that it is optimal competitive (\cref{thm:alg-CR} and \Cref{thm:WorstCase}). \revcolor{We make a slight modification to our algorithm to make it more practical. If the algorithm considers accepting a random variable $X_i$ with probability $p\in(0,1)$, and $X_i$ corresponds to a quantile of $X_\textrm{max}$ larger than $y_f(1)$, we change $p$ to $1$ if $X_i$ is at least $(1+f)$ times the previously accepted and value, and $0$ otherwise. This slight modification to our algorithm obtains a higher reward than our core algorithm in every sample path, and therefore the modified algorithm continues to be an optimal competitive algorithm.}
\item \underline{\textsc{BHK}}: The optimal deterministic algorithm tailored for adversarial arrivals, introduced by \citet{babaioff2008selling}. This algorithm is essentially a greedy algorithm with some margins: it will swap the held value $x$ with the current value $y$ if $y>\gamma\cdot x$, where parameter $\gamma >(1+f)$ is tuned properly as a function of $f$.

\item \underline{\textsc{BK}}: The optimal randomized algorithm for adversarial arrivals, proposed by \citet{ashwinkumar2009randomized}. This algorithm is essentially a greedy algorithm similar to \textsc{BHK} with $\gamma=1$, but with the difference that it first rounds the values down to a specific randomized grid before running the greedy algorithm. \revcolor{In line with the modification to \textsc{ALG}, we make a minor adjustment to this algorithm as well: our version of \textsc{BK} only buys back previously sold items when a new item is being sold.}
    \item \underline{\textsc{Median}}: The classic algorithm for the prophet inequality. This algorithm accepts the first value greater than the median of $X_{\max}$ and never buys back.
    \item \underline{\textsc{Threshold Greedy}}: This algorithm, introduced and analyzed in \Cref{sec:smallbuybackregime}, accepts a value greater than the $\frac{f}{1+2f}$-quantile of $X_{\max}$ for the first time, and switches to new variables if the margin is positive, i.e., if the value is greater than $1+f$ times the previously accepted value.
\end{itemize}

\medskip
\noindent\textbf{Results:}
The results are shown in \Cref{fig:numerical}. As observed, \Cref{alg:order-oblivious} consistently ranks among the best algorithms under random instances, while also enjoying the theoretical guarantee it provides. Notably, the advantage of \Cref{alg:order-oblivious} over algorithms tailored for adversarial arrivals becomes more apparent for larger values of $f$. In this regime, the algorithm performs similarly to the median algorithm, which is optimal for $f=\infty$. On the other hand, for small values of $f$, while the median algorithm performs poorly, \Cref{alg:order-oblivious} significantly outperforms it by efficiently utilizing the buyback option.

\begin{figure}
    \centering    \includegraphics[width=\textwidth,height = 1.35\textwidth]{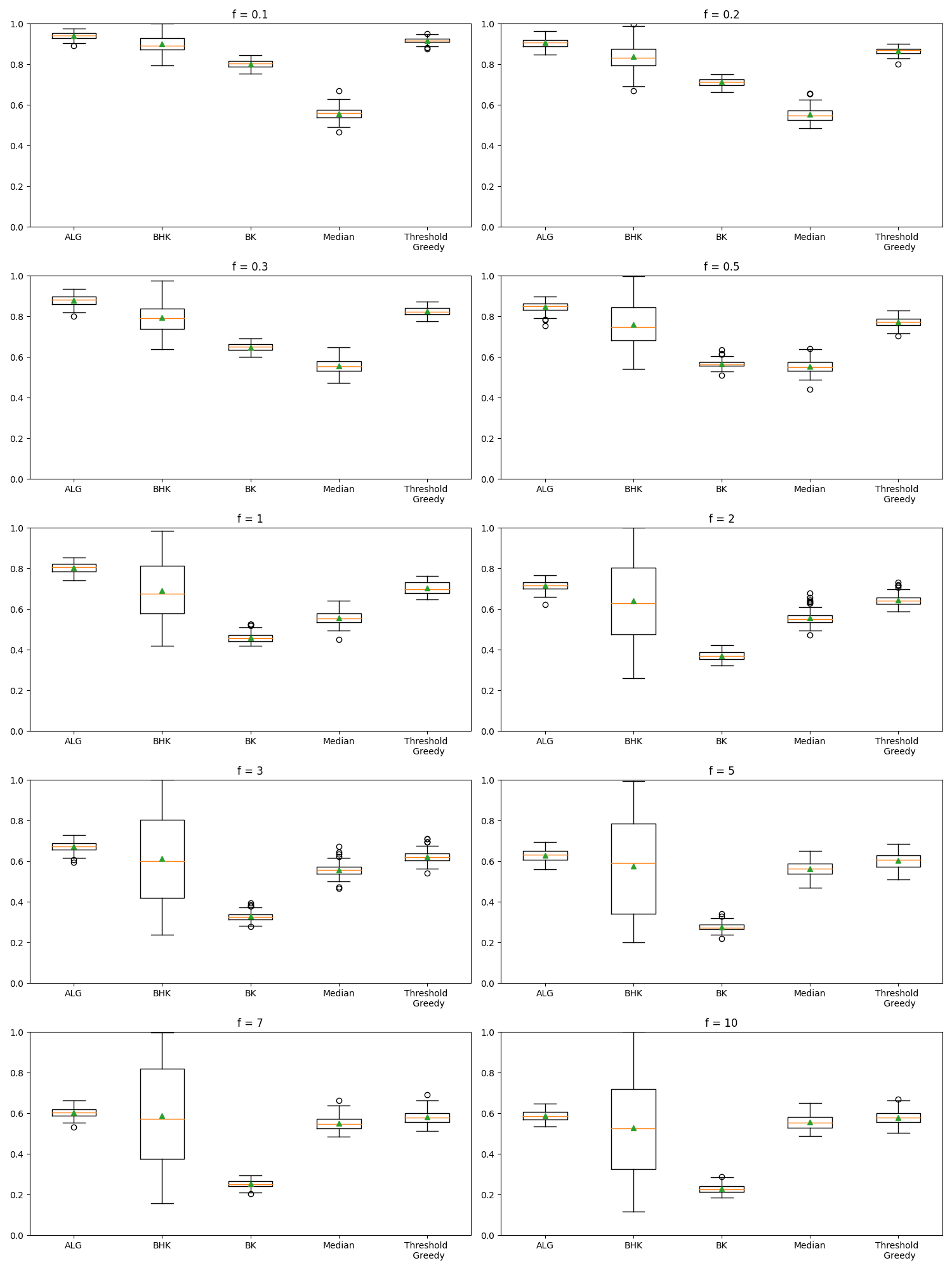}
    \caption{The results of the numerical experiments for different values of $f\in\{0.1,0.2,0.3,0.5,1,2,3,5,7,10\}$. The orange lines specify the median while the green triangles identify the mean performance ratio of the policy across 100 sampled instances.}
    \label{fig:numerical}
\end{figure}

\end{document}